\documentclass{article}
\usepackage[utf8]{inputenc}
\usepackage{amsfonts,amssymb,amsmath,amsthm,amscd,latexsym}
\usepackage{gastex}
\usepackage{graphicx}
\usepackage{xcolor}
\usepackage{float}
\usepackage[ruled, vlined, linesnumbered]{algorithm2e}
\usepackage{subfig}
\usepackage{caption}
\usepackage{multirow}
\usepackage{comment}
\usepackage{array}
\usepackage{fullpage}     
\newcolumntype{C}[1]{>{\centering\arraybackslash}p{#1}}
\def\red{\textcolor{red}}

\def\green{\textcolor{green}}
\def\abs#1{\ensuremath{\lvert #1\rvert}}

\makeatletter
\DeclareRobustCommand\sfrac[1]{\@ifnextchar/{\@sfrac{#1}}%
                                            {\@sfrac{#1}/}}
\def\@sfrac#1/#2{\leavevmode\scalebox{.9}{\kern.1em\raise.5ex
         \hbox{$\m@th\mbox{\fontsize\sf@size\z@
                           \selectfont#1}$}\kern-.1em
         /\kern-.15em\lower.25ex
          \hbox{$\m@th\mbox{\fontsize\sf@size\z@
                            \selectfont#2}$}}}
\DeclareRobustCommand\numfrac[1]{\@ifnextchar/{\@numfrac{#1}}%
                                            {\@numfrac{#1}}}
\def\@numfrac#1{\leavevmode \hbox{$\m@th\mbox{\fontsize\sf@size\z@
                           \selectfont#1}$}}
\makeatother

\newcommand{\Description}[1]{}      

\newcommand{\nat}{\mathbb N}
\newcommand{\real}{\mathbb R}

\newcommand{\tuple}[1]{\langle #1 \rangle}
\newcommand{\dist}{{\mathcal D}}

\let\oldupharpoonright\upharpoonright
\renewcommand{\upharpoonright}{\!\oldupharpoonright\!}

\renewcommand{\u}{{\sf u}}    
\renewcommand{\k}{{\sf k}}    
\renewcommand{\r}{{\sf r}}    

\newcommand{\G}{{\mathcal G}}      
\renewcommand{\H}{{\mathcal H}}      
\newcommand{\M}{{\mathcal M}}      

\newcommand{\C}{{\mathcal C}}      
\renewcommand{\P}{{\mathcal P}}      
\newcommand{\R}{{\sf R}}      
\newcommand{\U}{{\mathcal U}}      
\newcommand{\opnu}{\nu}      

\newcommand{\Supp}{{\sf Supp}}
\newcommand{\Pref}{{\sf Pref}}
\newcommand{\Play}{{\sf Play}}

\newcommand{\straa}{\sigma} \newcommand{\Straa}{\Sigma}
\newcommand{\strab}{\tau} \newcommand{\Strab}{\Theta}
\newcommand{\Prb}{\mathrm{Pr}}
\newcommand{\Cyl}{\mathrm{Cyl}}
\newcommand{\as}{\mathrm{as}}

\newcommand{\Last}{{\sf Last}}
\newcommand{\sink}{{\sf sink}}

\newcommand{\win}[2]{\langle \! \langle 1 \rangle \! \rangle_{\mathit{#2}}^{\mathit{#1}}}
\newcommand{\winsure}[1]{\langle \! \langle 1 \rangle \! \rangle_{\mathit{sure}}^{\mathit{#1}}}
\newcommand{\winas}[1]{\langle \! \langle 1 \rangle \! \rangle_{\mathit{almost}}^{\mathit{#1}}}

\newcommand{\Act}{{\sf A}}

\newcommand{\CPre}{{\sf CPre}}
\newcommand{\APre}{{\sf APre}}
\newcommand{\PosPre}{{\sf PosPre}}
\newcommand{\Attr}{{\sf Attr}}
\newcommand{\PosAttr}{{\sf PosAttr}}

\newcommand{\Reject}{{\sf Reject}}

\newcommand{\expand}{{\rm expand}}

\let\epsilon\varepsilon
\let\emptyset\varnothing

\DeclareMathOperator*{\argmax}{argmax}

\definecolor{gray50}{gray}{0.5}
\def\gray50#1{\textcolor{gray50}{#1}}

\newtheorem{theorem}{Theorem}
\newtheorem{lemma}{Lemma}

\newenvironment{linenomath*}{}{}

\begin{document}

\title{{\bf Stochastic Games with Synchronizing Objectives}}
\date{}
\author{Laurent Doyen\\CNRS \& LMF, ENS Paris-Saclay}

\maketitle
\begin{abstract}
We consider two-player stochastic games played on a finite graph
for infinitely many rounds. Stochastic games generalize both Markov 
decision processes (MDP) by adding an adversary player, and two-player 
deterministic games by adding stochasticity. 
The outcome of the game
is a sequence of distributions over the states of the game graph.
We consider synchronizing objectives, which require the probability
mass to accumulate in a set of target states, either always, once, 
infinitely often, or always after some point in the outcome sequence; 
and the winning modes
of sure winning (if the accumulated probability is equal to~$1$)
and almost-sure winning (if the accumulated probability is arbitrarily close to~$1$).

We present algorithms to compute the set of winning distributions 
for each of these synchronizing modes, 
showing that the corresponding decision problem is PSPACE-complete
for synchronizing once and infinitely often, and PTIME-complete
for synchronizing always and always after some point. 
These bounds are remarkably in line with the special case of MDPs, 
while the algorithmic solution and proof technique are considerably 
more involved, even for deterministic games.
This is because those games have a flavour of imperfect information,
in particular they are not determined and randomized strategies need
to be considered, even if there is no stochastic choice in the game graph. 
Moreover, in combination with stochasticity in the game graph, 
finite-memory strategies are not sufficient in general (for synchronizing 
infinitely often). 


\end{abstract}


\section{Introduction}\label{sec:intro}

Stochastic games are a central model to solve synthesis problems
for reactive systems~\cite{Buc62,Chu63}, which consist of a nonterminating
finite-state program receiving input from an arbitrary, possibly stochastic, 
environment.
The goal of synthesis is to construct a program that satisfies 
with the largest possible probability a given logical specification 
regardless of the behaviour of the environment. 

Synthesis naturally reduces to solving a two-player stochastic game
on a graph, where the logical specification defines the objective 
of the game, as a language of infinite words, representing the set 
of infinite paths through the graph that are winning for one player. 
A wealth of results are known for stochastic games with \emph{perfect information},
where the players are fully informed about the state of the game graph~\cite{CH12},
such as Martin's determinacy result and the existence of pure (non-randomized) 
$\epsilon$-optimal strategies for Borel objectives~\cite{Mar98}, as well as decidability for 
$\omega$-regular objectives, see Chatterjee and Henzinger's survey~\cite{CH12} 
for details and references.

The assumption of perfect information is not realistic in systems consisting
of several components where each component has no access to the internal
state of the other components. Models of games with  \emph{imperfect information}
are notoriously more complicated to solve~\cite{Reif84}, and, combined with the probabilistic and adversarial 
aspects of stochastic games in general lead to undecidability, 
even for the simple class of reachability objectives.
For instance, distributed games are undecidable, even without
stochasticity~\cite{PR90,Schewe14}, and partial-observation games 
are undecidable, even without adversary, for quantitative analysis 
of finitary objectives~\cite{PAZBook,MHC03} and 
for qualitative analysis of infinitary objectives~\cite{BGB12}; 
randomized strategies are more powerful than pure strategies~\cite{CDHR07}, 
and determinacy no longer holds~\cite{BGG17}.


Recent works proposed new decidable models with a flavour of imperfect 
information, for the control of a large population 
of identical processes, modeled as a finite-state machine.
The global state of the game
is a distribution over the local states of the processes,
and the specification describes which sequences of distributions
are winning. The distributions can be discrete~\cite{AAGT12,CFO20} or 
continuous~\cite{KVAK10,AGV18}. The control may be applied uniformly, 
independently of the local state of each process, as in non-deterministic~\cite{BDGG17}, 
and probabilistic automata~\cite{CFO20}, or it may depend on the local history
of states, as in Markov decision processes (MDPs)~\cite{AGV18,DMS19}. 
In both cases imperfect information arises: either
because the control is global, thus not aware of the local state
of individual processes,
or because the control is local, thus not aware of the global states on which 
the specification is defined.

In this paper, we consider the control problem for a continuous population
of processes modeled as a stochastic game with local control,
and objective defined by finitary and infinitary synchronization properties~\cite{BDGG17,CFO20,DMS19}.
Informally, synchronization happens in a sequence of distributions
when (almost) all processes are synchronously in a set of designated 
target states, 
that is when, either for $\epsilon = 0$, or for all $\epsilon > 0$,
there is a distribution in the sequence where 
the probability mass in the target states is at least $1-\epsilon$.
We consider finitary synchronization objectives where synchronization 
should happen once or forever along the sequence of distributions,
called respectively \emph{eventually} and \emph{always} synchronizing;
and infinitary objectives where synchronization 
should happen infinitely often or eventually forever,
called respectively \emph{weakly} and \emph{strongly} synchronizing~\cite{DMS19}.
We distinguish the \emph{sure} winning mode for $\epsilon = 0$, 
and the \emph{almost-sure} winning mode for $\epsilon \to 0$
(where the synchronization objective must be satisfied for all $\epsilon > 0$).

The most interesting and challenging objectives are eventually 
and weakly synchronizing, analogous to reachability and B\"uchi objectives.
For those objectives, it is known that finite memory is not sufficient for almost-sure winning, 
already in MDPs~\cite{DMS19}, and determinacy does not hold. Therefore, 
both the construction of a winning strategy (to show that player~$1$ is almost-sure winning), 
and the construction of a spoiling strategy for the adversary 
(to show that player~$1$ is not almost-sure winning) are non trivial. 
In particular, the traditional approach of constructing a winning strategy
for player~$2$ for the complement of the objective to obtain a spoiling
strategy cannot work. The construction of a spoiling strategy 
must be carried out after fixing an arbitrary infinite-memory 
strategy for player~$1$, which is a substantial complication.
This is the main technical challenge to prove the
correctness of our algorithm. We show that the control problem for
eventually and weakly synchronizing is PSPACE-complete. For always and 
strongly synchronizing, a simple reduction to traditional safety and 
coB\"uchi stochastic games induces a polynomial-time solution.

\smallskip\noindent{\em Applications and Related Works.}
The main interest of this contribution lies in the combination of 
adversarial, stochastic, and infinitary aspects with a flavour of 
imperfect information in a decidable model. 
The works on (discrete) parameterized control considered finitary 
synchronization objectives (reachability of a synchronized distribution), 
either with an adversary~\cite{BDGG17}, or with stochasticity~\cite{CFO20}, 
but not with both. 
With continuous distributions, the central model that has been studied is MDPs,
thus with stochasticity but no adversary, either for finitary~\cite{AGV18} or infinitary
objectives~\cite{AAGT12,DMS19}.
The solution of the control problem studied in this paper is known for MDPs~\cite{DMS19}.

Like in all the above previous works, the main limitation of this population model
is the absence of communication between the processes. While communication plays
a central role in distributed programming applications~\cite{Esp14}, self-organization and 
coordinated behaviour can emerge from large crowds of individuals with limited sensing 
ability, without signaling, and without centralized control~\cite{Cou09,CKFL05}. 
The highly developed local control necessary to achieve a complex collective 
behaviour may emerge naturally~\cite{Isa12} or be engineered~\cite{Mye16}.

The line of work followed in 
this paper can also be viewed as an attempt to propose decidable models that are
still rich enough to describe interesting natural phenomena.
Many systems in natural computing exhibit several instances of the same
anonymous process (without pre-defined identity or hierarchy), 
from particle physics to flock of birds.
Examples of biological systems such as yeast~\cite{BDGG17,AGV18}, 
and simple chemical systems~\cite{KVAK10} illustrate the synthesis applications
of this model. The same principle underlies synthetic biology where
a local control program is executed in every instance of the process~\cite{EKVY00,NDS16}.  
In more complex systems, the computational mechanisms 
behind local decision-making towards global behaviours have multiple origins 
that require more sophisticated computational models~\cite{CKFL05}. 

\section{Definitions}\label{sec:def}


A \emph{probability distribution} on a finite set~$S$ is a
function $d : S \to [0, 1]$ such that $\sum_{s \in S} d(s)= 1$. 
The \emph{support} of~$d$ is the set $\Supp(d) = \{s \in S \mid d(s) > 0\}$. 
We denote by $\dist(S)$ 
the set of all probability distributions 
on~$S$. 

Given a set $T\subseteq S$, let $d(T) = \sum_{s \in T} d(s)$.
For $T \neq \emptyset$, the \emph{uniform distribution} on $T$ assigns probability 
$\frac{1}{\abs{T}}$ to every state in $T$.
Given $s \in S$, we denote by $1_s$ the \emph{Dirac distribution} on~$s$ that 
assigns probability~$1$ to~$s$ (which we often identify with $s$). 

\smallskip\noindent{\em Stochastic games.}
A \emph{two-player stochastic game} (or simply, a game) 
$\G = \tuple{Q, \Act, \delta}$ consists 
of a finite set $Q$ of states, 
a finite nonempty set $\Act$ of actions, 
and a probabilistic transition function $\delta: Q \times \Act \times \Act \to \dist(Q)$.   
We typically denote by $n = \abs{Q}$ the size of the state space, 
and by $\eta$ the smallest positive probability in the transitions of $\G$.

From an initial state $q_0 \in Q$, the game is played in (infinitely many) rounds 
as follows. Each round starts in a state $q_i \in Q$, the first round
starts in the  initial state $q_0$.
In each round, player~$1$ chooses an action $a \in \Act$, 
and then given $a$, player~$2$ chooses an action $b \in \Act$. 
Given the state $q_i$ in which the round started, 
the next round starts in $q_{i+1}$ with probability $\delta(q_i,a,b)(q_{i+1})$.
Note that the game is turn-based as player~$2$ sees 
the action chosen by player~$1$ before playing.

A state $q$ is a \emph{player-$1$ state} if $\delta(q,a,b) = \delta(q,a,b')$ for
all $a,b,b' \in \Act$, and it is a \emph{player-$2$ state} if $\delta(q,a,b) = \delta(q,a',b)$ 
for all $a,a',b \in \Act$. We write $\delta(q,a,-)$ or $\delta(q,-,b)$
to emphasize and recall that $q$ is a player-$1$ or player-$2$ state.
In figures, player-$1$ states are shown as circles,
player-$2$ states as 
boxes (except in \figurename~\ref{fig:pre} where
boxes emphasize the action choices of player~$2$ within a round).
The value of the transition probabilities are not shown on figures,
but diamonds represent the probabilistic choices (the main results of 
this paper are independent of the exact value of transition probabilities).

Classical special cases of stochastic games include Markov decision processes (MDPs),
also called one-player stochastic games, where all states are player-$1$ states;
adversarial MDPs where all states are player-$2$ states; and deterministic games
where $\delta(q,a,b)$ is a Dirac distribution for all $q \in Q$ and all $a,b \in \Act$.
%
Note that it is not important that the action set $\Act$ is the same for both players.
For example, given two nonempty action sets $\Act_1$ and $\Act_2$, we can simulate a transition 
function $\delta_{12}: Q \times \Act_1 \times \Act_2 \to \dist(Q)$, by
defining $\Act = \Act_1 \cup \Act_2$ and, considering fixed actions $a_0 \in \Act_1$, and 
$b_0 \in \Act_2$, by defining $\delta(q,a,b) = \delta_{12}(q,a',b')$ 
where $a' = a$ if $a \in \Act_1$, and $a' = a_0$ otherwise, 
and $b' = b$ if $b \in \Act_2$, and $b' = b_0$ otherwise. 
A \emph{play} in $\G$ is an infinite sequence 
$\pi = q_0 \, a_0b_0 \, q_1 a_1b_1 \, q_2 \ldots \in (QAA)^{\omega}$ 
such that \mbox{$\delta(q_i,a_i,b_i)(q_{i+1}) > 0$} for all $i \geq 0$.
The prefix $q_0 \, a_0b_0 \, q_1 \ldots q_k$ of the play $\pi$ is denoted by $\pi(k)$, 
its length is $\abs{\pi(k)} = k$ and its last element is $\Last(\pi(k)) = q_k$.
The set of all plays in $\G$ is denoted by $\Play(\G)$,
and the set of corresponding finite prefixes (or histories) is denoted by $\Pref(\G)$.

\smallskip\noindent{\em Strategies.}
A \emph{strategy} for player~$1$ in $\G$ is a function 
$\straa: \Pref(\G) \to \dist(\Act)$, and for player~$2$ it is a function
$\strab: \Pref(\G) \times \Act \to \dist(\Act)$.
We denote by $\Straa$, and $\Strab$, the sets of all player-$1$ strategies,
and all player-$2$ strategies, respectively.
A strategy $\straa$ for player~$1$ is \emph{pure} if $\straa(\rho)$
is a Dirac distribution for all $\rho \in \Pref(\G)$; 
it is \emph{counting} if $\abs{\rho} = \abs{\rho'}$ and $\Last(\rho) = \Last(\rho')$ 
implies $\straa(\rho) = \straa(\rho')$ for all $\rho,\rho' \in \Pref(\G)$;
and it is \emph{memoryless} if $\Last(\rho) = \Last(\rho')$ 
implies $\straa(\rho) = \straa(\rho')$ for all $\rho,\rho' \in \Pref(\G)$.
We view deterministic strategies for player~$1$ as functions $\straa: \Pref(\G) \to \Act$,
and counting strategies as functions $\straa: \nat \times Q \to \dist(\Act)$.



A strategy $\straa$ (for player~$1$) uses \emph{finite memory}
if there exists a right congruence $\approx$ of finite index (i.e., that 
can be generated by a finite-state transducer)
over $\Pref(\G)$ 
such that $\rho \approx \rho'$ implies $\straa(\rho) = \straa(\rho')$. 
We omit analogous definitions of pure, counting, memoryless, 
and finite-memory strategies for player~$2$.

\smallskip\noindent{\em State-based objectives.}
The traditional view is to consider the semantics of probabilistic systems
as a probability distribution over sequences (of interleaved states and actions),
i.e., over plays. 



We denote by $\Prb_{d_0}^{\straa,\strab}$ the standard probability measure 
on the sigma-algebra over the set of (infinite) plays, 
generated by the cylinder sets spanned by the (finite) prefixes of plays~\cite{BK08}.
Given a prefix $\rho = q_0 \, a_0b_0 \, q_1 \ldots q_k$, the cylinder set
$\Cyl(\rho) = \{\pi \in \Play(\G) \mid \pi(k) = \rho \}$ has probability:
\begin{linenomath*}
$$\Prb_{d_0}^{\straa,\strab}(\Cyl(\rho)) = d_0(q_0) \cdot  
\prod_{i=0}^{k-1} \straa(\rho(i))(a_i) \cdot \strab(\rho(i),a_i)(b_i) \cdot \delta(q_i,a_i,b_i)(q_{i+1}).$$
\end{linenomath*}
We say that $\rho$ is \emph{compatible} with $\straa$ (from $d_0$) if 
$\Prb_{d_0}^{\straa,\strab}(\Cyl(\rho)) > 0$ for some player-$2$ strategy $\strab$.

State-based objectives, in this traditional semantics, are sets of plays. 
We consider the following state-based objectives, expressed by LTL formulas~\cite{BK08}
where $T \subseteq Q$ is a set of target states:
the reachability and safety objectives $\Diamond T$ and $\Box T$, 
their bounded variants $\Diamond^{=k}\, T$,  $\Diamond^{\leq k}\, T$, and $\Box^{\leq k}\, T$ (where $k \in \nat$), 
and the coB\"uchi objective $\Diamond \Box T$.
As each of the above objectives $\varphi$ is a measurable set,
the probability $\Prb_{d_0}^{\straa,\strab}(\varphi)$ that $\varphi$
is satisfied along a play with initial distribution $d_0$ and strategies $\straa$
for player~$1$ and $\strab$ for player~$2$ is well defined~\cite{Vardi-focs85}.
In particular, we say that player~$1$ is \emph{almost-sure} winning
from an initial distribution $d_0$
for a state-based objective $\varphi$ if he has a strategy to win with 
probability $1$, that is $\exists \straa \in \Straa \cdot \forall \strab \in \Strab:$
$\Prb_{d_0}^{\straa,\strab}(\varphi) = 1$.

\smallskip\noindent{\em Distribution-based objectives.}
An alternative view is to consider probabilistic systems as generators 
of sequences of probability distributions (over states)~\cite{KVAK10}. 
We denote by $\G_{d_0}^{\straa,\strab}$ the \emph{outcome sequence} $d_0, d_1, \ldots$ 
where $d_i \in \dist(Q)$ is, intuitively, the probability distribution 
over states after $i$ rounds defined, for all $q \in Q$, by:
\begin{linenomath*}
$$d_i(q) = \Prb_{d_0}^{\straa,\strab}(\Diamond^{=i} \{q\}) = \sum_{\begin{array}{c} {\scriptstyle \rho\,\in \Pref(\G)}\\[-3pt] {\scriptstyle \abs{\rho} = i} \\[-3pt] {\scriptstyle \Last(\rho) = q}\end{array}} \Prb_{d_0}^{\straa,\strab}(\Cyl(\rho)).$$  
\end{linenomath*}

For a Dirac distribution $d_0 = 1_q$, we often write $\G_{q}^{\straa,\strab}$
instead of $\G_{1_q}^{\straa,\strab}$. We also sometimes omit the subscript 
$d_0$ when the initial distribution is clear from the context.
We denote by $\G_{d_0}^{\straa,\strab}(T)$ the sequence of numbers 
$d_0(T), d_1(T), \ldots$.

%


Distribution-based objectives, in this alternative semantics, are sets of infinite sequences
of distributions over states. In particular, given a set $T \subseteq Q$ of target states, 
\emph{synchronizing objectives} informally require that the probability mass in $T$ 
tends to $1$ (or is equal to $1$) in a sequence $(d_k)_{k \in \nat}$, 
in either all, some, infinitely many, or all but finitely many positions~\cite{DMS19}.
%
%
For $0 \leq \epsilon \leq 1$, we say that a sequence $\bar{d} = d_0 d_1 \dots$ of probability distributions is,
always, eventually, weakly, or strongly $(1-\epsilon)$-synchronizing in $T$ if $d_i(T) \geq 1-\epsilon$,
respectively, for all $i \geq 0$, for some $i \geq 0$, for infinitely many $i$'s, or 
for all but finitely many $i$'s.

For each synchronizing mode $\lambda \in \{always, \penalty-5000 event, \penalty-5000 weakly, \penalty-1000 strongly\}$, 
we consider \emph{winning modes} that require either that $\epsilon$ equals $0$ (sure winning mode), or that 
$\epsilon$ tends to~$0$ (almost-sure winning mode).

We say that player~$1$ is:
\begin{itemize}
\item \emph{sure} winning
for a synchronizing mode $\lambda$ in $T$ from an initial distribution $d_0$  if he has a strategy
to ensure $1$-synchronizing in $T$, or $\exists \straa \in \Straa \cdot \forall \strab \in \Strab:$
$\G_{d_0}^{\straa,\strab}$ is $1$-synchronizing in $T$ in mode $\lambda$.

\item \emph{almost-sure} winning
for a synchronizing mode $\lambda$ in $T$ from an initial distribution $d_0$ if he has a strategy
to ensure $(1-\epsilon)$-synchronizing in $T$ for all $\epsilon>0$, 
or $\exists \straa \in \Straa \cdot \forall \strab \in \Strab \cdot \forall \epsilon > 0:$
$\G_{d_0}^{\straa,\strab}$ is $(1-\epsilon)$-synchronizing in $T$ in mode $\lambda$.
\end{itemize}


We denote by $\winsure{\lambda}(\G,T)$ (or simply $\winsure{\lambda}(T)$ when 
the game $\G$ is clear from the context) the set of distributions $d$ from which 
player~$1$ is sure winning for synchronizing mode $\lambda$ in $T$;
we define analogously the set $\winas{\lambda}(\G,T)$, and
we say that player~$1$ is (sure or almost-sure) winning from $d$, or that
$d$ is (sure or almost-sure) winning. 
If $d \not\in \winas{\lambda}(\G,T)$, we say that player~$2$ can \emph{spoil}
player~$1$ from $d$ for almost-sure synchronizing in mode $\lambda$.

It immediately follows from the definitions that 
for all $\lambda \in \{always, event, weakly, strongly\}$, 
and for all $\mu \in \{sure, \penalty-1000 almost\}$:
\begin{itemize}
\item $\win{always}{\mu}(T) \subseteq \win{strongly}{\mu}(T) \subseteq \win{weakly}{\mu}(T) \subseteq \win{event}{\mu}(T)$, and 
\item $\winsure{\lambda}(T) \subseteq \winas{\lambda}(T)$.
\end{itemize}

In general, these inclusions cannot be strengthened to equality
even for MDPs~\cite{DMS19}, except for always synchronizing 
where we show that $\winsure{always}(T) = \winas{always}(T)$ 
holds in stochastic 
games (Lemma~\ref{lem:always-sure-as} in Section~\ref{sec:other}). 

We are interested in computing the sets $\winsure{\lambda}(\G,T)$ and 
$\winas{\lambda}(\G,T)$ for the four synchronizing modes 
$\lambda \in \{always, event, weakly, strongly\}$,
which we generically call \emph{winning regions}. 
It is sufficient to have an algorithm that computes the set of states $q$ such 
that $1_q$ is winning: to know if a distribution $d$ is winning, consider
a new state $q_d$ with stochastic transitions $\delta(q_d,a,b)(q) = d(q)$ 
for all $q \in Q$, and all $a,b \in \Act$.
We consider the \emph{membership problem}, which is to decide,
given a game $\G$, a set $T$, and a state $q$, whether
$1_q \in \winsure{\lambda}(\G,T)$ (resp., whether $1_q \in \winas{\lambda}(\G,T)$).

\smallskip\noindent{\em Superposition of strategies.}
We use the following property in MDPs: 
given two strategies $\straa_1$, $\straa_2$ in an MDP $\M$, 
for all $0 \leq \alpha \leq 1$,
there exists a strategy $\straa$ such that $\M^{\straa} = 
\alpha \cdot \M^{\straa_1} + (1-\alpha) \cdot \M^{\straa_2}$ (for MDPs, 
we omit the strategy of the second player in all notations).
We say that $\straa$ is obtained by superposition of 
$\straa_1$ and $\straa_2$, and we denote it by 
$\alpha \cdot \straa_1 + (1-\alpha) \cdot \straa_2$.
This notation is slightly misleading, in the sense
that in general $\straa(\rho)(a) \neq 
\alpha \cdot \straa_1(\rho)(a) + (1-\alpha) \cdot \straa_2(\rho)(a)$.
In fact, $\straa(\rho)(a) = \frac{\alpha \cdot \Prb^{\straa_1}(\Cyl(\rho)) \cdot \straa_1(\rho)(a) 
+ (1-\alpha) \cdot \Prb^{\straa_2}(\Cyl(\rho)) \cdot \straa_2(\rho)(a)}{\alpha \cdot \Prb^{\straa_1}(\Cyl(\rho)) 
+ (1-\alpha) \cdot \Prb^{\straa_2}(\Cyl(\rho))}$.
This property will be useful in the (infinite-state) MDPs obtained from 
games after fixing a strategy of player~$1$.

\begin{figure*}[!tb]%
\begin{center}
\hrule
\hspace{10mm}
\subfloat[$q \in \CPre(s)$\qquad{\large \strut}]{
   \begin{picture}(35,35)(0,0)


\gasset{Nh=6, Nw=6, Nmr=3}

\node[Nmarks=n, Nframe=n, Nh=4, Nw=4, Nmr=2](t1)(30,34){}
\node[Nmarks=n, Nframe=n, Nh=4, Nw=4, Nmr=2](t2)(30,27){}
\node[Nmarks=n, Nframe=n, Nh=4, Nw=4, Nmr=2](t3)(30,20){}
\node[Nmarks=n, Nframe=n, Nh=4, Nw=4, Nmr=2](t4)(30,13){}
\node[Nmarks=n, Nframe=n, Nh=4, Nw=4, Nmr=2](t5)(30,6){}
\node[Nmarks=n, ExtNL=y, NLangle=270, NLdist=1, Nh=32, Nw=6, Nmr=3](t3frame)(28.5,20){$s$}

\node[Nmarks=n](s3)(5,20){{\small $q$}}

\node[Nframe=n, Nmarks=n, Nh=2,Nw=2,Nmr=0](a1)(14,26){}
\node[Nmarks=n, Nh=2,Nw=2,Nmr=0](a2)(14.5,20){}
\node[Nframe=n, Nmarks=n, Nh=2,Nw=2,Nmr=0](a3)(14,14){}
\rpnode[Nmarks=n](p1)(22,27)(4,1.4){}
\rpnode[Nmarks=n](p2)(22,20)(4,1.4){}
\rpnode[Nmarks=n](p3)(22,13)(4,1.4){}
\drawarc(22,27,2.5,318.8,41.2)
\drawarc(22,20,2.5,318.8,0)
\drawarc(22,13,2.5,318.8,0)

\drawedge[ELside=r,ELpos=50, ELdist=1, linegray=.6](s3,a1){}
\drawedge[ELside=r,ELpos=50, ELdist=1](s3,a2){}
\drawedge[ELside=r,ELpos=50, ELdist=1, linegray=.6](s3,a3){}
\drawedge[ELside=r,ELpos=50, ELdist=1](a2,p1){}
\drawedge[ELside=r,ELpos=50, ELdist=1](a2,p2){}
\drawedge[ELside=r,ELpos=50, ELdist=1](a2,p3){}
\drawedge[ELside=r,ELpos=50, ELdist=1](p1,t1){}
\drawedge[ELside=r,ELpos=50, ELdist=1](p1,t2){}
\drawedge[ELside=r,ELpos=50, ELdist=1](p1,t3){}
\drawedge[ELside=r,ELpos=50, ELdist=1](p2,t3){}
\drawedge[ELside=r,ELpos=50, ELdist=1](p2,t4){}
\drawedge[ELside=r,ELpos=50, ELdist=1](p3,t4){}
\drawedge[ELside=r,ELpos=50, ELdist=1](p3,t5){}


\end{picture}
   \label{fig:cpre}  
}
\hfill
\subfloat[$q \in \PosPre_1(s)$\quad{\large \strut}]{
   \begin{picture}(35,35)(0,0)


\gasset{Nh=6, Nw=6, Nmr=3}

\node[Nmarks=n, Nframe=n, Nh=4, Nw=4, Nmr=2](t1)(30,34){}
\node[Nmarks=n, Nframe=n, Nh=4, Nw=4, Nmr=2](t2)(30,27){}
\node[Nmarks=n, Nframe=n, Nh=4, Nw=4, Nmr=2](t3)(30,20){}
\node[Nmarks=n, Nframe=n, Nh=4, Nw=4, Nmr=2](t4)(30,13){}
\node[Nmarks=n, Nframe=n, Nh=4, Nw=4, Nmr=2](t5)(30,6){}
\node[Nmarks=n, Nframe=n, Nh=4, Nw=4, Nmr=2](t6)(30,0){}
\node[Nmarks=n, ExtNL=y, NLangle=285, NLdist=1, Nh=20, Nw=6, Nmr=3](t3frame)(28.5,20){$s$}

\node[Nmarks=n](s3)(5,20){{\small $q$}}

\node[Nframe=n, Nmarks=n, Nh=2,Nw=2,Nmr=0](a1)(14,26){}
\node[Nmarks=n, Nh=2,Nw=2,Nmr=0](a2)(14.5,20){}
\node[Nframe=n, Nmarks=n, Nh=2,Nw=2,Nmr=0](a3)(14,14){}
\rpnode[Nmarks=n](p1)(22,27)(4,1.4){}
\rpnode[Nmarks=n](p2)(22,20)(4,1.4){}
\rpnode[Nmarks=n](p3)(22,13)(4,1.4){}
\drawarc(22,27,2.5,318.8,41.2)
\drawarc(22,20,2.5,318.8,0)
\drawarc(22,13,2.5,301.6,0)

\drawedge[ELside=r,ELpos=50, ELdist=1, linegray=.6](s3,a1){}
\drawedge[ELside=r,ELpos=50, ELdist=1](s3,a2){}
\drawedge[ELside=r,ELpos=50, ELdist=1, linegray=.6](s3,a3){}
\drawedge[ELside=r,ELpos=50, ELdist=1](a2,p1){}
\drawedge[ELside=r,ELpos=50, ELdist=1](a2,p2){}
\drawedge[ELside=r,ELpos=50, ELdist=1](a2,p3){}
\drawedge[ELside=r,ELpos=50, ELdist=1](p1,t1){}
\drawedge[ELside=r,ELpos=50, ELdist=1](p1,t2){}
\drawedge[ELside=r,ELpos=50, ELdist=1](p1,t3){}
\drawedge[ELside=r,ELpos=50, ELdist=1](p2,t3){}
\drawedge[ELside=r,ELpos=50, ELdist=1](p2,t4){}
\drawedge[ELside=r,ELpos=50, ELdist=1](p3,t4){}
\drawedge[ELside=r,ELpos=50, ELdist=1](p3,t5){}
\drawedge[ELside=r,ELpos=50, ELdist=1](p3,t6){}


\end{picture}
    \label{fig:pospre1}
}
\hfill
\subfloat[$q \in \PosPre_2(s)$\quad{\large \strut}]{
   \begin{picture}(35,35)(0,0)


\gasset{Nh=6, Nw=6, Nmr=3}

\node[Nmarks=n, Nframe=n, Nh=4, Nw=4, Nmr=2](t1)(30,34){}
\node[Nmarks=n, Nframe=n, Nh=4, Nw=4, Nmr=2](t2)(30,27){}
\node[Nmarks=n, Nframe=n, Nh=4, Nw=4, Nmr=2](t3)(30,20){}
\node[Nmarks=n, Nframe=n, Nh=4, Nw=4, Nmr=2](t4)(30,13){}
\node[Nmarks=n, Nframe=n, Nh=4, Nw=4, Nmr=2](t5)(30,6){}
\node[Nmarks=n, Nframe=n, Nh=4, Nw=4, Nmr=2](t6)(30,0){}
\node[Nmarks=n, ExtNL=y, NLangle=285, NLdist=1, Nh=20, Nw=6, Nmr=3](t3frame)(28.5,20){$s$}

\node[Nmarks=n](s3)(5,20){{\small $q$}}

\node[Nframe=y, Nmarks=n, Nh=2,Nw=2,Nmr=0](a1)(14.5,27){}
\node[Nmarks=n, Nh=2,Nw=2,Nmr=0](a2)(14.5,20){}
\node[Nframe=y, Nmarks=n, Nh=2,Nw=2,Nmr=0](a3)(14.5,13){}
\rpnode[Nmarks=n](p1)(22,27)(4,1.4){}
\rpnode[Nmarks=n](p2)(22,20)(4,1.4){}
\rpnode[Nmarks=n](p3)(22,13)(4,1.4){}
\drawarc(22,27,2.5,318.8,41.2)
\drawarc(22,20,2.5,318.8,0)
\drawarc(22,13,2.5,301.6,0)

\drawedge[ELside=r,ELpos=50, ELdist=1](s3,a1){}
\drawedge[ELside=r,ELpos=50, ELdist=1](s3,a2){}
\drawedge[ELside=r,ELpos=50, ELdist=1](s3,a3){}
\drawedge[ELside=r,ELpos=50, ELdist=1](a1,p1){}
\drawedge[ELside=r,ELpos=50, ELdist=1](a2,p2){}
\drawedge[ELside=r,ELpos=50, ELdist=1](a3,p3){}
\drawedge[ELside=r,ELpos=50, ELdist=1](p1,t1){}
\drawedge[ELside=r,ELpos=50, ELdist=1](p1,t2){}
\drawedge[ELside=r,ELpos=50, ELdist=1](p1,t3){}
\drawedge[ELside=r,ELpos=50, ELdist=1](p2,t3){}
\drawedge[ELside=r,ELpos=50, ELdist=1](p2,t4){}
\drawedge[ELside=r,ELpos=50, ELdist=1](p3,t4){}
\drawedge[ELside=r,ELpos=50, ELdist=1](p3,t5){}
\drawedge[ELside=r,ELpos=50, ELdist=1](p3,t6){}

\rpnode[Nframe=n, Nmarks=n](pghost)(20,30.5)(4,1.4){}
\drawedge[ELside=r,ELpos=50, ELdist=1, linegray=.6](a1,pghost){}
\rpnode[Nframe=n, Nmarks=n](pghost)(20,23.5)(4,1.4){}
\drawedge[ELside=r,ELpos=50, ELdist=1, linegray=.6](a1,pghost){}

\rpnode[Nframe=n, Nmarks=n](pghost)(20,23.5)(4,1.4){}
\drawedge[ELside=r,ELpos=50, ELdist=1, linegray=.6](a2,pghost){}
\rpnode[Nframe=n, Nmarks=n](pghost)(20,16.5)(4,1.4){}
\drawedge[ELside=r,ELpos=50, ELdist=1, linegray=.6](a2,pghost){}

\rpnode[Nframe=n, Nmarks=n](pghost)(20,16.5)(4,1.4){}
\drawedge[ELside=r,ELpos=50, ELdist=1, linegray=.6](a3,pghost){}
\rpnode[Nframe=n, Nmarks=n](pghost)(20,9.5)(4,1.4){}
\drawedge[ELside=r,ELpos=50, ELdist=1, linegray=.6](a3,pghost){}


\end{picture}
    \label{fig:pospre2}
}
\hspace{10mm}
\hrule
\caption{The predecessor operators. \label{fig:pre}}%
\Description{Three game graphs illustrating the definition of the predecessor operators.}
\end{center}
\end{figure*}
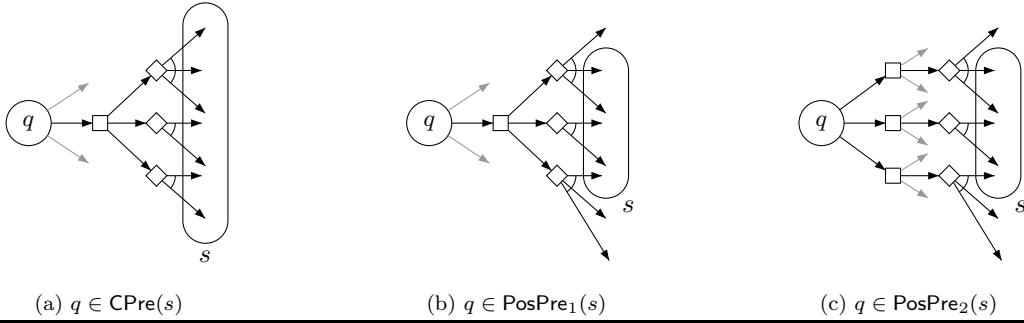

\smallskip\noindent{\em Attractors and subgames.}
Let $\CPre: 2^Q \to 2^Q$ be the \emph{controllable predecessor operator}
defined for all $s \subseteq Q$ by $\CPre(s) = \{q \in Q \mid \exists a \in \Act 
\cdot \forall b \in \Act: \Supp(\delta(q,a,b)) \subseteq s \}$. 
Intuitively, $\CPre(s)$ is the set of states from which player~$1$ can ensure
to be in $s$ after one round, regardless of the action chosen by player~$2$
and of the outcome of the probabilistic  
transition (\figurename~\ref{fig:cpre}).

For a set $T \subseteq Q$, the \emph{attractor} $\Attr(T,\G)$ is the least fixed
point of the operator $x \mapsto \CPre(x) \cup T$,    
that is $\Attr(T,\G) = \bigcup_{i\geq 0} \CPre^i(T)$ (where $\CPre^0(T) = T$). 
It is the set of states in $\G$ from which player~$1$ has a (pure memoryless) strategy 
to ensure eventually reaching $T$~\cite{Thomas97}. We refer to such a 
memoryless strategy as an attractor strategy.

Let $\PosPre_1: 2^Q \to 2^Q$ be the \emph{positive predecessor operator}
for player~$1$
defined for all $s \subseteq Q$ by $\PosPre_1(s) = \{q \in Q \mid \exists a \in \Act 
\cdot \forall b \in \Act: \Supp(\delta(q,a,b)) \cap s \neq \emptyset \}$,
and let $\PosPre_2: 2^Q \to 2^Q$ be defined symmetrically by 
$\PosPre_2(s) = \{q \in Q \mid \forall a \in \Act
\cdot \exists b \in \Act: \Supp(\delta(q,a,b)) \cap s \neq \emptyset \}$
(\figurename~\ref{fig:pospre1} and \figurename~\ref{fig:pospre2}).

For a set $T \subseteq Q$, the \emph{positive attractor} $\PosAttr_i(T,\G)$ 
for player~$i$ ($i=1,2$) is the least fixed point of 
the operator $x \mapsto \PosPre_i(x) \cup T$.   
There exists a pure memoryless strategy for player~$i$ (referred to as \emph{positive-attractor} 
strategy) to ensure, regardless of the strategy for player~\mbox{$3\!-\!i$}
that from all states in $\PosAttr_i(T,\G)$, the set $T$ 
is reached within $n = \abs{Q}$ steps with positive probability (in fact, bounded probability,
at least $\eta^{n}$ where $\eta$ is the smallest positive probability in the transitions of $\G$). 

A set $S \subseteq Q$ induces a \emph{subgame} of $\G$
if for all $q \in S$, there exist $a_q,b_q \in \Act$ such that $\delta(q,a_q,b_q)(S) = 1$.
We denote by $\G \upharpoonright [S] = \tuple{S, \Act, \delta_S}$ the subgame 
induced by $S$, where for all $q \in S$ and $a \in \Act$, 
if $\delta(q,a,b_a)(S) = 1$ for some $b_a \in \Act$ (note that this condition
holds for $a = a_q$), then for all $b \in \Act$
we define $\delta_S(q,a,b) = \delta(q,a,b)$ if $\delta(q,a,b)(S) = 1$,
and $\delta_S(q,a,b) = \delta(q,a,b_a)$ if $\delta(q,a,b)(S) < 1$;
otherwise $\delta(q,a,b)(S) < 1$ for all $b \in \Act$, 
and then we define $\delta_S(q,a,b) = \delta_S(q,a_q,b)$ for all $b \in \Act$.
We use this definition of subgame to keep the same
alphabet of actions in every state.
For instance, the set $S = Q \setminus \PosAttr_i(T,\G)$ (for $i=1,2$) 
induces a subgame of $\G$.

A set $S \subseteq Q$ is a \emph{trap} for player~$1$ in $\G$ 
if for all states $q \in S$ and all actions $a \in \Act$, there exists $b \in \Act$ 
such that $\delta(q,a,b)(S) = 1$. Intuitively, player~$2$ has a strategy to keep
player~$1$ trapped in $S$.
Dually, the set $S$ is a \emph{trap} for player~$2$ in $\G$ 
if for all states $q \in S$, there exists $a \in \Act$ such that for all $b \in \Act$ we have
$\delta(q,a,b)(S) = 1$. Note that $S$ is a trap for player~$2$ if and only 
if $S \subseteq \CPre(S)$. A key property of traps is that if player~$i$ has
no winning strategy for some state-based objective from a state $q$ in a trap $S$ for player~$i$, 
in the subgame $\G \upharpoonright [S]$, then for the same objective in $\G$ 
player~$i$ has no winning strategy from $q$.   

For deterministic games, the operators $\CPre$ and $\PosPre_1$ coincide,
as well as the attractor $\Attr(T,\G)$ and positive attractor $\PosAttr_1(T,\G)$,
and therefore the set $U = Q \setminus \Attr(T,\G)$ induces a subgame of $\G$.

We recall basic properties derived from the definitions of the positive attractor,
and the analysis of stochastic games with almost-sure reachability objective~\cite{AHK07,Chatterjee07}.

\begin{lemma}\label{lem:almost-sure-reach-game-pl1}
If a distribution $d_0$ is almost-sure winning for a reachability
objective $\Diamond T$ in a game $\G$, then there exists a memoryless player-$1$ 
strategy $\straa$ such that for all $\epsilon > 0$, there exists an integer 
$h_{\epsilon}$ such that for all player-$2$ strategies $\straa$, 
$\Prb_{d_0}^{\straa,\strab}(\Diamond^{\leq h_{\epsilon}} T) \geq 1-\epsilon$.
\end{lemma}

\begin{proof}
Consider the memoryless almost-sure winning strategy for player~$1$
that plays according to a positive-attractor strategy to reach $T$
(while staying in the almost-sure winning region). Then within $n$
rounds, the probability to reach $T$ is at least $\eta^n$. Therefore
the probability of not having reached $T$ is a most $1-x$ where $x = \eta^n > 0$, 
and within $k \cdot n$ rounds it is at most $(1-x)^k$.
Since $1-x < 1$, we get $(1-x)^{k_{\epsilon}} \leq \epsilon$ by 
taking $k_{\epsilon}$ sufficiently large. Taking 
$h_{\epsilon} = k_{\epsilon} \cdot n$ concludes the proof.
\end{proof}

\begin{lemma}\label{lem:almost-sure-reach-game}
If a distribution $d_0$ is not almost-sure winning for a reachability
objective $\Diamond T$ in a game $\G$, then there exists a memoryless player-$2$ 
strategy $\strab$ such that for all player-$1$ strategies $\straa$, 
for all $i \geq 0$ 
we have 
$\Prb_{d_0}^{\straa,\strab}(\Diamond^{=i} T) \leq 1 - \eta_0 \cdot \eta^{n}$
where 
$\eta_0 = \min \{d_0(q) \mid q \in \Supp(d_0)\}$ is the smallest 
positive probability in the initial distribution $d_0$.
\end{lemma}

\begin{proof}
The set of almost-sure winning states for the reachability objective 
$\Diamond T$ can be computed by graph-theoretic algorithms~\cite{AHK07}, 
and can be expressed succinctly by the $\mu$-calculus fixpoint formula 
$\varphi_{{\sf AS}} = \opnu Y.\mu X.~\!T \cup \APre(Y,X)$~\cite{ConcOmRegGames} where,
given $X,Y \subseteq Q$:
\begin{linenomath*}
$$\APre(Y,X) = \{q \in Q \mid \exists a \in \Act\cdot \forall b \in \Act:  
\Supp(\delta(q,a,b)) \subseteq Y \land \Supp(\delta(q,a,b)) \cap X \neq \emptyset\}$$
\end{linenomath*}
is the set of states from which there is an action to ensure that all successor 
states are in~$Y$ and that with positive probability the successor state is in $X$.
Note that $\CPre(Y) = \APre(Y,Q)$.

We briefly recall the interpretation of  $\mu$-calculus formulas~\cite{BS07,BW18}
based on Knaster-Tarski theorem.
Given a monotonic function $\psi: 2^Q \to 2^Q$
(i.e., such that $X \subseteq Y$ implies $\psi(X) \subseteq \psi(Y)$),
the expression $\opnu Y. \psi(Y)$ is the (unique) greatest fixpoint of $\psi$,
which can be computed as the limit of the sequence $(Y_i)_{i \in \nat}$ defined by
$Y_0 = Q$, and $Y_{i} = \psi(Y_{i-1})$ for all $i \geq 1$.
Dually, the expression $\mu X.~\psi(X)$ is the (unique) least fixpoint of $\psi$,
and the limit of the sequence $(X_i)_{i \in \nat}$ defined by
$X_0 = \emptyset$, and $X_{i} = \psi(X_{i-1})$ for all $i \geq 1$.
If $\abs{Q} = n$, then it is not difficult to see that the 
limit of those sequences is reached after at most $n$ iterations, 
$X_{n} = X_{n+1}$ and $Y_{n} = Y_{n+1}$.

Intuitively, the formula $\varphi_{{\sf AS}}$ computes the largest set $S$ of states
such that every state $q \in S$ has a strategy to ensure reaching $T$ with 
positive probability, while at the same time staying in $S$ with probability~$1$.
It follows that $S$ is the set of all states from which there exists a strategy 
to reach $T$ with probability~$1$.

Now consider the states $q_0 \not\in \varphi_{{\sf AS}}$ that are not almost-sure 
winning, and let the rank of $q_0$ be the least integer $i$ such that
$q_0 \in Y_{i}$ and $q_0 \not\in Y_{i+1}$.
Let $\Omega_{i} = Y_i \setminus Y_{i+1}$ for all $i \geq 0$ 
be the set of states of rank~$i$, and for $\mbox{$\sim$}\!\in\!\{<,\leq\}$
let $\Omega_{\sim i} = \bigcup_{j\sim i} \Omega_{j}$. 
It is easy to show by induction that $Q \setminus Y_i = \Omega_{<i}$ 
since $\Omega_{<i} = (Y_0 \setminus Y_{1}) \cup (Y_1 \setminus Y_{2}) \cup \dots 
 \cup (Y_{i-1} \setminus Y_{i})$ and $Y_0 = Q$. Moreover since $T \subseteq Y_i$
we have $\Omega_{i} \cap T = \emptyset$.

For $q \in \Omega_{i}$, we show that there exists a player-$2$ strategy $\strab$
such that for all player-$1$ strategies $\straa$ we have 
$\Pr^{\straa,\strab}_{q}(\Box \Omega_i \cup \Diamond \Omega_{<i}) \geq \eta$.
Since $q \not\in Y_{i+1} = \mu X.~T \cup \APre(Y_i,X)$, we have 
$q \not\in \APre(Y_{i},Y_{i+1})$.
Then for all actions $a \in \Act$, there exists $b \in \Act$ such that:
\begin{linenomath*}
$$\text{either } \Supp(\delta(q,a,b)) \not\subseteq Y_{i}, 
\text{ or } \Supp(\delta(q,a,b)) \cap Y_{i+1} = \emptyset,$$
\noindent that is 
$$\text{either } \Supp(\delta(q,a,b)) \cap \Omega_{<i} \neq \emptyset,
\text{ or } \Supp(\delta(q,a,b)) \subseteq \Omega_{\leq i}.$$ 
\end{linenomath*}

Consider the memoryless strategy $\strab$ for player-$2$ that,
given a state $q \in \Omega_{i}$ and action $a$ of player~$1$,
plays such an action $b$.

It follows in both cases
that if not all successors of a state $q_0  \in \Omega_{i}$ on action $a$ are in $\Omega_{i}$,
then at least one of them is in $\Omega_{<i}$, which entails that 
$\Pr^{\straa,\strab}_{q_0}(\Box \Omega_i \cup \Diamond \Omega_{<i}) \geq \eta$
for all player~$1$ strategies $\straa$.
For $i=0$, since $\Omega_{<0} = \emptyset$, we have 
$\Pr^{\straa,\strab}_{q_0}(\Box \Omega_0) = 1$ and thus 
$\Pr^{\straa,\strab}_{q_0}(\Box (Q \setminus T)) = 1$ 
for all $q_0 \in \Omega_{0}$ and all strategies $\straa \in \Straa$.
Inductively, if $\Pr^{\straa,\strab}_{q_0}(\Box (Q \setminus T)) \geq p_i$ for all 
$q_0 \in \Omega_{< i}$ and all $\straa \in \Straa$, then 
$\Pr^{\straa,\strab}_{q_0}(\Box (Q \setminus T)) \geq \eta \cdot p_i$ for all 
$q_0 \in \Omega_{i}$ and all $\straa \in \Straa$.

It follows that if $q_0 \not\in \varphi_{{\sf AS}}$ is not almost-sure winning
for the reachability objective $\Diamond T$, then 
$q_0 \in \Omega_{<n}$ and $\Pr^{\straa,\strab}_{q_0}(\Box (Q \setminus T)) \geq \eta^n$
for all $\straa \in \Straa$. The result of the lemma follows.
\end{proof}

\smallskip\noindent{\em Strongly connected component.}
In a directed graph $\tuple{V,E}$, a \emph{strongly connected component (SCC)}
is a nonempty set $s \subseteq V$ such that for all $v,v' \in s$,
there exists a nonempty path from $v$ to $v'$. Our definition excludes 
singletons $\{v\}$ to be an SCC if there is no self-loop $(v,v)$ in $E$.
The \emph{period} of an SCC is the greatest common divisor of the lengths
of all its cycles. 

\smallskip\noindent{\em Overview of the results.}
We present the solution of the membership problem for sure winning in 
Section~\ref{sec:sure} and for almost-sure winning in Section~\ref{sec:as}.
The most challenging case is almost-sure weakly synchronizing, 
which we solve first for deterministic games (Section~\ref{sec:det-games}), 
and then in the general case of stochastic games (Section~\ref{sec:stoch-games}).
The other synchronizing modes for almost-sure winning are considered in 
Section~\ref{sec:other}. We also present tight complexity bounds and memory 
requirement for all winning modes.

\section{Sure Synchronizing}\label{sec:sure}
In the rest of this paper, when the initial distribution $d_0$
is irrelevant or clear from the context, we denote the $i$-th element $d_i$
in the sequence $\G_{d_0}^{\straa,\strab} = d_0, d_1, \ldots$
by $\G_{i}^{\straa,\strab}$.

For sure winning, only the support of distributions is important (not the
exact value of probabilities). Intuitively for player~$1$, the worst case that can happen
is that player~$2$ uses the \emph{uniform strategy} $\strab_{\u}$ that plays all 
actions uniformly at random, in order to scatter the probability mass
in as many states as possible, where $\strab_{\u}(\rho)(a) = \frac{1}{\abs{\Act}}$
for all $\rho \in \Pref(\G)$ and $a \in \Act$.
Formally, given a set $T \subseteq Q$ and an arbitrary player-$1$ strategy 
$\straa \in \Straa$ it is easy to show that 
$\Supp(\G_i^{\straa,\strab}(T)) \subseteq \Supp(\G_i^{\straa,\strab_{\u}}(T))$,
for all player-$2$ strategies $\strab \in \Strab$ and all $i \geq 0$. 

Therefore, in all synchronizing modes, player~$1$ is sure winning if and only if 
player~$1$ is sure winning against the uniform strategy $\strab_{\u}$
for player~$2$, and computing the sure winning distributions
in stochastic games reduces to the same problem in MDPs (obtained by fixing 
$\strab_{\u}$ in $\G$), which is 
known~\cite{DMS19}. We immediately derive the following results.

\begin{theorem}\label{theo:sure-winning}
The membership problem for sure always and strongly synchronizing can be solved
in polynomial time, and pure memoryless strategies are sufficient for player~$1$.

The membership problem for sure eventually and weakly synchronizing is PSPACE-complete, 
and pure strategies with exponential memory are sufficient (and may be necessary)
for player~$1$.
\end{theorem}

Note that when the uniform strategy $\strab_{\u}$ is fixed, the 
controllable predecessor operator $\CPre$ coincides with the
predecessor operator used to solve the membership problem for 
MDPs~\cite{DMS19}.

Also note that even in the case of deterministic games, the winning regions
for sure and almost-sure winning do not coincide, for eventually and weakly 
synchronizing, as illustrated by the game $\G_{{\sf win}}$ (\figurename~\ref{fig:gwin}).
The state $q_1$ is almost-sure winning (as we show in the beginning of 
Section~\ref{sec:det-games}), but not sure winning
for weakly synchronizing in $T = \{q_1,q_3\}$ (e.g., against the uniform strategy
for player~2).
We show in Section~\ref{sec:other} that in deterministic games, the winning regions
for sure and almost-sure winning do coincide for always and strongly
synchronizing (even for all stochastic games in the case of always synchronizing).

\section{Almost-Sure Synchronizing}\label{sec:as}

We first consider almost-sure weakly synchronizing, which is the most 
interesting and challenging case. 
We present an algorithm to compute the set $\winas{weakly}(T)$ and
we show that pure counting strategies are sufficient for player~$1$.

\subsection{Weakly synchronizing in deterministic games}\label{sec:det-games}

The key ideas of the algorithm are easier to present in the special
case of deterministic games, and with the assumption that pure 
counting strategies are sufficient for player~$1$ (but player~$2$ is 
allowed to use an arbitrary strategy). 
We show in Section~\ref{sec:stoch-games} how to compute the winning region for 
almost-sure weakly synchronizing in general stochastic games, and without any 
assumption on the strategies of player~$1$. It will follow from our results
that pure counting strategies are in fact sufficient for player~$1$.

Given a deterministic game $\G = \tuple{Q, \Act, \delta}$, a \emph{selector} is a function
$\alpha: Q \to \Act$, and for a set $s \subseteq Q$, let $\delta_{\alpha}(s) = 
\{\delta(q,\alpha(q),b) \mid \mbox{$q \in s$} \land \mbox{$b \in \Act$} \} \subseteq Q$.
A pure counting strategy can be viewed as an infinite sequence of selectors.
The \emph{subset construction} for $\G$
is the graph $\P(\G) = \tuple{V,E}$ where $V = 2^Q \setminus \{\emptyset\}$ 
and $E = \{(s,\delta_{\alpha}(s)) \mid s \in V \land \alpha \text{ is a selector} \}$.
Given a set $T\subseteq Q$ of target states, and a set $s \in V$,     
we say that $s$ is \emph{accepting} if $s \subseteq T$ (for singletons $\{q\}$ we 
simply say that $q$ is accepting).

The \emph{central property} of the subset construction $\P(\G)$ is that 
for every sequence of selectors $\alpha_1, \alpha_2, \dots, \alpha_k$,
the sequence $s_1,s_2, \dots, s_{k+1}$ such that $s_{i+1} = \delta_{\alpha_i}(s_i)$
for all $1 \leq i \leq k$, is a path in $\P(\G)$ and that for every
state $q \in s_{k+1}$, there exists a play prefix $q_1 \, a_1b_1 \, q_2 \ldots q_{k+1}$ in $\G$ 
such that $q_{k+1} = q$ and $q_i \in s_i$ for all $1 \leq i \leq k$, 
that is compatible with the given sequence of selectors, $a_i = \alpha_i(q_i)$.
The central property is easily proved by induction on $k$~\cite{CDHR07}.
Using K\"{o}nig's Lemma, the central property holds for infinite sequences,
namely for every infinite path $s_1,s_2, \dots$ in $\P(\G)$, there 
exists an infinite play $q_1 \, a_1b_1 \, q_2 \ldots$ in $\G$ 
such that $q_i \in s_i$ and $a_i = \alpha_i(q_i)$ for all $i \geq 1$.
We also mention a simple monotonicity property: if $s \subseteq s'$, then 
$\delta_{\alpha}(s) \subseteq \delta_{\alpha}(s')$ for all selectors~$\alpha$.
\smallskip

\begin{figure}[!tb]%
\begin{center}
\hrule
\subfloat[$\G_{{\sf win}}${\large \strut}]{
   \begin{picture}(40,42)(0,10)

\node[Nmarks=r, Nmr=0](q1)(10,40){$q_1$}
\node[Nmarks=n](q2)(10,20){$q_2$}
\node[Nmarks=r](q3)(30,20){$q_3$}

\drawedge[ELside=r,ELpos=50, ELdist=1](q1,q2){$b_2$}
\drawedge[ELside=r,ELpos=50, ELdist=1, curvedepth=-6](q2,q3){$a_2$}
\drawedge[ELside=l,ELpos=50, curvedepth=-6](q3,q2){}


\drawloop[ELside=l,loopCW=y, loopangle=90, loopdiam=5](q1){$b_1$}
\drawloop[ELside=l,loopCW=y, loopangle=225, loopdiam=5, ELdist=0.5](q2){$a_1$}

\end{picture}
   \label{fig:gwin}  
}
\quad
\subfloat[Subset construction{\large \strut}]{
   \begin{picture}(60,55)(0,5)

\node[Nmarks=r](q13)(10,30){$q_{13}$}

\node[Nmarks=r](q1)(30,50){$q_{1}$}
\node[Nmarks=n](q12)(30,30){$q_{12}$}
\node[Nmarks=n](q123)(30,10){$q_{123}$}

\node[Nmarks=n](q23)(50,50){$q_{23}$}
\node[Nmarks=n](q2)(50,30){$q_{2}$}
\node[Nmarks=r](q3)(50,10){$q_{3}$}

\drawedge[ELside=r,ELpos=50, ELdist=1](q13,q12){}
\drawedge[ELside=l,ELpos=50, curvedepth=0](q1,q12){}
\drawedge[ELside=l,ELpos=50, curvedepth=5](q12,q123){$a_2$}
\drawedge[ELside=l,ELpos=50, curvedepth=5](q123,q12){$a_1$}
\drawedge[ELside=l,ELpos=50, curvedepth=0](q23,q2){$a_1$}
\drawedge[ELside=l,ELpos=50, curvedepth=5](q2,q3){$a_2$}
\drawedge[ELside=l,ELpos=50, curvedepth=5](q3,q2){}


\drawloop[ELside=l,loopCW=y, loopangle=0, loopdiam=5, ELdist=.5](q2){$a_1$}
\drawloop[ELside=l,loopCW=y, loopangle=0, loopdiam=5, ELdist=.5](q23){$a_2$}
\drawloop[ELside=l,loopCW=y, loopangle=180, loopdiam=5](q123){$a_2$}
\drawloop[ELside=l,loopCW=y, loopangle=45, loopdiam=5, ELdist=0](q12){$a_1$}
\end{picture}
    \label{fig:gwin-subset}
}

\hrule
\caption{The deterministic game $\G_{{\sf win}}$ where player~$1$ is almost-sure
 weakly synchronizing (from all states), and its subset construction. \label{fig:gamewin}}%
\Description{A deterministic game graph and its subset construction.}
\end{center}
\end{figure}
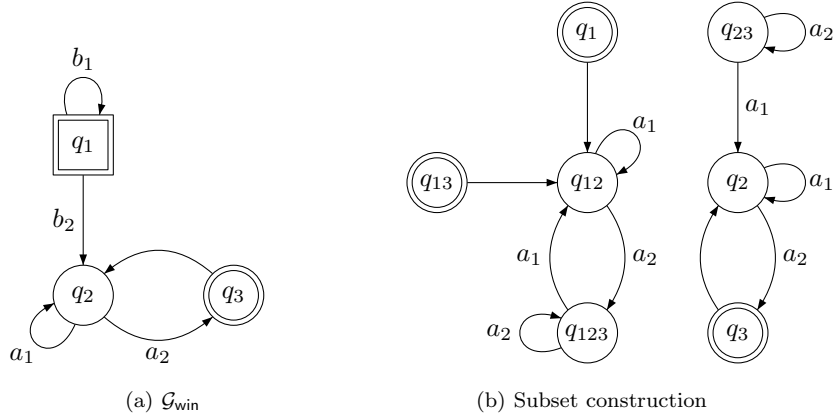

We illustrate the key technical insights with two examples.
First consider the deterministic game $\G_{{\sf win}}$ in \figurename~\ref{fig:gwin}, where the target set is $T = \{q_1,q_3\}$. 
Only state $q_2$ has a relevant choice for player~$1$, and only $q_1$ has a relevant choice for
player~$2$ (for the sake of clarity, we denote by $a_1,a_2$ the actions of player~$1$, and by 
$b_1,b_2$  the actions of player~$2$).

Player~$1$ is almost-sure weakly synchronizing in $T$ from every state. A winning strategy
plays $a_1$ at even rounds, and $a_2$ at odd rounds. Note that without the self-loop on $q_2$,
player~$1$ is no longer almost-sure weakly synchronizing in $T$ from $q_1$ 
(but still from $q_2$ and from $q_3$). 

Consider the subset construction in \figurename~\ref{fig:gwin-subset}, 
obtained by considering all subsets 
$q_{12} = \{q_1, q_2\}$, $q_{13} = \{q_1, q_3\}$, etc. of $Q$,
and with an edge from $s$ to $s'$ if there exists a selector $\alpha: Q \to \Act$
such that $s' = \delta_{\alpha}(s)$. 
Intuitively, the selector describes the actions played by a pure counting strategy of 
player~$1$ at a given round. \figurename~\ref{fig:gwin-subset} labels the 
edges $(s,s')$ with the action played in $q_2$ by the corresponding selector in $s$
(if relevant, i.e., if $q_2 \in s$). Accepting sets $s \subseteq T$
are marked by a double line.

An accepting strongly connected component is an SCC containing an accepting set, 
such as $\C = \{\{q_2\},\{q_3\}\}$ in our example. This is a witness that player~$1$
is almost-sure weakly synchronizing in $T$ from all states $q$ such that $\C$
is reachable from $\{q\}$ in $\P(\G)$. This sufficient condition for almost-sure
winning is not necessary, as player~$1$ is almost-sure
winning from $q_1$ as well, but the set $\{q_1\}$ cannot reach an accepting SCC 
in $\P(\G)$.    
However, we will show that if there is no accepting SCC in the subset construction,
then there is no state from which player~$1$ is almost-sure weakly synchronizing 
in $T$.

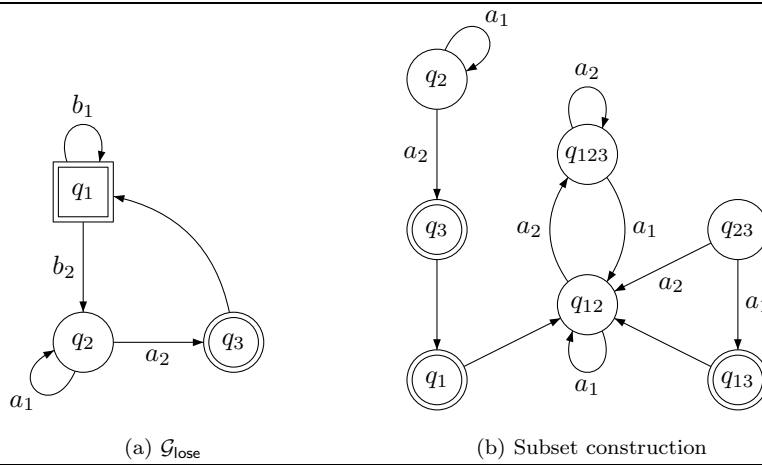
\begin{figure}[!tb]%
\begin{center}
\hrule
\subfloat[$\G_{{\sf lose}}${\large \strut}]{
   \begin{picture}(40,42)(0,10)

\node[Nmarks=r, Nmr=0](q1)(10,40){$q_1$}
\node[Nmarks=n](q2)(10,20){$q_2$}
\node[Nmarks=r](q3)(30,20){$q_3$}

\drawedge[ELside=r,ELpos=50, ELdist=1](q1,q2){$b_2$}
\drawedge[ELside=r,ELpos=50, ELdist=1](q2,q3){$a_2$}
\drawedge[ELside=l,ELpos=50, curvedepth=-6](q3,q1){}


\drawloop[ELside=l,loopCW=y, loopangle=90, loopdiam=5](q1){$b_1$}
\drawloop[ELside=l,loopCW=y, loopangle=225, loopdiam=5, ELdist=.5](q2){$a_1$}

\end{picture}
   \label{fig:glose}
}
\quad
\subfloat[Subset construction{\large \strut}]{
   \begin{picture}(60,55)(0,5)

\node[Nmarks=n](q2)(10,50){$q_{2}$}
\node[Nmarks=r](q3)(10,30){$q_{3}$}
\node[Nmarks=r](q1)(10,10){$q_{1}$}
\node[Nmarks=n](q123)(30,40){$q_{123}$}
\node[Nmarks=n](q12)(30,20){$q_{12}$}
\node[Nmarks=n](q23)(50,30){$q_{23}$}
\node[Nmarks=r](q13)(50,10){$q_{13}$}

\drawedge[ELside=r,ELpos=50, ELdist=1](q2,q3){$a_2$}
\drawedge[ELside=r,ELpos=50, ELdist=1](q3,q1){}
\drawedge[ELside=l,ELpos=50, curvedepth=0](q1,q12){}
\drawedge[ELside=l,ELpos=50, curvedepth=5](q12,q123){$a_2$}
\drawedge[ELside=l,ELpos=50, curvedepth=5](q123,q12){$a_1$}
\drawedge[ELside=l,ELpos=50, curvedepth=0](q23,q12){$a_2$}
\drawedge[ELside=l,ELpos=50, curvedepth=0](q23,q13){$a_1$}
\drawedge[ELside=l,ELpos=50, curvedepth=0](q13,q12){}


\drawloop[ELside=l,loopCW=y, loopangle=45, loopdiam=5, ELdist=.5](q2){$a_1$}
\drawloop[ELside=l,loopCW=y, loopangle=90, loopdiam=5](q123){$a_2$}
\drawloop[ELside=l,loopCW=y, loopangle=270, loopdiam=5](q12){$a_1$}

\end{picture}
   \label{fig:glose-subset}
}

\hrule
\caption{The deterministic game $\G_{{\sf lose}}$ where player~$1$ is not almost-sure
 weakly synchronizing (no matter the initial state), and its subset construction. \label{fig:gamelose}}%
\Description{A deterministic game graph and its subset construction.}
\end{center}
\end{figure}

This situation is illustrated in \figurename~\ref{fig:gamelose} where the
subset construction for the game $\G_{{\sf lose}}$ contains no accepting SCC,
and player~$1$ is not almost-sure weakly synchronizing in $T$
(no matter from which initial state). This is not trivial to see,
and we present the crux of the argument below. 
Although player~$1$ is not almost-sure weakly synchronizing in $T$,
it is not true either that player~$2$ can fix a strategy $\strab$ in $\G_{{\sf lose}}$
to prevent player~$1$ from almost-sure winning for weakly synchronizing in $T$.
This means that in general spoiling strategies can be constructed only after
a strategy for the other player has been fixed, which brings technical difficulty 
in the proofs.

\smallskip\noindent{\em Why player~$1$ can spoil player~$2$ in $\G_{{\sf lose}}$.}
Given an arbitrary strategy $\strab$ for player~$2$, we can construct
a strategy $\straa$ for player~$1$ such that the outcome sequence 
$\G^{\straa,\strab}_{{\sf lose}}$ (from any initial distribution)
is almost-sure weakly synchronizing in~$T$.

Consider the strategy $\straa_{loop}$ that always plays $a_1$ (to 
loop through $q_2$), and note that in the outcome 
$\G^{\straa_{loop},\strab}_{{\sf lose}} = d_0, d_1, \dots$
(from any initial distribution~$d_0$),
the probability mass in $q_2$ is non-decreasing, hence 
$\lim_{k \to \infty} d_k(q_2)$ exists, which we denote by $\alpha(d_0)$. 
We construct $\straa$ to play as follows, starting with $\epsilon = \frac{1}{2}$:
$(1)$ Given $\epsilon > 0$ and the current distribution $d$, 
play $\straa_{loop}$ for $n_{\epsilon}$ rounds, 
where $n_{\epsilon}$ is such that 
$d_{n_{\epsilon}}(q_2) \geq \alpha(d) - \epsilon$, 
then
$(2)$ play $a_2$ in the next round,
and $(3)$ repeat from $(1)$ with $\epsilon := \frac{\epsilon}{2}$.
In the outcome $\G^{\straa,\strab}_{{\sf lose}}$, after
playing $a_2$, the probability mass in $q_2$ is the 
probability mass transferred from $q_1$ in the previous step, 
which is at most $\epsilon$. It follows that the probability mass
in $T = \{q_1,q_3\}$ is at least $1-\epsilon$. The repetition
of this pattern for $\epsilon \to 0$ entails that 
$\G^{\straa,\strab}_{{\sf lose}}$ is almost-sure weakly 
synchronizing in $T$.

\begin{figure}[t]
\hrule
\begin{center}
    \newcommand{\token}{{\LARGE $\cdot$}}
\newcommand{\minitoken}{{\scriptsize $\cdot$}}
\newcommand{\btoken}{\green{\token}}
\newcommand{\rtoken}{\red{\token}}
\newcommand{\bminitoken}{\green{\minitoken}}
\newcommand{\rminitoken}{\red{\minitoken}}

\begin{picture}(83,38)(0,-2)

\gasset{Nh=5,Nw=4, Nmr=0, Nframe=n}

\gasset{AHnb=1, AHangle=30, AHLength=.8, AHlength=0}


\node[Nmarks=n](q01)(3,32){${1}$}

\nodelabel[ExtNL=n, NLdist=2](q01){\btoken\!\token\!\rtoken}


\node[Nmarks=n](q11)(10,32){$1$}
\node[Nmarks=n](q12)(10,25){${2}$}

\nodelabel[ExtNL=n, NLdist=2](q11){\token\!\rtoken}
\nodelabel[ExtNL=n, NLdist=2](q12){\btoken}

\drawedge[ELside=l,ELpos=45, ELdist=-.1](q01,q11){\minitoken\minitoken}
\drawedge[ELside=l,ELpos=45, ELdist=-.1](q01,q12){\minitoken}


\node[Nmarks=n](q21)(17,32){${1}$}
\node[Nmarks=n](q22)(17,25){${2}$}

\nodelabel[ExtNL=n, NLdist=2](q21){\token\!\rtoken}
\nodelabel[ExtNL=n, NLdist=2](q22){\btoken}

\drawedge[ELside=l,ELpos=45, ELdist=-.1](q11,q21){\minitoken\minitoken}
\drawedge[ELside=l,ELpos=45, ELdist=-.1, linegray=.6](q11,q22){}
\drawedge[ELside=l,ELpos=45, ELdist=-.1](q12,q22){\minitoken}


\node[Nmarks=n](q31)(24,32){${1}$}
\node[Nmarks=n](q32)(24,25){${2}$}
\node[Nmarks=n](q33)(24,18){${3}$}

\nodelabel[ExtNL=n, NLdist=2](q31){\rtoken}
\nodelabel[ExtNL=n, NLdist=2](q32){\token}
\nodelabel[ExtNL=n, NLdist=2](q33){\btoken}

\drawedge[ELside=l,ELpos=45, ELdist=-.1](q21,q31){\minitoken}
\drawedge[ELside=l,ELpos=45, ELdist=-.1](q21,q32){\minitoken}
\drawedge[ELside=l,ELpos=45, ELdist=-.1](q22,q33){\minitoken}


\node[Nmarks=n](q41)(31,32){$1$}
\node[Nmarks=n](q42)(31,25){${2}$}

\nodelabel[ExtNL=n, NLdist=2](q41){\btoken\!\rtoken}
\nodelabel[ExtNL=n, NLdist=2](q42){\token}

\drawedge[ELside=l,ELpos=45, ELdist=-.1](q31,q41){\minitoken}
\drawedge[ELside=l,ELpos=45, ELdist=-.1, linegray=.6](q31,q42){}
\drawedge[ELside=l,ELpos=52, ELdist=-.1](q32,q42){\minitoken}
\drawedge[ELside=l,ELpos=30, ELdist=.1, exo=-1, eyo=1](q33,q41){\minitoken}


\node[Nmarks=n](q51)(38,32){${1}$}
\node[Nmarks=n](q52)(38,25){${2}$}

\nodelabel[ExtNL=n, NLdist=2](q51){\btoken\!\rtoken}
\nodelabel[ExtNL=n, NLdist=2](q52){\token}

\drawedge[ELside=l,ELpos=45, ELdist=-.1](q41,q51){\minitoken\minitoken}
\drawedge[ELside=l,ELpos=45, ELdist=-.1, linegray=.6](q41,q52){}
\drawedge[ELside=l,ELpos=45, ELdist=-.1](q42,q52){\minitoken}


\node[Nmarks=n](q61)(45,32){$1$}
\node[Nmarks=n](q62)(45,25){${2}$}

\nodelabel[ExtNL=n, NLdist=2](q61){\btoken\!\rtoken}
\nodelabel[ExtNL=n, NLdist=2](q62){\token}

\drawedge[ELside=l,ELpos=45, ELdist=-.1](q51,q61){\minitoken\minitoken}
\drawedge[ELside=l,ELpos=45, ELdist=-.1, linegray=.6](q51,q62){}
\drawedge[ELside=l,ELpos=45, ELdist=-.1](q52,q62){\minitoken}


\node[Nmarks=n](q71)(52,32){${1}$}
\node[Nmarks=n](q72)(52,25){${2}$}
\node[Nmarks=n](q73)(52,18){${3}$}

\nodelabel[ExtNL=n, NLdist=2](q71){\rtoken}
\nodelabel[ExtNL=n, NLdist=2](q72){\btoken}
\nodelabel[ExtNL=n, NLdist=2](q73){\token}

\drawedge[ELside=l,ELpos=45, ELdist=-.1](q61,q71){\minitoken}
\drawedge[ELside=l,ELpos=45, ELdist=-.1](q61,q72){\minitoken}
\drawedge[ELside=l,ELpos=45, ELdist=-.1](q62,q73){\minitoken}


\node[Nmarks=n](q81)(59,32){${1}$}
\node[Nmarks=n](q82)(59,25){${2}$}
\node[Nmarks=n](q83)(59,18){${3}$}

\nodelabel[ExtNL=n, NLdist=2](q81){\token}
\nodelabel[ExtNL=n, NLdist=2](q82){\rtoken}
\nodelabel[ExtNL=n, NLdist=2](q83){\btoken}

\drawedge[ELside=l,ELpos=45, ELdist=-.1, linegray=.6](q71,q81){}
\drawedge[ELside=l,ELpos=45, ELdist=-.1](q71,q82){\minitoken}
\node[Nmarks=n, Nh=4,Nw=3](q72)(52,25){}
\drawedge[ELside=l,ELpos=45, ELdist=-.1](q72,q83){\minitoken}
\drawedge[ELside=l,ELpos=45, ELdist=0, exo=-1, eyo=1](q73,q81){\minitoken}


\node[Nmarks=n](q91)(66,32){${1}$}
\node[Nmarks=n](q92)(66,25){${2}$}
\node[Nmarks=n](q93)(66,18){${3}$}

\nodelabel[ExtNL=n, NLdist=2](q91){\btoken}
\nodelabel[ExtNL=n, NLdist=2](q92){\token}
\nodelabel[ExtNL=n, NLdist=2](q93){\rtoken}

\drawedge[ELside=l,ELpos=45, ELdist=-.1, linegray=.6](q81,q91){}
\drawedge[ELside=l,ELpos=45, ELdist=-.1](q81,q92){\minitoken}
\node[Nmarks=n, Nh=4,Nw=3](q82)(59,25){}
\drawedge[ELside=l,ELpos=45, ELdist=-.1](q82,q93){\minitoken}
\drawedge[ELside=l,ELpos=45, ELdist=0, exo=-1, eyo=1](q83,q91){\minitoken}


\node[Nmarks=n](qA1)(73,32){$1$}
\node[Nmarks=n](qA2)(73,25){${2}$}

\nodelabel[ExtNL=n, NLdist=2](qA1){\btoken\!\rtoken}
\nodelabel[ExtNL=n, NLdist=2](qA2){\token}

\drawedge[ELside=l,ELpos=45, ELdist=-.1](q91,qA1){\minitoken}
\drawedge[ELside=l,ELpos=45, ELdist=-.1, linegray=.6](q91,qA2){}
\drawedge[ELside=l,ELpos=52, ELdist=-.1](q92,qA2){\minitoken}
\drawedge[ELside=l,ELpos=30, ELdist=.1, exo=-1, eyo=1](q93,qA1){\minitoken}


\node[Nmarks=n](qB1)(80,32){${1}$}
\node[Nmarks=n](qB2)(80,25){${2}$}

\nodelabel[ExtNL=n, NLdist=2](qB1){\btoken\!\rtoken}
\nodelabel[ExtNL=n, NLdist=2](qB2){\token}

\drawedge[ELside=l,ELpos=45, ELdist=-.1](qA1,qB1){\minitoken\minitoken}
\drawedge[ELside=l,ELpos=45, ELdist=-.1, linegray=.6](qA1,qB2){}
\drawedge[ELside=l,ELpos=45, ELdist=-.1](qA2,qB2){\minitoken}


\node[Nmarks=n](label)(3,0){{\tiny ${0}$}}
\node[Nmarks=n](label)(10,0){{\tiny ${1}$}}
\node[Nmarks=n](label)(17,0){{\tiny ${2}$}}
\node[Nmarks=n](label)(24,0){{\tiny ${3}$}}
\node[Nmarks=n](label)(31,0){{\tiny ${4}$}}
\node[Nmarks=n](label)(38,0){{\tiny ${5}$}}
\node[Nmarks=n](label)(45,0){{\tiny ${6}$}}
\node[Nmarks=n](label)(52,0){{\tiny ${7}$}}
\node[Nmarks=n](label)(59,0){{\tiny ${8}$}}
\node[Nmarks=n](label)(66,0){{\tiny ${9}$}}
\node[Nmarks=n](label)(73,0){{\tiny ${10}$}}
\node[Nmarks=n](label)(80,0){{\tiny ${11}$}}

\node[Nmarks=n](label)(6.5,10){{\tiny $\begin{array}{ll} \frac{2}{3}b_1\\[+4pt] \frac{1}{3}b_2\end{array}$}}
\node[Nmarks=n](label)(13.5,10){{\tiny $b_1$}}
\node[Nmarks=n](label)(20.5,10){{\tiny $\begin{array}{ll} \frac{1}{2}b_1\\[+4pt] \frac{1}{2}b_2\end{array}$}}
\node[Nmarks=n](label)(27.5,10){{\tiny $b_1$}}
\node[Nmarks=n](label)(34.5,10){{\tiny $b_1$}}
\node[Nmarks=n](label)(41.5,10){{\tiny $b_1$}}
\node[Nmarks=n](label)(48.5,10){{\tiny $\begin{array}{ll} \frac{1}{2}b_1\\[+4pt] \frac{1}{2}b_2\end{array}$}}
\node[Nmarks=n](label)(55.5,10){{\tiny $b_2$}}
\node[Nmarks=n](label)(62.5,10){{\tiny $b_2$}}
\node[Nmarks=n](label)(69.5,10){{\tiny $b_1$}}
\node[Nmarks=n](label)(76.5,10){{\tiny $b_1$}}



\end{picture}
\end{center} 
\hrule
 \caption{Construction of a spoiling strategy for player~$2$ (in $\G_{{\sf lose}}$).
 \Description{A tree of paths, with some tokens labeling the edges.}
\label{fig:Glose-path}}
\end{figure}
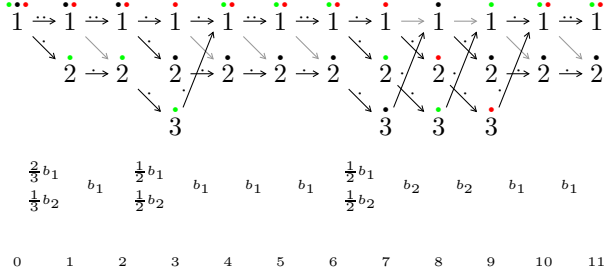

\smallskip\noindent{\em Why player~$2$ can spoil player~$1$ in $\G_{{\sf lose}}$.}
We sketch the crux of the argument for initial state $q_1$,
showing that player~$2$ can spoil an arbitrary
strategy $\straa: \nat \times Q \to \Act$ for player~$1$ (that is 
pure and counting), which is an infinite sequence of selectors
and thus corresponds to an infinite path from $\{q_1\}$ in the subset construction 
(shown in \figurename~\ref{fig:glose-subset}).

Such an infinite path is shown in \figurename~\ref{fig:Glose-path}
where an edge $(s,s')$ labeled by a selector $\alpha$ is drawn as the set 
of edges $(q,q')$ such that $q \in s$ and $q' = \delta(q,\alpha(q),b)$
for some $b \in \Act$. Note that, from some point on, all sets in
such a path contain a non-target state $q \in Q \setminus T$,
and a spoiling strategy for player~$2$ must ensure a bounded probability mass
in $Q \setminus T$, at every round from some point on.

In the example, player~$2$ can ensure a probability mass of $\frac{1}{3}$
in state $q_2 \in Q \setminus T$ from the second round on. 
In \figurename~\ref{fig:Glose-path}, we put three tokens, each carrying 
a probability mass of $\frac{1}{3}$, in the initial state $q_1$ and we 
show how the three tokens can move along the edges to always maintain 
one token in $q_2$ (after the first round). The choice of which edge 
from $q_1$ is taken by a token is made by player~$2$ with the corresponding
action $b_1$ or $b_2$ (the corresponding randomized selector is shown 
below the figure for each round -- edges are drawn in gray if no token 
flows through it). It is easy to show that the pattern 
suggested in \figurename~\ref{fig:Glose-path} can be prolonged ad infinitum.
A key insight is that player~$2$ should not move a token from $q_1$ to $q_2$ 
at every round (which may cause a depletion of tokens), but 
only in the rounds where player~$1$ sends the probability mass from $q_2$
to $q_3$. 

Each token follows a play that is compatible with the strategy of player~$1$.
The set of three plays that are followed by the three tokens 
has the property that in every 
round after round~$1$ at least one of the plays is in $q_2$. We say that the plays (or 
the tokens) \emph{cover} the state $q_2$ from round~$2$ on. Note that it would be
easy to cover $q_2$ from some round $n_0$ on by using plays of the form 
$(q_1)^nq_2Q^{\omega}$ ($n=n_0,n_0+1,\dots$), but this is an infinite set of plays,
which would require infinitely many tokens and would not allow a positive 
lower bound on the probability mass of each token. 
The key to cover $q_2$ with a finite number of plays is
to reuse the tokens when possible. It is however not obvious in general how to 
construct finitely many such plays, given an arbitrary strategy of player~$1$. 

We show in Lemma~\ref{lem:SCC-somewhere} that in deterministic games, a fixed number $K$ of 
tokens (each representing a probability mass $\frac{1}{K}$) is sufficient for 
player~$2$ to spoil any pure counting strategy of player~$1$, 
where $K = 2^{\abs{Q}}$. 

\begin{lemma}\label{lem:SCC-somewhere}
The following equivalence holds in deterministic games: there exists a state $q$ from 
which player~$1$ has a pure counting strategy that is almost-sure 
winning for weakly synchronizing in $T$
if and only if 
there exist a strongly connected component $\C$ 
in $\P(\G)$ and a set $s \in \C$ that is accepting. 
\end{lemma}

\begin{proof}
First, if there exists a strongly connected component $\C$ 
in $\P(\G)$ and an accepting set $s \in \C$, then there exists an infinite
path $s_0 s_1 \dots$ in $\P(\G)$ from $s_0 = s$ that visits $s$ infinitely often. 
Consider the corresponding sequence of selectors $\alpha_0 \alpha_1 \dots$ (such that
$s_{i+1} = \delta_{\alpha_i}(s_i)$). It is easy to see that from all states $q\in s$,
the pure counting strategy $\straa$ defined by 
$\straa(i,q) = \alpha_i(q)$ is sure (thus also almost-sure) winning for weakly 
synchronizing in $s \subseteq T$.

For the second direction, we prove the contrapositive. Assume that no SCC
in $\P(\G)$ contains an accepting set, and show that from all states~$q_0$ 
no pure counting strategy for player~$1$ is almost-sure winning for 
weakly synchronizing in $T$.

Let $\straa: \nat \to (Q \to \Act)$ be a pure counting strategy for player~$1$, 
and we construct a spoiling strategy for player~$2$ from $q_0$. 
First, consider the sequence $s_0, s_1, \dots$ defined by $s_0 = \{q_0\}$
and $s_{i+1} = \delta_{\straa(i)}(s_{i})$, and as this sequence is a path
in the subset construction $\P(\G)$, by our assumption 
there exists an index $i_0$ such that $s_i$ is non-accepting for all $i \geq i_0$. 
Let $\Reject = Q \setminus T$. We have: 
\begin{linenomath*}
$$ s_i \cap \Reject \neq \emptyset \text{ for all } i \geq i_0.$$
\end{linenomath*}
By the central property of the subset construction, for every set $s_i$
and every state $q_i \in s_i$ (in particular for $q_i \in s_i \cap \Reject$ 
if $i \geq i_0$), there exists a play of length $i$ from $q_0$ to
$q_i$ that is compatible with the strategy $\straa$. Player~$2$ can use
such plays to inject positive probability into $q_i$ at every position $i$. 
However, if all those plays form an infinite set (as illustrated by the 
plays $(q_1)^nq_2Q^{\omega}$ in the example of \figurename~\ref{fig:glose}), 
then player~$2$ may not be able to guarantee a bounded probability mass  
in $q_i$ at every round $i$ from some point on. 

Now we construct a data structure that will help player~$2$ to determine
their strategy and to construct a \emph{finite} set $\Pi$ of plays 
compatible with $\straa$ in $\G$
such that, for all $i \geq i_0$, there exists a play $\pi \in \Pi$
with $\Last(\pi(i)) \in \Reject$.

The data structure consists, for each position $i \geq i_0$, 
of a tuple $u_i = \tuple{r_1, \dots, r_k}$ of registers, 
each storing a nonempty subset of $Q$.  
The number of registers is not fixed (but will never decrease along the sequence
$u_{i_0}, u_{i_0+1}, \dots$), and for $i = i_0$ let 
$u_{i_0} = \tuple{r_1}$, thus $u_{i_0}$ consists of one register,
and let $r_1 = \{q_{i_0} \}$ where $q_{i_0} \in s_{i_0} \cap \Reject$.
Each register $r$ in $u_i$ corresponds to one token, and stores the possible
states in which the token can be after $i$ steps by following a play 
compatible with $\straa$.

Given $u_i = \tuple{r_1, \dots, r_k}$, define $u_{i+1}$ as follows:
first, let $r'_j = \straa(i)(r_{j})$ be the set obtained from $r_j$
by following the selector $\straa(i)$ at position $i$ (which is
the same selector used to define $s_{i+1}$ from $s_{i}$). We consider
two cases:

\begin{enumerate}
\item if all registers are accepting, that is 
$r'_j \cap \Reject = \emptyset \text{ for all } 1 \leq j \leq k$:
let $q_{i+1} \in s_{i+1} \cap \Reject$ and create a new register $r'_{k+1} = \{q_{i+1}\}$,
define  $u_{i+1} = \tuple{\underbrace{r'_1, \dots, r'_k}_{\text{accepting}}, r'_{k+1}}$;

\item otherwise, some register is non-accepting, and let $j$ be the largest 
index such that $r'_j \cap \Reject \neq \emptyset$. 
Let $q_{i+1} \in r'_j \cap \Reject$ and we remove the register at position $j$,
and replace it by a new register $r'_{k+1} = \{q_{i+1}\}$ at the end of the tuple.
Define $u_{i+1} = \tuple{r'_1, \dots, r'_{j-1}, \underbrace{r'_{j+1}, \dots, r'_k}_{\text{accepting}}, r'_{k+1}}$.
\end{enumerate}

In both cases, we say that the parent of a register $r'_l$ (for $l \leq k$)
that occurs in $u_{i+1}$ 
is the register $r_l$ in $u_{i}$, and that $r'_{k+1}$ has no parent (in the second 
case above, we say that $r'_{k+1}$ is a clone child of $r_j$). An ancestor
of a register $r$ in $u_{i}$ is either $r$ or an ancestor of the parent of $r$
(but not of the clone parent of $r$).
In the sequel, we assume that the registers in a tuple $u_i$ are called
$r_1, r_2, \dots$ with consecutive indices in the order they appear in $u_i$.

We now state key invariant properties of the sequence $u_{i_0}, u_{i_0+1}, \dots$,
for all $i \geq i_0$:
\begin{itemize}
\item (\emph{consistency property}) in every $u_i$, for every register $r$ in $u_i$
we have $r \subseteq s_{i}$; 
\item (\emph{singleton property}) in every $u_i$, the rightmost register is a 
singleton containing a non-accepting state; 
\item (\emph{chain property}) for all $j \geq i$, every chain of registers (with order defined by the 
parent relation) from an ancestor register $r$ in $u_i$ to a register $r'$ in $u_j$ 
is a path in 
the subset construction labeled by the selectors $\straa(i), \dots, \straa(j)$;
\item (\emph{key property}) in every $u_i$, if there are $k$ registers at the right
of a register $r$, then $r$ has at least $k$ ancestors that are accepting.
\end{itemize}
It is easy to verify that these properties hold by construction of $u_{i_0}$,
and of $u_{i+1}$ from $u_{i}$ (by induction). In particular the key property
holds because whenever a register is appended (or moved) to the right of a
register $r$, then $r$ is accepting.

By our initial assumption, a register $r$ cannot have more than $2^{\abs{Q}}$ 
ancestors that are accepting (using the chain property).
Then it follows from the key property that a tuple $u_i$ cannot have more
than $2^{\abs{Q}}$ registers, and since the number of registers is not decreasing, 
there is an index $i^*_0$ such that all $u_{i}$, for $i \geq i^*_0$,
contain the same number $K \leq 2^{\abs{Q}}$ of registers. 
We are now ready to construct the set $\Pi$, containing $K$ plays.

For each $i > i^*_0$, consider the permutation $f_i$ on $\{1,\dots, K\}$
that maps index $j$ to $k = f_i(j)$ such that $r_{k}$ in $u_{i-1}$ is the parent 
of $r_j$ in $u_{i}$ (or clone parent if $r_j$ has no parent).

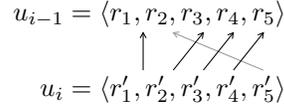
\begin{figure}[h]
\begin{center}
    \newcommand{\token}{{\LARGE $\cdot$}}
\newcommand{\minitoken}{{\scriptsize $\cdot$}}
\newcommand{\btoken}{\green{\token}}
\newcommand{\rtoken}{\red{\token}}
\newcommand{\bminitoken}{\green{\minitoken}}
\newcommand{\rminitoken}{\red{\minitoken}}

\begin{picture}(40,15)(0,0)

\gasset{Nh=5,Nw=4, Nmr=0, Nframe=n}

\gasset{AHnb=0, AHangle=30, AHLength=.8, AHlength=0}

\gasset{ATnb=1, ATangle=30, ATLength=.8, ATlength=0}


\node[Nmarks=n,Nw=60](q0)(37,10){\makebox(0,0)[r]{$u_{i-1} = \tuple{r_1,r_2,r_3,r_4,r_5}$}}
\node[Nmarks=n,Nw=60](q1)(37,0){\makebox(0,0)[r]{$u_{i} = \tuple{r'_1,r'_2,r'_3,r'_4,r'_5}$}}


\drawedge[ELside=l, curvedepth=0, syo=-2.5, eyo=2.5, sxo=-19, exo=-19](q0,q1){}
\drawedge[ELside=l, curvedepth=0, linegray=.6, syo=-2.5, eyo=2.5, sxo=-15, exo=-3](q0,q1){}
\drawedge[ELside=l, curvedepth=0, syo=-2.5, eyo=2.5, sxo=-11, exo=-15](q0,q1){}
\drawedge[ELside=l, curvedepth=0, syo=-2.5, eyo=2.5, sxo=-7, exo=-11](q0,q1){}
\drawedge[ELside=l, curvedepth=0, syo=-2.5, eyo=2.5, sxo=-3, exo=-7](q0,q1){}






\end{picture}
\end{center} 
 \caption{A permutation on registers. \label{fig:permutation1}}
 \Description{Two tuples of registers, with arcs connecting the pairs in a permutation.}
\end{figure}


Following the permutation $f_i$ at position $i$ (for $i > i^*_0$), 
we can define $K$ equivalence classes $C_1,\dots,C_k$ of registers that contain exactly 
one register from each $u_i$: 
the register $r_j$ in the tuple $u_i$ belongs to the class $C_k$ with 
$k = (f_{i^*_0+1} \circ f_{i^*_0+2} \circ \dots \circ f_i)(j)$,
see the illustration in \figurename~\ref{fig:permutation2}.
Note in particular that every rightmost register $r_K$ (highlighted 
in \figurename~\ref{fig:permutation2}), which is a singleton, belongs 
to some class.

\begin{figure}[h]
\begin{center}
    \newcommand{\token}{{\LARGE $\cdot$}}
\newcommand{\minitoken}{{\scriptsize $\cdot$}}
\newcommand{\btoken}{\green{\token}}
\newcommand{\rtoken}{\red{\token}}
\newcommand{\bminitoken}{\green{\minitoken}}
\newcommand{\rminitoken}{\red{\minitoken}}

\begin{picture}(40,45)(0,0)

\gasset{Nh=.5,Nw=.5, Nmr=.25, Nframe=y, Nfill=y}

\gasset{AHnb=0, AHangle=30, AHLength=.8, AHlength=0}

\gasset{ATnb=0, ATangle=30, ATLength=.8, ATlength=0}



\node[Nmarks=n, NLext=y, NLangle=90, NLdist=2](q00)(5,40){{\scriptsize $C_1$}}
\node[Nmarks=n, NLext=y, NLangle=90, NLdist=2](q01)(10,40){{\scriptsize $C_2$}}
\node[Nmarks=n, NLext=y, NLangle=90, NLdist=2](q02)(15,40){{\scriptsize $C_3$}}
\node[Nmarks=n, NLext=y, NLangle=90, NLdist=2](q03)(20,40){{\scriptsize $C_4$}}
\node[Nmarks=n, NLext=y, NLangle=90, NLdist=2, Nfill=n, Nh=.8, Nw=.8, Nmr=.4](q04)(25,40){{\scriptsize $C_5$}}

\node[Nmarks=n, Nframe=n, Nfill=n](label)(30,40){\makebox(0,0)[l]{\small $i^*_0$}}

\node[Nframe=n, Nfill=n, Nh=1.5, Nw=1.5, Nmr=.75](q00)(5,40){}
\node[Nframe=n, Nfill=n, Nh=1.5, Nw=1.5, Nmr=.75](q01)(10,40){}
\node[Nframe=n, Nfill=n, Nh=1.5, Nw=1.5, Nmr=.75](q02)(15,40){}
\node[Nframe=n, Nfill=n, Nh=1.5, Nw=1.5, Nmr=.75](q03)(20,40){}
\node[Nframe=n, Nfill=n, Nh=2, Nw=2, Nmr=1](q04)(25,40){}


\node[Nmarks=n](q10)(5,35){}
\node[Nmarks=n](q11)(10,35){}
\node[Nmarks=n](q12)(15,35){}
\node[Nmarks=n](q13)(20,35){}
\node[Nmarks=n, Nfill=n, Nh=.8, Nw=.8, Nmr=.4](q14)(25,35){}

\node[Nmarks=n, Nframe=n, Nfill=n](label)(30,35){\makebox(0,0)[l]{\small $i^*_0+1$}}

\node[Nframe=n, Nfill=n, Nh=1.5, Nw=1.5, Nmr=.75](q10)(5,35){}
\node[Nframe=n, Nfill=n, Nh=1.5, Nw=1.5, Nmr=.75](q11)(10,35){}
\node[Nframe=n, Nfill=n, Nh=1.5, Nw=1.5, Nmr=.75](q12)(15,35){}
\node[Nframe=n, Nfill=n, Nh=1.5, Nw=1.5, Nmr=.75](q13)(20,35){}
\node[Nframe=n, Nfill=n, Nh=2, Nw=2, Nmr=1](q14)(25,35){}

\drawedge[ELside=l, curvedepth=0](q00,q10){}
\drawedge[ELside=l, curvedepth=0](q01,q11){}
\drawedge[ELside=l, curvedepth=0](q02,q12){}
\drawedge[ELside=l, curvedepth=0, linegray=.6](q03,q14){}
\drawedge[ELside=l, curvedepth=0](q04,q13){}


\node[Nmarks=n](q20)(5,30){}
\node[Nmarks=n](q21)(10,30){}
\node[Nmarks=n](q22)(15,30){}
\node[Nmarks=n](q23)(20,30){}
\node[Nmarks=n, Nfill=n, Nh=.8, Nw=.8, Nmr=.4](q24)(25,30){}

\node[Nmarks=n, Nframe=n, Nfill=n](label)(30,30){\makebox(0,0)[l]{\small $i^*_0+2$}}

\node[Nframe=n, Nfill=n, Nh=1.5, Nw=1.5, Nmr=.75](q20)(5,30){}
\node[Nframe=n, Nfill=n, Nh=1.5, Nw=1.5, Nmr=.75](q21)(10,30){}
\node[Nframe=n, Nfill=n, Nh=1.5, Nw=1.5, Nmr=.75](q22)(15,30){}
\node[Nframe=n, Nfill=n, Nh=1.5, Nw=1.5, Nmr=.75](q23)(20,30){}
\node[Nframe=n, Nfill=n, Nh=2, Nw=2, Nmr=1](q24)(25,30){}

\drawedge[ELside=l, curvedepth=0, linegray=.6](q10,q24){}
\drawedge[ELside=l, curvedepth=0](q11,q20){}
\drawedge[ELside=l, curvedepth=0](q12,q21){}
\drawedge[ELside=l, curvedepth=0](q13,q22){}
\drawedge[ELside=l, curvedepth=0](q14,q23){}


\node[Nmarks=n](q30)(5,25){}
\node[Nmarks=n](q31)(10,25){}
\node[Nmarks=n](q32)(15,25){}
\node[Nmarks=n](q33)(20,25){}
\node[Nmarks=n, Nfill=n, Nh=.8, Nw=.8, Nmr=.4](q34)(25,25){}

\node[Nmarks=n, Nframe=n, Nfill=n](label)(30,25){\makebox(0,0)[l]{\small $i^*_0+3$}}

\node[Nframe=n, Nfill=n, Nh=1.5, Nw=1.5, Nmr=.75](q30)(5,25){}
\node[Nframe=n, Nfill=n, Nh=1.5, Nw=1.5, Nmr=.75](q31)(10,25){}
\node[Nframe=n, Nfill=n, Nh=1.5, Nw=1.5, Nmr=.75](q32)(15,25){}
\node[Nframe=n, Nfill=n, Nh=1.5, Nw=1.5, Nmr=.75](q33)(20,25){}
\node[Nframe=n, Nfill=n, Nh=2, Nw=2, Nmr=1](q34)(25,25){}

\drawedge[ELside=l, curvedepth=0](q20,q30){}
\drawedge[ELside=l, curvedepth=0](q21,q31){}
\drawedge[ELside=l, curvedepth=0, linegray=.6](q22,q34){}
\drawedge[ELside=l, curvedepth=0](q23,q32){}
\drawedge[ELside=l, curvedepth=0](q24,q33){}


\node[Nmarks=n](q40)(5,20){}
\node[Nmarks=n](q41)(10,20){}
\node[Nmarks=n](q42)(15,20){}
\node[Nmarks=n](q43)(20,20){}
\node[Nmarks=n, Nfill=n, Nh=.8, Nw=.8, Nmr=.4](q44)(25,20){}

\node[Nmarks=n, Nframe=n, Nfill=n](label)(30,20){\makebox(0,0)[l]{\small $i^*_0+4$}}

\node[Nframe=n, Nfill=n, Nh=1.5, Nw=1.5, Nmr=.75](q40)(5,20){}
\node[Nframe=n, Nfill=n, Nh=1.5, Nw=1.5, Nmr=.75](q41)(10,20){}
\node[Nframe=n, Nfill=n, Nh=1.5, Nw=1.5, Nmr=.75](q42)(15,20){}
\node[Nframe=n, Nfill=n, Nh=1.5, Nw=1.5, Nmr=.75](q43)(20,20){}
\node[Nframe=n, Nfill=n, Nh=2, Nw=2, Nmr=1](q44)(25,20){}

\drawedge[ELside=l, curvedepth=0](q30,q40){}
\drawedge[ELside=l, curvedepth=0](q31,q41){}
\drawedge[ELside=l, curvedepth=0](q32,q42){}
\drawedge[ELside=l, curvedepth=0](q33,q43){}
\drawedge[ELside=l, curvedepth=0, linegray=.6](q34,q44){}


\node[Nmarks=n](q50)(5,15){}
\node[Nmarks=n](q51)(10,15){}
\node[Nmarks=n](q52)(15,15){}
\node[Nmarks=n](q53)(20,15){}
\node[Nmarks=n, Nfill=n, Nh=.8, Nw=.8, Nmr=.4](q54)(25,15){}

\node[Nmarks=n, Nframe=n, Nfill=n](label)(30,15){\makebox(0,0)[l]{\small $i^*_0+5$}}

\node[Nframe=n, Nfill=n, Nh=1.5, Nw=1.5, Nmr=.75](q50)(5,15){}
\node[Nframe=n, Nfill=n, Nh=1.5, Nw=1.5, Nmr=.75](q51)(10,15){}
\node[Nframe=n, Nfill=n, Nh=1.5, Nw=1.5, Nmr=.75](q52)(15,15){}
\node[Nframe=n, Nfill=n, Nh=1.5, Nw=1.5, Nmr=.75](q53)(20,15){}
\node[Nframe=n, Nfill=n, Nh=2, Nw=2, Nmr=1](q54)(25,15){}

\drawedge[ELside=l, curvedepth=0](q40,q50){}
\drawedge[ELside=l, curvedepth=0, linegray=.6](q41,q54){}
\drawedge[ELside=l, curvedepth=0](q42,q51){}
\drawedge[ELside=l, curvedepth=0](q43,q52){}
\drawedge[ELside=l, curvedepth=0](q44,q53){}


\node[Nmarks=n](q60)(5,10){}
\node[Nmarks=n](q61)(10,10){}
\node[Nmarks=n](q62)(15,10){}
\node[Nmarks=n](q63)(20,10){}
\node[Nmarks=n, Nfill=n, Nh=.8, Nw=.8, Nmr=.4](q64)(25,10){}

\node[Nmarks=n, Nframe=n, Nfill=n](label)(30,10){\makebox(0,0)[l]{\small $i^*_0+6$}}

\node[Nframe=n, Nfill=n, Nh=1.5, Nw=1.5, Nmr=.75](q60)(5,10){}
\node[Nframe=n, Nfill=n, Nh=1.5, Nw=1.5, Nmr=.75](q61)(10,10){}
\node[Nframe=n, Nfill=n, Nh=1.5, Nw=1.5, Nmr=.75](q62)(15,10){}
\node[Nframe=n, Nfill=n, Nh=1.5, Nw=1.5, Nmr=.75](q63)(20,10){}
\node[Nframe=n, Nfill=n, Nh=2, Nw=2, Nmr=1](q64)(25,10){}

\drawedge[ELside=l, curvedepth=0](q50,q60){}
\drawedge[ELside=l, curvedepth=0](q51,q61){}
\drawedge[ELside=l, curvedepth=0](q52,q62){}
\drawedge[ELside=l, curvedepth=0, linegray=.6](q53,q64){}
\drawedge[ELside=l, curvedepth=0](q54,q63){}


\node[Nmarks=n](q70)(5,5){}
\node[Nmarks=n](q71)(10,5){}
\node[Nmarks=n](q72)(15,5){}
\node[Nmarks=n](q73)(20,5){}
\node[Nmarks=n, Nfill=n, Nh=.8, Nw=.8, Nmr=.4](q74)(25,5){}

\node[Nmarks=n, Nframe=n, Nfill=n](label)(30,5){\makebox(0,0)[l]{\small $i^*_0+7$}}

\node[Nframe=n, Nfill=n, Nh=1.5, Nw=1.5, Nmr=.75](q70)(5,5){}
\node[Nframe=n, Nfill=n, Nh=1.5, Nw=1.5, Nmr=.75](q71)(10,5){}
\node[Nframe=n, Nfill=n, Nh=1.5, Nw=1.5, Nmr=.75](q72)(15,5){}
\node[Nframe=n, Nfill=n, Nh=1.5, Nw=1.5, Nmr=.75](q73)(20,5){}
\node[Nframe=n, Nfill=n, Nh=2, Nw=2, Nmr=1](q74)(25,5){}

\drawedge[ELside=l, curvedepth=0](q60,q70){}
\drawedge[ELside=l, curvedepth=0](q61,q71){}
\drawedge[ELside=l, curvedepth=0, linegray=.6](q62,q74){}
\drawedge[ELside=l, curvedepth=0](q63,q72){}
\drawedge[ELside=l, curvedepth=0](q64,q73){}


\node[Nmarks=n, Nframe=n, Nfill=n](q80)(5,2){}
\node[Nmarks=n, Nframe=n, Nfill=n](q81)(10,2){}
\node[Nmarks=n, Nframe=n, Nfill=n](q82)(15,2){}
\node[Nmarks=n, Nframe=n, Nfill=n](q83)(20,2){}
\node[Nmarks=n, Nframe=n, Nfill=n](q84)(25,2){}

\drawedge[ELside=l, curvedepth=0, dash={.2 .5}0](q70,q80){}
\drawedge[ELside=l, curvedepth=0, dash={.2 .5}0](q71,q81){}
\drawedge[ELside=l, curvedepth=0, dash={.2 .5}0](q72,q82){}
\drawedge[ELside=l, curvedepth=0, dash={.2 .5}0](q73,q83){}
\drawedge[ELside=l, curvedepth=0, dash={.2 .5}0](q74,q84){}








\end{picture}
\end{center} 
 \caption{A sequence of permutations on registers. \label{fig:permutation2}}
 \Description{Several tuples of registers, with arcs connecting the pairs in a permutation on successive tuples.}
\end{figure}
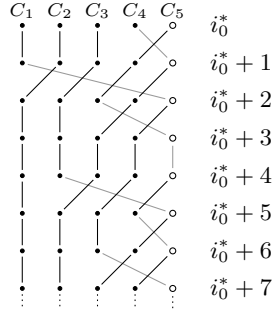

Using the central property of the subset construction, from those $K$ classes 
we construct $K$ infinite plays in $\G$, all 
compatible with $\straa(i^*_0), \straa(i^*_0+1), \dots$, and such that for all
$i \geq i^*_0$, for all registers $r_j$ in $u_i$, one of the plays is in 
a state of $r_j$ at position $i$. In particular, since $r_K$ is always
a non-accepting singleton, the constructed plays cover a non-accepting state
at every position $i \geq i^*_0$. 

Given an equivalence class $C$, consider the register $r$ of $u_i$ in $C$,
and the register $r'$ of $u_{i+1}$ in $C$. Then, by the central property 
of the subset construction, for all states $q' \in r'$, there exists
a state $q \in r$ such that $q' = \delta(q,\alpha_i(q),b)$
for some action $b \in \Act$ of player~$2$ (no matter whether $r$ is the parent
or clone parent of $r'$).
Now consider all \emph{singular} positions $i$ where $r_K$ (the rightmost, singleton, register) 
is the register of $u_i$ in $C$. Between any such two positions $i_1 < i_2$,  
there exists a (segment of) play in $\G$ from the state in the register $r_K$ at position $i_1$
to the state in the register $r_K$ at position $i_2$, that is compatible 
with the strategy $\straa$ (at the corresponding positions $i_1,i_1+1, \dots, i_2$)
using the chain property. 
Analogously, from some state in $s_{i^*_0}$ (using the consistency property) 
there is a segment of play to the state in register $r_K$ at the first 
singular position $i_0$. The constructed segments can be concatenated to form
a single play, and if this play is finite, we can prolong it to
an infinite play compatible with~$\straa$.

Given the $K$ plays in $\G$ constructed in this way from $C_1,\dots,C_k$ , 
it is easy to construct a strategy for player~$2$ 
that ensures a probability mass of $\frac{1}{K}$ (a token) will move along each
of the $K$ plays (possibly using randomization), and therefore 
a probability mass of at least $\frac{1}{K}$ in a non-accepting state
at every round $i \geq i^*_0$, showing that the strategy $\straa$ of player~$1$
is not almost-sure winning for weakly synchronizing in $T$,
which concludes the proof.
\end{proof}

In the deterministic game $\G_{{\sf win}}$ of \figurename~\ref{fig:gwin},
the strongly connected component $\C = \{\{q_2\},\{q_3\}\}$ in the subset 
construction $\P(\G_{{\sf win}})$,
which contains the accepting set $U = \{q_2\}$, shows that player~$1$ is almost-sure
winning from some state (Lemma~\ref{lem:SCC-somewhere}), 
namely from $q_2$ and from $q_3$. This holds even if there is no self-loop on $q_2$.
However, whether player~$1$ is almost-sure winning from $q_1$ depends on the
presence of that self-loop: with the self-loop on $q_2$, the period of the 
SCC $\C$ is $p=1$ (see \figurename~\ref{fig:gwin-subset}) and player~$1$ is almost-sure 
winning from $q_1$, 
whereas without the self-loop, the period of $\C$ is $p=2$ and player~$1$ is 
not almost-sure winning from $q_1$ (player~$2$ can inject an equal 
mass of probability from $q_1$ to $q_2$ in two successive rounds, and 
as those masses can never merge, the probability mass in $q_2$ is always bounded 
away from $1$). In fact, player~$1$ needs to ensure that any probability mass
injected in $\C$ is always injected at the same round (modulo $p$), where $p$ 
is the period of $\C$.

To track the number of rounds modulo $p$, define the game $\G \times [p]$
that follows the transitions of $\G$ and decrements a tracking counter (modulo $p$)
along each transition (
\figurename~\ref{fig:gwin-counter-all}). Formally, let  $\G \times [p] = \tuple{Q', \Act, \delta'}$
where $Q' = Q \times \{p-1,\dots,1,0\}$ and $\delta'$ is defined as follows,
for all $\tuple{q,i}, \tuple{q',j} \in Q'$ and $a \in \Act$:
\begin{linenomath*}
$$\delta'(\tuple{q,i},a)(\tuple{q',j}) = 
\begin{cases}
\delta(q,a)(q') & \text{ if } j=i-1 \!\!\mod p,\\
0 & \text{ otherwise.}
\end{cases}
$$
\end{linenomath*}



Given $0 \leq t_0 < p$,
there is a bijection $\mu_{t_0}$\label{ref:bijection} between the histories 
in $\G$ and in $\G \times [p]$ that maps
the history $q_0 \, a_0b_0 \, q_1 \ldots q_k$ in $\G$ 
to the history $(q_0,t_0) \, a_0b_0 \, (q_1,t_1) \ldots (q_k, t_k)$ in $\G \times [p]$
where $t_i = t_{i-1} - 1 \mod p$ for all $0 < i \leq k$.
Therefore, when $t_0$ is fixed, strategies in $\G$ can be transformed into
strategies in $\G \times [p]$, via this bijection. 
In the sequel, we take the freedom to omit mentioning that this bijection
needs to be applied, and we consider that strategies can be played 
both in $\G$ and in $\G \times [p]$.
We say that a distribution $d$ on $Q \times \{p-1,\dots,1,0\}$ is 
\emph{proper} if there exists an index $0 \leq j < p$ such that if 
$\Supp(d) \subseteq Q \times \{j\}$. We also omit the bijection $\mu_{j}$ between 
proper distributions and distributions on $Q$ (when the index $j$ is not relevant, 
or clear from the context). 
For a distribution $d$ in $Q$, we denote by $d \times \{j\}$ 
the corresponding proper distribution such that $d \times \{j\}(q,i)$ 
is equal to $d(q)$ if $i=j$ (and $0$ otherwise).

A simple property relating the game $\G$ with the games $\G \times [p]$
is that player~$1$ can fix (in advance, regardless of the strategy of player~$2$) 
the value of the counter when synchronization occurs in $T$.

\begin{lemma}\label{lem:game-counter}
Player~$1$ is almost-sure weakly synchronizing in $T$ from $q$
in the game $\G$
if and only if 
for all $p \geq 0$, there exists $0 \leq i \leq p-1$
such that player~$1$ is almost-sure weakly synchronizing in $T \times \{0\}$ 
from $\tuple{q,i}$ in $\G \times [p]$.
\end{lemma}

\begin{proof}
It is immediate that almost-sure weakly synchronizing in $T \times \{0\}$ 
from $\tuple{q,i}$ in the game $\G \times [p]$ implies almost-sure weakly synchronizing 
in $T$ from $q$ in $\G$.

For the converse implication, we prove the contrapositive.
If for all $0 \leq i \leq p-1$, player~$1$ is not almost-sure 
weakly synchronizing in $T \times \{0\}$ from $\tuple{q,i}$ in $\G \times [p]$,
then for all player-$1$ strategies $\straa$, 
there exist player-$2$ strategies $\strab_{0},\strab_{1},\dots,\strab_{p-1}$ 
such that $\G^{\straa,\strab_{i}}_{\tuple{q,i}}$ is not 
almost-sure weakly synchronizing in $T \times \{0\}$,
that is there exist rounds $j_0,j_1,\dots,j_{p-1}$ and positive bounds 
$\epsilon_0,\epsilon_1,\dots,\epsilon_{p-1}$ such that
for all $0 \leq i \leq p-1$, for all $j \geq j_i$ with $j = i \mod p$, 
under strategies $\straa,\strab_{i}$ in $\G$ from $q$
the probability mass in $T$ at position~$j$ is at most $1-\epsilon_i$.
Consider the strategy $\strab$ playing the superposition $\sum_i \frac{1}{p}\strab_{i}$. 
For $\epsilon = \min_i \epsilon_i$ and $j^* = \max_i j_i$, 
under strategies $\straa,\strab$ in $\G$ from $q$ 
the probability mass in $T$ at all positions after $j^*$ is at most $1-\epsilon$,
which shows that $\straa$ is not almost-sure weakly synchronizing in $T$.
\end{proof}

\begin{figure}[!tb]%
\begin{center}
\hrule
\subfloat[\mbox{$\G_{{\sf win}} \times [1]$}{\large \strut}]{
   \begin{picture}(20,60)(0,0)


\drawline[AHnb=0, linegray=.6, dash={1.2}0](3,49)(3,39)(17,39)(17,2)

\node[Nmarks=r, Nmr=0](q10)(10,45){{\footnotesize ${q_1,0}$}}
\node[Nmarks=n, linegray=.6](q20)(10,25){\gray50{{\footnotesize $q_2,0$}}}
\node[Nmarks=r, linegray=.6](q30)(10,5){\gray50{{\footnotesize $q_3,0$}}}

\drawedge[ELside=r,ELpos=50, linegray=.6](q10,q20){\gray50{$b_2$}}

\drawedge[ELside=r,ELpos=50, ELdist=1, curvedepth=-5, linegray=.6](q20,q30){}
\drawedge[ELside=r,ELpos=50, ELdist=1, curvedepth=-5, linegray=.6](q30,q20){}


\drawloop[ELside=l,loopCW=y, loopangle=145, loopdiam=5, linegray=.6](q20){}
\drawloop[ELside=l,loopCW=y, loopangle=90, loopdiam=5](q10){$b_1$}

\end{picture}
   \label{fig:gwin-counter}
}
\quad\quad\quad\quad   
\subfloat[\mbox{$\G'_{{\sf win}} \times [2]${\large \strut}}]{
   \begin{picture}(40,60)(0,0)


\drawline[AHnb=0, linegray=.6, dash={1.2}0](3,49)(3,42)(20,25)(20,2)

\node[Nmarks=r, Nmr=0](q10)(10,45){{\footnotesize ${q_1,0}$}}
\node[Nmarks=n, Nmr=0](q11)(30,45){{\footnotesize ${q_1,1}$}}
\node[Nmarks=n, linegray=.6](q21)(10,25){\gray50{{\footnotesize $q_2,1$}}}
\node[Nmarks=n](q20)(30,25){{\footnotesize $q_2,0$}}
\node[Nmarks=r, linegray=.6](q30)(10,5){\gray50{{\footnotesize $q_3,0$}}}
\node[Nmarks=n](q31)(30,5){{\footnotesize $q_3,1$}}

\drawedge[ELside=l,ELpos=50, ELdist=1, curvedepth=5](q10,q11){$b_1$}
\drawedge[ELside=l,ELpos=50, ELdist=1, curvedepth=5](q11,q10){$b_1$}

\drawedge[ELside=r,ELpos=50, linegray=.6](q10,q21){\gray50{$b_2$}}
\drawedge[ELside=l,ELpos=50](q11,q20){$b_2$}

\drawedge[ELside=r,ELpos=50, ELdist=1, curvedepth=-5, linegray=.6](q21,q30){}
\drawedge[ELside=r,ELpos=50, ELdist=1, curvedepth=-5, linegray=.6](q30,q21){}
\drawedge[ELside=r,ELpos=50, ELdist=1, curvedepth=-5](q20,q31){}
\drawedge[ELside=r,ELpos=50, ELdist=1, curvedepth=-5](q31,q20){}



\end{picture}
   \label{fig:gwin-counter-noselfloop}
}
\hrule
\caption{Removal of positive attractor in $\G_{{\sf win}} \times [1]$ 
and $\G'_{{\sf win}} \times [2]$, where $\G'_{{\sf win}}$ is the 
variant of $\G_{{\sf win}}$ without a self-loop on $q_2$. \label{fig:gwin-counter-all}}%
\Description{Two (deterministic) game graphs.}
\end{center}
\end{figure}
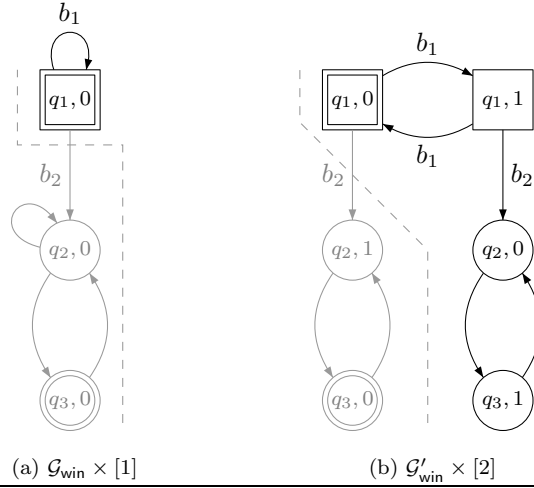

Given an accepting set $U \subseteq T$ that belongs to an SCC with period $p$ 
in $\P(\G)$, we solve the game $\G$ by removing from $\G \times [p]$ 
the attractor $W$ of $U \times \{0\}$, and by recursively solving the subgame 
of $\G \times [p]$ induced by $Q \setminus W$, with target set $T \times \{0\}$.

In $\G_{{\sf win}}$, the set $U = \{q_3\}$ belongs to an SCC of period~$1$,
and its attractor is $W = \{q_2,q_3\}$. After removal of the attractor,
the subgame is winning for player~$1$ (\figurename~\ref{fig:gwin-counter}).
In the variant $\G'_{{\sf win}}$ of $\G_{{\sf win}}$ without a self-loop on $q_2$,
the set $U = \{q_3\}$ is also in an SCC, but with period~$2$. 
After removal of the attractor to $U \times [0]$ in $\G'_{{\sf win}} \times [2]$,
the subgame is not winning for player~$1$ (\figurename~\ref{fig:gwin-counter-noselfloop}).

By Lemma~\ref{lem:game-counter}, solving $\G$ with target set $T$
is equivalent to solving $\G \times [p]$ with target set $T \times \{0\}$.
However, in general combining two almost-sure winning strategies
constructed in two different subgames may not give an almost-sure
winning strategy (if, for instance, one strategy ensures the probability
mass in $T \times \{0\}$ tends to $1$ at even rounds, and the other strategy 
at odd rounds). To establish the correctness of our solution, 
note that all states in $U \times \{0\}$ belong to the (controllable predecessor of the) 
attractor of $U \times \{0\}$ in 
$\G \times [p]$, and since the period of the SCC containing $U$ is $p$, 
in $\P(\G\times [p])$ there is a path from $U \times \{0\}$ to itself 
of length $\ell = k\cdot p$ for all sufficiently large $k$. See the Frobenius
problem~\cite{Kan92} for questions related to computing the largest $k_0$ such 
that there is no such path of length $\ell = k_0\cdot p$. 
It follows that, for $W = \Attr(U \times \{0\},\G \times [p])$,
given a state $\tuple{q,i} \in W$
and an arbitrary length $\ell_0 = i + k\cdot p$ for $k \geq k_0$, 
player~$1$ has a strategy to get eventually synchronized 
in $U \times \{0\} \subseteq T \times \{0\}$ at round $\ell_0$ (where the tracking counter is $0$),
and by the same argument at any round $\ell_1 = \ell_0 + k'\cdot p$ for $k' \geq k_0$,
and so on. That is, for any sequence $i_0,i_1,\dots$ such that 
$i_{j} - i_{j-1} \geq k_0 \cdot p$  (where $i_{-1} = 0$) and $i_j = i \mod p$ 
for all $j\geq 0$,
player~$1$ has a strategy 
from $\tuple{q,i}$ such that for all outcomes $d_0, d_1, \ldots$ of a strategy 
of player~$2$ in $\G \times [p]$, we have
$\liminf_{k \to \infty} d_{i_k}(T \times \{0\}) = 1$ 
(thus also $\limsup_{k \to \infty} d_{i_k}(T \times \{0\}) = 1$).

Intuitively, we can choose the sequence $i_0,i_1,\dots$ in order
to synchronize the probability mass in $W$ with the outcome of the strategy 
constructed in the subgame of $\G \times [p]$ induced by the complement
of $W$.

\begin{lemma}\label{lem:recursive-alg}
Given a set $U$ in a strongly connected component of period $p$ in the subset construction $\P(\G)$
that is accepting ($U \subseteq T$), let $W = \Attr(U \times \{0\},\G \times [p])$
and $\H = \G \times [p] \upharpoonright [Q \times [p] \setminus W]$. We have:
\begin{linenomath*}
$$\winas{weakly}(\G,T) \cap \{1_q \mid q \in Q\}=  \{1_q \mid \exists i: \tuple{q,i} \in 
W \cup \winas{weakly}(\H,T \times \{0\}) \}.$$
\end{linenomath*}
\end{lemma}

\begin{proof}
We establish the claim by showing an inclusion in both ways.
We recall that we generally identify $1_s$ with $s$.
First, we show that if $q \in \winas{weakly}(\G,T)$, then there exists
$0 \leq i \leq p-1$ such that either $\tuple{q,i} \in W$, or 
$\tuple{q,i} \in \winas{weakly}(\H,T \times \{0\})$.
By Lemma~\ref{lem:game-counter}, there exists $0 \leq i \leq p-1$ 
such that $\tuple{q,i} \in \winas{weakly}(\G \times [p], T \times \{0\})$,
and thus if $\tuple{q,i}$ is in the state space of $\H$, then since all choices
of player~$2$ in $\H$ are also possible in $\G \times [p]$, we have 
$\tuple{q,i} \in \winas{weakly}(\H,T \times \{0\})$, and otherwise $\tuple{q,i} \in W$.

Second, we show that if there exists
$0 \leq i \leq p-1$ such that either $\tuple{q,i} \in W$, or 
$\tuple{q,i} \in \winas{weakly}(\H,T \times \{0\})$, then 
$q \in \winas{weakly}(\G,T)$.
To show this, we construct a strategy $\straa$ that is almost-sure 
weakly synchronizing in $T \times \{0\}$ from $\tuple{q,i}$ (in the
game $\G \times [p]$, which entails that player~$1$ is almost-sure
weakly synchronizing in $T$ from $q$ (in the game $\G$).

The construction proceeds by induction on the subgames. Assume that 
we have constructed an almost-sure winning strategy $\straa_{\H}$ 
in $\H$ as well as a sequence of indices $i_0,i_1,\dots$ such that 
for all outcomes $d_0, d_1, \ldots$ of a strategy $\strab$ of player~$2$
in $\H$, we have $\lim_{k \to \infty} d_{i_k}(T \times \{0\}) = 1$
(we say that under $\straa_{\H}$, synchronization is guaranteed to occur 
at indices $i_0,i_1,\dots$).
For the induction step, we construct an almost-sure winning strategy $\straa_{\G}$
in $\G \times [p]$ together with a sequence of indices $j_0,j_1,\dots$ 
where the synchronization is guaranteed to occur. 
Note that $i_{k} - i_{k-1}$ is a multiple of $p$ (for all $k \geq 1$)
and that the counter $c$ is $0$ in any state $\tuple{q',c}$ carrying 
positive probability at round $i_k$. 
The sequence $j_0,j_1,\dots$ is a(ny) sub-sequence of $i_0,i_1,\dots$ 
such that $j_{k} - j_{k-1} \geq k_0 \cdot p$ (where $j_{-1} = 0$).
There is a strategy $\straa_{W}$ for player~$1$ (essentially following 
paths of the appropriate lengths in $\P(\G \times [p])$)
such that given a state $\tuple{q',c}$ carrying probability $\eta$ at round $j$ 
with $j_{k-1} \leq j \leq j_{k}$, for all outcomes $d'_0, d'_1, \ldots$ of a 
strategy $\strab$ of player~$2$, we have 
$\lim_{m \to \infty} d'_{j_{k+m}-j}(T \times \{0\}) = 1$, 
that is synchronization occurs at rounds $j_{k+1}, j_{k+2},\ldots$ from $\tuple{q',c}$.

The strategy $\straa_{G}$ plays according to $\straa_{W}$ whenever the state is
in $W$, and according to the almost-sure winning strategy $\straa_{\H}$ in $\H$
otherwise. Note that $\straa$ is well-defined because once a play 
leaves $\H$ (and enters $W$) it never leaves $W$.
It is easy to see that $\straa$ is almost-sure
weakly synchronizing in $T \times \{0\}$ since, for all strategies
$\strab \in \Strab$ of player~$2$, the outcome sequence from $\tuple{q,i}$
is the sum of a sequence $d_0, d_1, \dots$ of probability distributions
with support in $W$, and a sequence $d'_0, d'_1, \dots$ of probability 
distributions with support in $Q \times [p] \setminus W$ and such that,
for some $\eta \in [0,1]$:
\begin{linenomath*}
$$\lim_{k \to \infty} d_{j_k}(T \times \{0\}) = \eta, \text{ and }$$ 
$$\lim_{k \to \infty} d'_{j_k}(T \times \{0\}) = 1 - \eta,$$
\end{linenomath*}
from which it follows that 
$\limsup_{i \to \infty} d_i(T \times \{0\}) + d'_i(T \times \{0\}) = 1$,
thus player~$1$ is almost-sure weakly synchronizing in $T \times \{0\}$ 
from $\tuple{q,i}$.
\end{proof}

\begin{figure}[!t]
\hrule
\begin{center}
    \begin{picture}(50,50)(0,0)

\node[Nmarks=n, Nmr=0](q)(25,45){$q$}
\node[Nmarks=n](r)(15,25){$r$}
\node[Nmarks=n](s)(5,45){$s$}
\node[Nmarks=r](t)(35,25){$t$}
\node[Nmarks=n](u)(45,45){$u$}

\node[Nmarks=n](x)(15,5){$x$}
\node[Nmarks=r](y)(35,5){$y$}

\drawedge[ELside=r,ELpos=50, ELdist=.5](q,r){$b_1$}
\drawedge[ELside=l,ELpos=50, ELdist=.5](q,t){$b_2$}
\drawedge[ELside=l,ELpos=50, ELdist=.5](r,s){$a_1$}
\drawedge[ELside=r,ELpos=50, ELdist=1](r,x){$a_2$}
\drawedge[ELside=r,ELpos=50, ELdist=.5](t,u){$a_1$}
\drawedge[ELside=l,ELpos=50, ELdist=1](t,y){$a_2$}
\drawedge[ELside=r,ELpos=50, ELdist=1](u,q){}
\drawedge[ELside=r,ELpos=50, ELdist=1](s,q){}

\drawedge[ELside=r,ELpos=50, ELdist=1, curvedepth=5](x,y){}
\drawedge[ELside=r,ELpos=50, ELdist=1, curvedepth=5](y,x){}



\end{picture}
\end{center} 
\hrule
 \caption{A deterministic game. \label{fig:game32}}
 \Description{A deterministic game with 7 states.}
\end{figure}
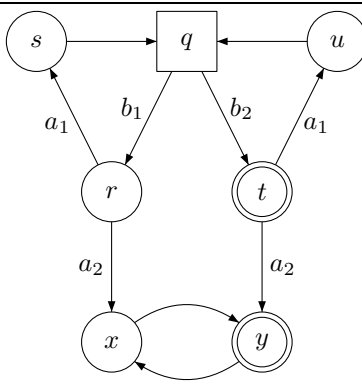

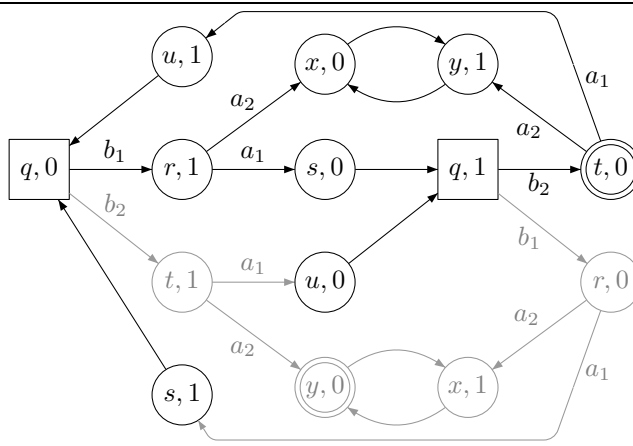
\begin{figure}[!t]
\hrule
\begin{center}
    \begin{picture}(84,60)(0,0)

\node[Nmarks=n, Nmr=0](q0)(4,38){$q,0$}
\node[Nmarks=n, Nmr=0](q1)(61,38){$q,1$}
\node[Nmarks=n,linegray=.6](r0)(80,23){\gray50{$r,0$}}
\node[Nmarks=n](r1)(23,38){$r,1$}
\node[Nmarks=n](s0)(42,38){$s,0$}
\node[Nmarks=n](s1)(23,8){$s,1$}
\node[Nmarks=r](t0)(80,38){$t,0$}
\node[Nmarks=n,linegray=.6](t1)(23,23){\gray50{$t,1$}}
\node[Nframe=n,Nmarks=n](u0)(42,23){$u,0$}   
\node[Nmarks=n](u1)(23,53){$u,1$}

\node[Nmarks=n](x0)(42,52){$x,0$}
\node[Nmarks=n,linegray=.6](x1)(61,9){\gray50{$x,1$}}
\node[Nmarks=r,linegray=.6](y0)(42,9){\gray50{$y,0$}}
\node[Nmarks=n](y1)(61,52){$y,1$}

\drawline[AHnb=0,arcradius=2](75,52)(72.5,59)(30.6,59)(28.7,57.5)
\node[Nframe=n,Nmarks=n, Nh=0, Nw=0, Nmr=0](t0out)(75,52){}
\node[Nframe=n,Nmarks=n, Nh=0, Nw=0, Nmr=0](u1in)(28.7,57.5){}
\drawedge[ELside=r,ELpos=80, ELdist=.5,AHnb=0](t0,t0out){$a_1$}
\drawedge[ELside=r,ELpos=50, ELdist=1](u1in,u1){}

\drawline[AHnb=0,arcradius=2,linegray=.6](75,9)(72.5,2)(26.8,2)(25.85,3.5)
\node[Nframe=n,Nmarks=n, Nh=0, Nw=0, Nmr=0](r0out)(75,9){}
\node[Nframe=n,Nmarks=n, Nh=0, Nw=0, Nmr=0](s1in)(25.85,3.5){}
\drawedge[ELside=l,ELpos=80, ELdist=.5,AHnb=0,linegray=.6](r0,r0out){\gray50{$a_1$}}
\drawedge[ELside=r,ELpos=50, ELdist=1,linegray=.6](s1in,s1){}

\drawedge[ELside=l,ELpos=53, ELdist=1](q0,r1){$b_1$}
\drawedge[ELside=l,ELpos=45, ELdist=.5,linegray=.6](q0,t1){\gray50{$b_2$}}
\drawedge[ELside=r,ELpos=50, ELdist=1,linegray=.6,linegray=.6](r0,x1){\gray50{$a_2$}}
\drawedge[ELside=l,ELpos=50, ELdist=.5](t0,y1){$a_2$}
\drawedge[ELside=r,ELpos=50, ELdist=1](u0,q1){}
\drawedge[ELside=r,ELpos=50, ELdist=1](s0,q1){}
\drawedge[ELside=r,ELpos=50, ELdist=0,linegray=.6](q1,r0){\gray50{$b_1$}}
\drawedge[ELside=r,ELpos=50, ELdist=.5](q1,t0){$b_2$}
\drawedge[ELside=l,ELpos=50, ELdist=1](r1,s0){$a_1$}
\drawedge[ELside=l,ELpos=50, ELdist=.5](r1,x0){$a_2$}
\drawedge[ELside=l,ELpos=50, ELdist=1,linegray=.6](t1,u0){\gray50{$a_1$}}
\drawedge[ELside=r,ELpos=50, ELdist=.5,linegray=.6](t1,y0){\gray50{$a_2$}}
\drawedge[ELside=r,ELpos=50, ELdist=1](u1,q0){}
\drawedge[ELside=r,ELpos=50, ELdist=1](s1,q0){}

\drawedge[ELside=r,ELpos=50, ELdist=1, curvedepth=5](x0,y1){}
\drawedge[ELside=r,ELpos=50, ELdist=1, curvedepth=5](y1,x0){}
\drawedge[ELside=r,ELpos=50, ELdist=1, curvedepth=5,linegray=.6](x1,y0){}
\drawedge[ELside=r,ELpos=50, ELdist=1, curvedepth=5,linegray=.6](y0,x1){}

\node[Nmarks=n](u0)(42,23){}   



\end{picture}
\end{center} 
\hrule
 \caption{The game $\H_{sub}$ computed by Algorithm~\ref{alg:solve-as-ws} 
(line~\ref{alg:solve-as-ws:Hsub}) for the game of \figurename~\ref{fig:game32},
where the shaded region is the set $W$ (line~\ref{alg:solve-as-ws:W}),
and the set $U = \{\tuple{y,0}\}$ is self-recurrent. \label{fig:game32-counter}}
 \Description{A deterministic game with 14 states.}
\end{figure}

Lemma~\ref{lem:recursive-alg} suggests a recursive procedure to compute
the almost-sure winning set for weakly synchronizing objective, shown 
as Algorithm~\ref{alg:solve-as-ws}. We illustrate the execution on
the example of \figurename~\ref{fig:game32}. First the set $U = \{y\}$
is accepting and belongs to an SCC of period $p=2$ (line~\ref{alg:solve-as-ws:period})
in the subset construction (line~\ref{alg:solve-as-ws:SCC}). 
The game $H = \G \times [p]$ with tracking counter modulo $p=2$
is shown in \figurename~\ref{fig:game32-counter}, with the attractor~$W$
to $U \times \{0\} = \{\tuple{y,0}\}$ shaded 
(lines~\ref{alg:solve-as-ws:H}-\ref{alg:solve-as-ws:W}). 
In the induced subgame (line~\ref{alg:solve-as-ws:Hsub}), all states
except $\tuple{x,0}$ and $\tuple{y,1}$ are winning, which is found in the
recursive call (line~\ref{alg:solve-as-ws:recursive}). 
It follows that all states (all Dirac distributions) are winning for player~$1$,
and in fact all distributions that do not contain both $x$ and $y$ in their
support are winning.

We establish the correctness and termination of Algorithm~\ref{alg:solve-as-ws} 
as follows. The correctness straightforwardly follows 
from Lemma~\ref{lem:recursive-alg}, and we show that the depth
of the recursive calls in $Solve(\G_0,T)$ is bounded by the size $\abs{Q_{\G_0}}$ 
of the state space of $\G_0$. This is not immediately obvious, since the size 
of the first argument $\G$ in a recursive call may increase (the game $\H_{sub}$ 
is a subgame of $\H$, which is $p$ times bigger than $\G$). 
However, we claim that an invariant of the execution of $Solve(\G_0,T_0)$ is that, 
in all recursive calls $Solve(\G,T)$, the algorithm only needs to consider 
states of the first argument $\G$ that form a subgame (isomorphic to a subgame) of 
$\G_0 \times [k]$ for some $k$. 
This holds in the initial call (take $k=1$), and if 
$\G$ is a subgame of $\G_0 \times [k]$, then the period $p$ computed at 
line~\ref{alg:solve-as-ws:period} is a multiple of $k$, and therefore 
in all states $\tuple{\tuple{q,i},j}$ in $(\G_0 \times [k]) \times [p]$
the value $j-i \mod k$ is constant along the transitions. 
Given the target states $(T\times \{0\}) \times \{0\}$ 
we only need to consider states $\tuple{\tuple{q,i},j}$ with $j-i = 0 \mod k$,
and we can project $\tuple{\tuple{q,i},j}$ to $\tuple{q,j}$ without loss.
It follows that $\H$ (and also $\H_{sub}$ used in the 
recursive call) can be viewed as a subgame of $\G_0 \times [p]$.
Moreover, the attractor $W$ contains at least one state for every value of 
the tracking counter, and therefore the size of the game $\G$ 
measured as $\max_i \abs{\{q \in Q_{\G_0} \mid \tuple{q,i} \in Q_{\G}\}}$
is strictly decreasing. It follows that there are at most $\abs{Q_{\G_0}}$ recursive calls
in $Solve(\G_0,T)$.
 
Given $0 \leq i \leq p-1$, the \emph{slice} at $i$ of a 
set $W \subseteq Q \times [p]$ is the set 
$\{q \in Q \mid \tuple{q,i} \in W\}$.

\begin{algorithm}[t]
\caption{$Solve(\G,T)$}
\label{alg:solve-as-ws}
{
  \SetKwInOut{Input}{Input}\SetKwInOut{Output}{Output}
 
  \Input{$\G = \tuple{Q, \Act, \delta}$ is a deterministic game, $T \subseteq Q$ is a target set.}
  \Output{The set $\{q \in Q \mid 1_q \in \winas{weakly}(\G,T)\}$.}%
\begin{flushleft}
 \Begin{
                \nl \If{there is an SCC $\C$ of $\P(\G)$ containing a set $U \in \C$ with $U\subseteq T$ \label{alg:solve-as-ws:SCC}}
                {
                   \nl $p \gets \text{period of } \C$  \label{alg:solve-as-ws:period}\;
                   \nl $\H \gets \G \times [p]$ \label{alg:solve-as-ws:H} \;
                   \nl $W \gets \Attr(U \times \{0\},\H)$ \label{alg:solve-as-ws:W}\;
                   \nl $\H_{sub} \gets \H \upharpoonright [Q \times [p] \setminus W]$ \label{alg:solve-as-ws:Hsub}\;
                   \nl \KwRet{ $\{q \in Q \mid \exists i: \tuple{q,i} \in W \cup Solve(\H_{sub},T \times \{0\} \setminus W)\}$} \label{alg:solve-as-ws:recursive} \;
                }
                \Else
                {
                   \nl \KwRet{ $\emptyset$} \;
                }
 }
\end{flushleft}
}
\end{algorithm}

\begin{lemma}\label{lem:alg-det-PSPACE}
Algorithm~\ref{alg:solve-as-ws} computes the almost-sure winning 
Dirac distributions for weakly synchronizing in deterministic games.
It can be implemented in PSPACE.
\end{lemma}

\begin{proof}
The correctness and termination of Algorithm~\ref{alg:solve-as-ws}
have been established above. To show the PSPACE upper bound,
note that if $\G$ is a subgame of $\G_0 \times [k]$, then 
the the period $p$ computed at 
line~\ref{alg:solve-as-ws:period} is at most $2^n \cdot k$,
because slices are subsets of $Q_{\G_0}$. 
Since the depth of recursive calls is at most $n$, 
the game $\H$ constructed at line~\ref{alg:solve-as-ws:H}
is a subgame of $\G_0 \times [p]$ where $p \leq (2^n)^n = 2^{n^2}$.

In a PSPACE implementation of Algorithm~\ref{alg:solve-as-ws},
we can store the game $\G_0$ but not the subgames $\H_{sub}$ of $\G_0 \times [k]$
(constructed at line~\ref{alg:solve-as-ws:Hsub}). 
However, we will show that there is a PSPACE procedure to 
determine the transitions of $\H_{sub}$, namely, given
$\tuple{q,i},\tuple{q',j} \in Q \times [k]$ and $a,b\in \Act$ 
to decide whether $(\tuple{q,i},a,b,1_{\tuple{q',j}})$ 
is a transition in $\H_{sub}$. 

We describe a (N)PSPACE implementation of Algorithm~\ref{alg:solve-as-ws} as follows.
In the first call to the algorithm, at line~\ref{alg:solve-as-ws:SCC} 
we guess the set $U \subseteq T$.
We can check in PSPACE that $U$ belongs to an SCC of $\P(\G)$ (by guessing
the selectors along a path from $U$ to itself) and we can compute its period $p$.
We can construct a PSPACE procedure to check if a state $\tuple{q,i}$ belongs 
to the attractor $W$ of $U \times 0$ in $\G \times [p]$ (analogously,
by guessing the selectors along a path from the singleton $\{\tuple{q,i}\}$ to 
$W$ in $\P(\G)$). A transition $(\tuple{q,i},a,b,1_{\tuple{q',j}})$ is
in $\H_{sub}$ if $j=i-1$ and $(q,a,b,1_{q'})$ is a transition in $\G$ 
and $\tuple{q,i},\tuple{q',j}$ are in the attractor $W$. All these
conditions can be checked in PSPACE.
In the recursive calls (where $\G$ is $\H_{sub}$ from caller), 
we guess a slice of the set $U$ at line~\ref{alg:solve-as-ws:SCC},
and we check that $U$ belongs to an SCC of $\P(\G)$ using our PSPACE
to decide the transitions in $\H_{sub}$. We can do this for all recursive calls
up to depth $n$ by storing the slice of the set $U$ for each recursive call.
\end{proof}

\begin{theorem}\label{theo:alg-det-PSPACE-complete}
The membership problem for almost-sure weakly synchronizing in deterministic games
is PSPACE-complete.
\end{theorem}

\begin{proof}
The PSPACE upper bound is given by Lemma~\ref{lem:alg-det-PSPACE}.
To establish the PSPACE lower bound, we present a reduction from the 
membership problem\footnote{As a side note, we recall that this problem
is the same as emptiness of one-letter alternating automata~\cite{AFA1}.} 
for sure eventually synchronizing in MDPs,
which is PSPACE-complete~\cite[Theorem~2]{DMS19}. 

Given an MDP $\M$ and target set $T$, construct a deterministic game $\G$ as a copy
of $\M$ (illustrated in \figurename~\ref{fig:reduction-det-game})
where each probabilistic choice
in $\M$ becomes a choice for player~$2$ in $\G$: for each $q' \in \delta_{\M}(q,a,-)$,
there is an action $b \in \Act$ such that $\delta_{\G}(q,a,b) = 1_{q'}$ 
(assuming w.l.o.g. that there are sufficiently many actions in $\Act$). 
The alphabet of $\G$ is $\Act \cup \{\sharp\}$ where the action $\sharp$
can be used by player~$1$ from the states in $T$ to visit the new (target)
state $q_{\sharp}$. From the non-target states, playing $\sharp$ leads
to a sink state.

We claim that player~$1$ is sure winning for eventually synchronizing in $T$
from $q_0$ (in $\M$) if and only if player~$1$ is almost-sure winning 
for weakly synchronizing in $\{q_{\sharp}\}$ from $q_0$ (in $\G$).

First, if player~$1$ has a sure-winning strategy for eventually synchronizing in $T$
from $q_0$ (in $\M$), then player~$1$ can play the same strategy in $\G$
followed by playing two times $\sharp$ to visit $q_{\sharp}$ and restart from $q_0$
repeating the same strategy. Hence player~$1$ is almost-sure winning 
for weakly synchronizing in $\{q_{\sharp}\}$ from $q_0$ (in $\G$).

Conversely, if player~$1$ is almost-sure winning 
for weakly synchronizing in $\{q_{\sharp}\}$ from $q_0$ (in $\G$),
then by Lemma~\ref{lem:SCC-somewhere} there exists an SCC $\C$ in $\P(\G)$
containing $\{q_{\sharp}\}$. Given a path $s_0,s_1, \dots, s_k$ (induced by
a sequence of selectors $\alpha_1, \alpha_2, \dots, \alpha_k$)
from $\{q_{\sharp}\}$ to itself in $\P(\G)$ ($s_0 = s_k = \{q_{\sharp}\}$), 
consider the largest index $i<k$ such that $q_{\sharp} \in s_i$.
It is easy to see that the selectors $\alpha_{i+2}, \alpha_{i+3}, \dots, \alpha_{k}$
play only actions in $\Act$, and that they define a sure-winning 
strategy for eventually synchronizing in $T$ from $q_0$ (in $\M$), 
which concludes the proof.
\begin{figure}[t]
\begin{center}
\hrule
\begin{picture}(80,35)(0,0)

\node[Nmarks=n, Nw=33, Nh=18, dash={0.3 0.9}0](m1)(16.5,25){}
\node[Nframe=n](label)(7.5,20){MDP $\M$}
\node[Nmarks=n, Nw=14, Nh=16, dash={0.8 0.9}0](mdp)(25,25){}
\node[Nframe=n](label)(25,20){$T\subseteq Q$}
\node[Nmarks=r](n1)(25,27){$q$}
\node[Nframe=n](label)(14,27){$\dots$}
\node[Nmarks=n](n2)(5,27){$q_0$}   

\node[Nframe=n](arrow)(40,24){{\Large $\Rightarrow$}}


\node[Nmarks=n, Nw=33, Nh=18, dash={0.3 0.9}0](m1)(63.5,25){}
\node[Nmarks=n, Nw=14, Nh=16, dash={0.8 0.9}0](mdp)(72,25){}
\node[Nframe=n](label)(72,20){$T\subseteq Q$}
\node[Nmarks=n](n2)(72,27){$q$}
\node[Nframe=n](label)(61,27){$\dots$}
\node[Nmarks=n](n1)(52,27){$q_0$}   

\node[Nmarks=n](sink)(52,8){$q_{\bot}$} 
\node[Nmarks=r](qq)(72,8){$q_{\sharp}$} 


\drawedge[ELpos=40, ELside=l](n1,sink){$\sharp$}
\drawedge[ELpos=60, ELside=l, curvedepth=7](n2,qq){$\sharp$}
\drawline[AHnb=1, arcradius=2](69.17,5.17)(66,2)(42,2)(42,17)(49.17,24.17)
\node[Nframe=n](label)(36,9){$\Act \! \cup \!\{\sharp\}$}




\drawloop[ELpos=57, ELside=l, NLdist=-2, loopCW=y, loopangle=35, loopdiam=5](sink){$\Act \! \cup \!\{\sharp\}$}

\end{picture}
\hrule
\caption{Sketch of the reduction to show PSPACE-hardness of the membership problem
for almost-sure weakly synchronizing.}\label{fig:reduction-det-game}
\Description{An MDP on the left, a deterministic game on the right.}
\end{center}
\end{figure}
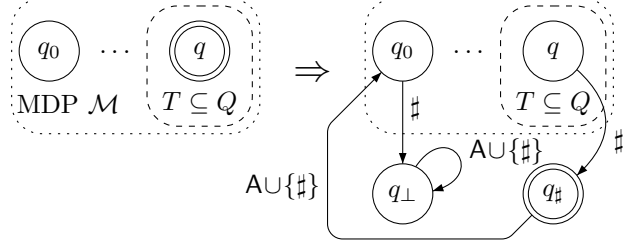
\end{proof}

\subsection{Weakly synchronizing in stochastic games}\label{sec:stoch-games}

We present an algorithm to compute the almost-sure winning region
for weakly synchronizing objectives in stochastic games,
which generalizes the result of Section~\ref{sec:det-games}.
This algorithm has the flavor of the algorithm for
deterministic games, with additional complications due
to the probabilistic transitions in the game. The proof
is also more technical because we no longer assume 
that pure counting strategies are sufficient 
for player~$1$ (but we show that such strategies are indeed always
sufficient for almost-sure winning).

Recall that throughout this section we consider a stochastic game 
$\G = \tuple{Q, \Act, \delta}$ and we denote by $n = \abs{Q}$ the size of the state space, 
and by $\eta$ the smallest positive probability in the transitions of $\G$.
We consider the almost-sure weakly synchronizing objective defined
by a set $T \subseteq Q$ of accepting states.

Given a set $U \subseteq Q$, consider the sequence $U_i = \CPre^i(U)$
for $i \geq 1$ (and $U_0 = U$). Since $U_i \subseteq Q$, this sequence 
is ultimately periodic. Consider the least $k \geq 0$ for which there exists 
$r > 1$ such that $U_{k} = U_{k+r}$, and consider the least such $r$,
called the \emph{period}. It is easy to see that  $k,r \leq 2^n$.
For $R = U_{k}$ we call $\tuple{R,r,k}$ the \emph{periodic scheme} of $U$
and we refer to its elements as $\R(U) = U_k$, $\r(U) = r$, and $\k(U) = k$.
The set $U$ is \emph{self-recurrent} if $U \neq \emptyset$ and there exists 
an index $0\leq t < r$ such that all states in $U \times \{t\}$ are almost-sure 
winning for the (state-based) reachability objective $\Diamond(R \times \{0\})$ in $\G \times [r]$.
The intuitive meaning of being self-recurrent appears in Lemma~\ref{lem:SCC-somewhere-gen} below.
Note that in deterministic games~$\G$, an accepting set $U \subseteq T$ 
contained in a strongly connected component $\C$ of $\P(\G)$ 
is self-recurrent. Self-recurrent sets are the key to generalize 
the result of Lemma~\ref{lem:SCC-somewhere} to stochastic games. 
The argument of the proof is more involved, and presented in the 
following three lemmas.


\begin{lemma}\label{lem:SCC-somewhere-gen}
If there exists a self-recurrent set $U \subseteq T$, 
then 
there exists a state from which player~$1$ is almost-sure winning 
for weakly synchronizing in $T$.
\end{lemma}

\begin{proof}
Let $\tuple{R,r,k}$ be the periodic scheme of $U$.
We show that if all states  in $U \times \{t\}$
are almost-sure winning for the (state-based) reachability objective
$\Diamond(R \times \{0\})$ in $\G \times [r]$, 
then we can construct a strategy $\straa_{\as}$ for player~$1$ in $\G$ that 
is almost-sure winning for weakly synchronizing in $U$ (and thus also in $T$)
from all distributions with support in $U$.

The key argument is to show that for all distributions $d$ on $U$, 
for all $\epsilon > 0$, there exists a strategy $\straa_{\epsilon}$ 
for player~$1$ that ensures from $d \times \{t\}$, 
against all strategies of player~$2$, that after finitely many rounds, 
a distribution $d' \times \{t'\}$ is reached such that 
$(i)$  $d'(U) > 1 - \epsilon$, and 
$(ii)$ $d' \times \{t'\}$ is almost-sure winning for the reachability objective 
$\Diamond(R \times \{0\})$ in $\G \times [r]$.
From this key argument, we can construct a strategy $\straa_{\as}$ for player~$1$ that 
successively plays according to the strategies $\straa_{\frac{1}{2}}$, $\straa_{\frac{1}{4}}$,
$\straa_{\frac{1}{8}}$,\dots and show that it is almost-sure weakly synchronizing in $U$
(and thus also in $T$) from the distribution $d$, which concludes the proof.

To prove the key argument, we construct $\straa_{\epsilon}$ as follows (see 
also~\cite[Theorem~7]{DMS19}).
Since $d \times \{t\}$ is almost-sure winning for the reachability objective 
$\Diamond(R \times \{0\})$ in $\G \times [r]$, 
there exists a (pure memoryless) strategy $\straa$ and 
an integer $h_{\epsilon}$ such that for all strategies 
$\strab$ of player~$2$, we have 
$\Prb_{q_0}^{\straa,\strab}(\Diamond^{\leq h_{\epsilon}}\, R \times \{0\}) 
 \geq 1-\epsilon$ (Lemma~\ref{lem:almost-sure-reach-game-pl1}).

We construct the strategy $\straa_{\epsilon}$ to play according to $\straa$ 
as long as no state in $R \times \{0\}$ is reached.
Whenever a state in $R \times \{0\}$ is reached,
if it happens within the first $h_{\epsilon}$ rounds,
then the strategy $\straa_{\epsilon}$ plays to reach again $R \times \{0\}$ 
after $r$ more rounds (which is possible since $R = \CPre^r(R)$); 
if it happens after $h_{\epsilon}$ rounds,
then it plays according to a sure-winning strategy for eventually synchronizing 
in $U$ from $R$ (thus for $k$ more steps, where $k$ is such that $R = \CPre^k(U)$). 
It immediately follows that 
for all strategies $\strab$ of player~$2$, after finitely many rounds  
(at most $h_{\epsilon} + r + k$ rounds)
a distribution~$d' \times \{t'\}$ is reached such that 
$(i)$ $d'(U) \geq 1 - \epsilon$, and 
$(ii)$ $d' \times \{t'\}$ is almost-sure winning for the reachability objective 
$\Diamond(R \times \{0\})$ in $\G \times [r]$ (recall that for a reachability objective, 
a distribution is almost-sure winning if 
all states in its support are almost-sure winning). 
\end{proof}

Note that the strategy $\straa_{\as}$ constructed in the proof 
of Lemma~\ref{lem:SCC-somewhere-gen} is pure and counting.

For the converse of Lemma~\ref{lem:SCC-somewhere-gen}, the structure of the
argument is similar to the proof of Lemma~\ref{lem:SCC-somewhere} for deterministic games
and pure strategies.
However, the technical details are more involved due to stochasticity (in the game graph, and
in the strategy of player~$1$).
For convenience, we separate the structure of the argument and the technical details.

\smallskip\noindent{\em Structure of the argument (substitution game).}
We present the \emph{substitution game}, loosely inspired by player substitutions 
in ice hockey: for simplicity there is a single player on ice 
and a coach who manages $(i)$ a main team of players (those who have been on ice at least once
during the game) initially empty, and $(ii)$ a reserve team of $K$ players (fresh 
players who have never been on ice). 
We consider the following protocol for substitutions of players.
A player stays on ice for a one-minute period, then needs to rest 
and to be substituted: he returns to the main team, and a player from 
the main team (possibly the same player) is selected by the coach to go on ice. 
At each one-minute period, 
the players in the main team have 
the possibility to \emph{pass} (e.g., if they are tired), which they can
do whenever they want, not necessarily on consecutive periods but at most a fixed 
number $N$ of times. When a player returns from ice
to the main team, he gets recharged with $N$ pass. 
The next player to go on ice is chosen, at the discretion of the coach,
among the players of the main team who did not pass.
If at some period all players in the main team pass, 
a player from the reserve team is called to go on ice instead (reserve players 
cannot pass), and will integrate the main
team after a one-minute period, thereby getting $N$ possibilities to pass.
Initially, the main team is empty. The question is how many players
should there be initially in the reserve team, in order for the coach 
to be able to play the substitution game forever regardless of when the players 
decide to pass, and what is an optimal strategy for the coach (assuming there 
is no player on ice initially). 

Formally, a configuration of the substitution game consists of a set $M$,
initially $M = \emptyset$, and a function $f: M \to \nat$ that maps
each member of the main team to their number of pass. 
Given a configuration $\tuple{M,f}$, the players choose a set $P \subseteq \{p \in M \mid f(p) \geq 1\}$
of players who will pass, 
and then the coach chooses a player $p$ such that
$p \in M \setminus P$ if $P \neq M$, and $p = \abs{M}$ if $P = M$.
The next configuration is $\tuple{M',f'}$ where $M' = M \cup \{p\}$
and $f'$ is defined by $f'(p) = N$, and for all $m \in M \setminus \{p\}$:
\begin{linenomath*}
$$f'(m) = \begin{cases} f(m) - 1 & \text{if } m \in P, \\
                                    f(m) &     \text{if } m \not\in P. \\ \end{cases}$$
\end{linenomath*}
The game is won by the coach if there exists $K$ (the number of reserve players)
such that the game continues forever, with $\abs{M} \leq K$ in all configurations.

It is easy to see that $K \geq N + 1$ is necessary for the coach to win (e.g.,
if all players always pass, then $K=N$ players is not sufficient). 
To show that $K = N + 1$ is sufficient, consider the strategy of the
coach that, given the current configuration $\tuple{M,f}$, 
chooses a player $p$ in $M \setminus P$ with largest number of pass (if $P \neq M$),
that is $p \in \argmax_{x \in M \setminus P} f(x)$.
Under this strategy, the following invariant holds in all configurations $\tuple{M,f}$ during 
the game: for $k= \abs{M}$, let $p_{0}, p_{1}, \dots, p_{k-1}$ be the players in $M$ 
in ascending order according to $f$ (i.e., if $i \leq j$, then $f(p_{i}) \leq f(p_{j})$),
then $f(p_i) \leq N+1-k+i$. The proof is by induction, where the initial 
configuration satisfies the invariant trivially, and given the invariant
holds in a configuration $\tuple{M,f}$ with $k= \abs{M}$, we consider two
cases: $(1)$ if $P = M$, then all players in $M$ are passing, thus 
$\abs{M'} = k+1$ in the next configuration $\tuple{M',f'}$ and
$f'(p_i) =  f(p_i) - 1 \leq (N+1-k+i)-1 = N+1-(k+1)+i$ for $p_i \in M$ (in ascending order),
and for the new player $p_k$ in position $k$, we verify that $f'(p_k) = N \leq N+1-(k+1)+k$;
$(2)$ if $P \neq M$, then $\abs{M'} = \abs{M} = k$ in the next configuration $\tuple{M',f'}$,
and for the players $p_i$ on the left of the player chosen by the coach, 
their position in $M'$ is the same as in $M$ and their value of $f$ does not increase,
and for the players $p_i$ on the right of the player chosen by the coach, 
their position in $M'$ and their value of $f$ decrease by $1$, 
thus in both cases $f'(p_i) \leq N+1-k+i$ holds, while for the player chosen by the 
coach, its new position is $k-1$ and we verify that $f'(p_{k-1}) = N \leq N+1-k+(k-1)$.

It follows from this invariant that $f(p_0) \leq N+1-k$ for the first player 
in ascending order, and $k \leq N+1$ since $f(p_0) \geq 0$, that is each
configuration has at most $N+1$ players in $M$, showing that $N+1$ reserve
players are sufficient.

In the proof of Lemma~\ref{lem:SCC-somewhere}, the players of the substitution 
game were the registers, and the number $N$ of allowed pass was $N = 2^n-1$
(which is an upper bound on the number of positions that player~$2$ could not
cover). 
We concluded that $K = 2^n$ registers were sufficient, which allowed us to 
define a spoiling strategy for player~$2$ using $K$ tokens, each representing
a probability mass of $\frac{1}{K}$.

\smallskip\noindent{\em Technical aspects.}
For stochastic games, we use a structure of argument based on 
the substitution game to show the
converse implication of Lemma~\ref{lem:SCC-somewhere-gen}. 
In Lemma~\ref{lem:epsilon-weakly} we construct a number $\epsilon_w > 0$ 
and a bound $N_w$ on
the number of positions that player~$2$ may not be able to cover
against an arbitrary strategy of player~$1$,
where covering a position informally means that a probability mass 
$\epsilon_w > 0$ is in $Q \setminus T$ at that position.



\begin{lemma}\label{lem:epsilon-weakly}
Let $\G$ be a stochastic game.
There exists $\epsilon_w > 0$ and $N_w \in \nat$ such that the following holds:
if there exists no set $U \subseteq T$ that is self-recurrent, 
then in $\G$ for all player-$1$ strategies $\straa$, 
for all but at most $N_w$ rounds $i$, 
there exists a player-$2$ strategy $\strab$
such that $\G^{\straa,\strab}_{i}(T) \leq 1 - \epsilon_w$.
\end{lemma}

\begin{proof}
By the conditions of the lemma and by Lemma~\ref{lem:almost-sure-reach-game}, 
given a set $U \subseteq T$ with periodic scheme $\tuple{R,r,k}$ and 
given $0\leq t < r$, there exists a strategy $\strab^{U}_t$ 
for player~$2$ in $\H = \G \times [r]$ 
(with initial distribution $d_0$ such that $\Supp(d_0) = U \times \{t\}$),
such that for all player-$1$ strategies $\straa$, we have 
$\H^{\straa,\strab^{U}_t}_{i}(R \times \{0\}) \leq 1 - \eta_0 \cdot \eta^{n \cdot 2^n}$
where $\eta_0 = \min \{ d_0((q,t)) \mid q \in U \}$ is the smallest positive probability 
in~$d_0$.

\smallskip\noindent{\em Construction of the strategy $\strab^U$.}
Each strategy $\strab^{U}_t$ maps to a strategy $\strab$ in $\G$ via the
bijection $\mu_{t}$ (namely, $\strab(\rho) = \strab^{U}_t(\mu_t(\rho))$, 
see p.\pageref{ref:bijection}), 
and by an abuse of notation we also denote by 
$\strab^{U}_t$ the corresponding strategy $\strab$ in $\G$. Thus 
$\strab^{U}_t$ plays in $\G$ as if the initial index (in $\G \times [r]$)
was $t$. Define the strategy $\strab^U = \sum_{t=0}^{\r(U)-1} \frac{1}{\r(U)} \cdot \strab^{U}_t$
as the uniform superposition of the strategies $\strab^{U}_t$ for $t=0,\dots,\r(U)-1$.
Notice that if $d_0(q) \geq \eta_0$ for all $q \in U$, 
then from $d_0$ the strategy $\strab^U$ ensures, against all player-$1$ strategies $\straa$, 
that $\G^{\straa,\strab^{U}}_i(Q \setminus R) \geq \frac{\eta_0 \cdot \eta^{n \cdot 2^n}}{2^n}$ 
for all $i \geq 0$, since $\r(U) \leq 2^n$.

\smallskip\noindent{\em Construction of the set $I_w$.}
We define the numbers $\epsilon_w > 0$ and $N_w \in \nat$ required
in the lemma. Let $N_w = 4^n$ and 
$\epsilon_w = \frac{1}{2n}\cdot \left(\frac{\eta^{(n+1) \cdot 2^n}}{n\cdot 4^n}\right)^{2^n}$.

To prove the lemma, fix an arbitrary player-$1$ strategy $\straa$ in $\G$, and we construct 
a set $I_w$ of positions such that $\abs{I_w} \leq N_w$, 
and for all $i\in \nat \setminus I_w$, there exists a player-$2$ strategy $\strab$
such that $\G^{\straa,\strab}_{i}(T) \leq 1 - \epsilon_w$.

We construct the set $I_w$ iteratively as follows. 
We construct a sequence $I = I_0, I_1, \dots$ of (finite) sets $I_i \subseteq \nat$ 
of indices (initially, $I_0 = \emptyset$) and a sequence $\U = \U_0, \U_1, \dots$
of sets $\U_i \subseteq 2^T$ of nonempty subsets of accepting states (initially, $\U_0 = \emptyset$). 
The sequence $\U$ will be strictly increasing $\U_0 \subsetneq \U_1 \subsetneq \dots$,
and therefore the construction has (at most) $2^n$ iterations. 
Each iteration $k$ corresponds to a round, numbered $\hat{\imath}_k$, 
and defines the sets $I_k$ and $\U_k$. 

At iteration $k$ ($k=1,2,\dots$), we consider the round $\hat{\imath}_{k-1}$ 
defined from the previous iteration (initially, $\hat{\imath}_0 = 0$), and 
either the construction terminates at iteration $k$, 
or we define a new index $\hat{\imath}_k > \hat{\imath}_{k-1} + 2^n$.
We associate with index $\hat{\imath}_k$ a set of sub-distributions\footnote{A 
\emph{sub-distribution} on~$S$ is a function $d : S \to [0, 1]$ 
such that $\sum_{s \in S} d(s) \leq 1$.}
that 
sum up to the distribution $d_{\hat{\imath}_k}$ in the outcome of the game
$\G$ (from initial distribution $d_0$) under the player-$1$ strategy $\straa$ 
and some player-$2$ strategy $\strab$.
The sub-distributions associated with index $\hat{\imath}_k$ 
are described below (where $K = 2^n$), see also \figurename~\ref{fig:hitting-strategy}:
\begin{itemize}
\item for each $U \in \U_{k}$, we have $(K-k)$ identical sub-distributions $f^U$
(which we call a \emph{$U$-token});

\item we have a sub-distribution $g_k$ storing the remaining probability mass.
\end{itemize}

All tokens for a given set $U$
were created at the iteration, say $j_U$, where $U$ was added to $\U$.
At the end of iteration $j_U$, each token carries probability $\epsilon_{j_U}$
in each state of $U$
and the tokens are updated according to the strategy $\strab^{U}$, which ensures
that at every step after $j_U$, a bounded probability mass (namely, 
$\epsilon_{j_U} \cdot \frac{\eta^{n \cdot 2^n}}{2^n}$)
lies in $Q \setminus \R(U)$ from where, by definition of $\R(U)$, 
player~$2$ has the possibility, using the spoiling strategy against sure 
eventually synchronizing in $U$, to 
inject a fraction $\eta^{\k(U)}$ of the probability mass into $Q \setminus U$ 
after $\k(U)$ steps. Since $\k(U) \leq 2^n$, for all $j \geq \hat{\imath}_k + 2^n$,
there is a player-$2$ strategy to ensure bounded probability in $Q \setminus U$
at step $j$ (namely, $\epsilon_{j_U} \cdot \frac{\eta^{(n+1) \cdot 2^n}}{2^n}$).
Then the state in $Q \setminus U$ with largest probability mass, carries at
least $\frac{\epsilon_{j_U}}{n} \cdot \frac{\eta^{(n+1) \cdot 2^n}}{2^n}$.
If that state is not in $T$, then the index $j$ is covered.
However, if that state is in $T$ (and if this is the case for all $U \in \U_{k}$), 
then the index $j$ is not covered, and we insert $j$ into $I_w$. 
On the other hand, if player~$2$ does inject probability in $Q \setminus U$ in this way, 
we can show that a new set $U$ (i.e., $U \not\in \U_{k}$) carries a bounded probability.
Therefore, the situation where an index $j$ is not covered cannot occur more
than $2^n$ times (which is an upper bound on the cardinality of the sets in 
the sequence $\U$). Note that when player~$2$ does inject probability in $Q \setminus U$,
one token is ``consumed'', and therefore the number 
of tokens for each $U$ will decrease (by $1$) at each iteration. 
We create sufficiently many tokens for each $U$ (namely $K = 2^n$ tokens)
to avoid depletion.

\begin{figure}[!t]
\begin{center}
    \input{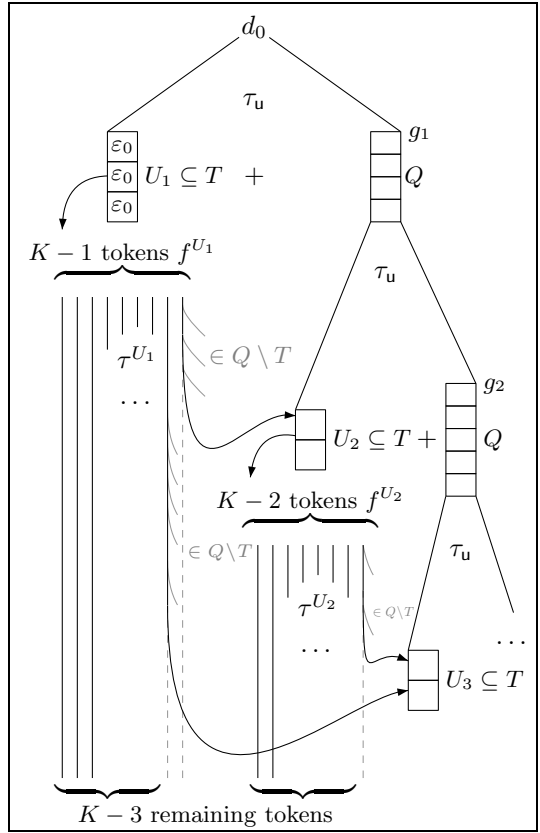}
\end{center} 
 \caption{The construction of the strategy $\strab_w$ as the limit of the sequence $\strab_0,\strab_1,\dots$
 (proof of Lemma~\ref{lem:epsilon-weakly}). \label{fig:hitting-strategy}}
 \Description{A tree rooted at $d_0$ showing which strategy of player~$2$ is used, and a flow of tokens.}
\end{figure}

Initially we have $\U_0 = \emptyset$ and $g_0 = d_0$ is the (full)
initial distribution. 
At iteration $k$, given index $\hat{\imath}_{k-1}$ we first construct the index $\hat{\imath}_k$ and 
then the sub-distributions associated with~$\hat{\imath}_k$. The construction
is illustrated in \figurename~\ref{fig:hitting-strategy}.

The first iteration ($k=1$) of the construction is slightly 
different from the other iterations because $\U_0$ is empty. 
Consider the outcome $\G^{\straa,\strab_{\u}} = d_0, d_1, \dots$
from the initial distribution $g_0 = d_0$ and the uniform strategy $\strab_{\u}$.
Let $\hat{\imath}_1$ be the smallest index $j \geq 0$ such that
the set $T \cap \{q \mid d_j(q) \geq \epsilon_{0}\}$ is nonempty,
where $\epsilon_0 = \frac{1}{2n}$.
We consider two cases: $(1)$ if no such index $j$ exists, then 
$d_j(T) \leq n \cdot \epsilon_0 = \frac{1}{2} \leq \epsilon_w$, for all $j \geq 0$,
hence we can take $I_w = \emptyset$ and the construction terminates;
$(2)$ otherwise, let 
$U_1 = T \cap \{q \mid d_{\hat{\imath}_1}(q) \geq \epsilon_{0}\} \neq \emptyset$
and decompose the sub-distribution $d_{\hat{\imath}_1}$ as
$d_{\hat{\imath}_1} = (K-1) \cdot f^{U_1} + g_{1}$
where $f^{U_1}(q) = \frac{\epsilon_0}{K-1}$ if $q \in U_{1}$, 
and $f^{U_1}(q) = 0$ otherwise (and let $g_{1} = d_{\hat{\imath}_1} - (K-1) \cdot f^{U_1}$), 
thus $f^{U_1}$ and $g_{1}$ are sub-distributions. 
We update $I$ and $\U$ as follows: 
let $I_1 := I_0 \cup [\hat{\imath}_1,\hat{\imath}_1+2^n]$ 
and $\U_{1} := \U_{0} \cup \{U_{1}\}$. Let $\epsilon_1 = \frac{\epsilon_0}{K-1}$
be the probability mass in each state of $U_1$ according to $f_{U_1}$.

Now we present iteration $k$ for $k\geq 2$.
For each $U \in \U_{k-1}$,
pick one of the $K-k+1$ tokens of iteration $k-1$ 
(i.e., a sub-distribution $f^U$ associated with index $\hat{\imath}_{k-1}$ 
that carried probability mass at least $\epsilon_{k-1}$
in each state of $U$ when $U$ was added to $\U$) 
and for each $j \geq i_{k-1} + 2^n$, consider the strategy $\strab^U$ played
until step $j-\k(U)$, then the spoiling strategy for sure eventually 
synchronizing in $U$. The strategy $\strab^U$ would ensure
probability at least $\frac{\epsilon_{k-1} \cdot \eta^{n \cdot 2^n}}{2^n}$
in $Q \setminus \R(U)$ at every step, thus in particular at step $j-\k(U)$,
and the spoiling strategy can inject probability at least 
$\beta \cdot \epsilon_{k-1}$ in $Q \setminus U$ (at step $j$)
where $\beta = \eta^{2^n} \cdot \frac{\eta^{n \cdot 2^n}}{2^n}  
\leq  \eta^{\k(U)} \cdot \frac{\eta^{n \cdot 2^n}}{2^n}$.

Hence under that strategy, 
there is a state $q^U_j \in Q \setminus U$ that contains probability at 
least $\frac{\beta}{n} \cdot \epsilon_{k-1}$. 
We consider two cases:
$(1)$ if for all $j \geq \hat{\imath}_{k-1} + 2^n$ there exists $U \in \U_{k-1}$
such that $q^U_j \not\in T$, then the construction is over and we define 
$I_w = I_{k-1}$;  
$(2)$ otherwise, let $\hat{\imath}_k$ be the smallest index $j \geq \hat{\imath}_{k-1} + 2^n$ 
such that $q^U_j \in T$
for all $U \in \U_{k-1}$, and consider the sub-distribution $d$ at step $\hat{\imath}_k$
that originates from the superposition of playing $\strab_{\u}$ from $g_{k-1}$
and playing, from each $U$-token for $U \in \U_{k-1}$ the strategy $\strab^U$ 
followed (at step $j-\k(U)$ for one of the $U$-tokens) by the spoiling strategy for sure eventually 
synchronizing in $U$.
Let $U_k = T \cap \{q \mid d(q) \geq \frac{\beta}{n}  \cdot \epsilon_{k-1}\} \neq \emptyset$ 
be the set of accepting states that carry 
a sufficiently significant probability mass according to $d$. 
Note that $q^U_j \in U_k$ for all $U \in \U_{k-1}$,
and since $q^U_j \not\in U$, it follows that $U_k \neq U$ for all $U \in \U_{k-1}$,
that is $U_k \not\in \U_{k-1}$.
Decompose the sub-distribution $d$ as
$d = (K-k) \cdot f^{U_k} + g_{k}$
where $f^{U_k}(q) = \frac{\beta}{n}  \cdot \frac{\epsilon_{k-1}}{K-k}$ if $q \in U_{k}$, 
and $f^{U_k}(q) = 0$ otherwise (and let $g_{k} = d - (K-k) \cdot f^{U_k}$), 
thus $f^{U_k}$ and $g_{k}$ are sub-distributions. 

With step $\hat{\imath}_k$, we associate the sub-distribution $g_k$,
the $K-k$ copies of the sub-distribution $f^{U_k}$, and 
for each $U \in \U_{k-1}$ the $K-k$ remaining $U$-tokens updated according 
to $\strab^U$ (see \figurename~\ref{fig:hitting-strategy}).
We update $I$ and $\U$ as follows: 
let $I_k := I_{k-1} \cup [\hat{\imath}_k,\hat{\imath}_k+2^n]$ 
and $\U_{k} := \U_{k-1} \cup \{U_{k}\}$. Let $\epsilon_k = 
\frac{\beta}{n}  \cdot \frac{\epsilon_{k-1}}{K-k}$.

The construction must terminate after at most $2^n$ iterations
because the cardinality of the sets $\U_i$ is increasing by $1$
at each iteration (in fact $\abs{\U_i} = i$), and $\U_i \subseteq 2^T$ 
thus $\abs{\U_i}$ is bounded by $2^n$.
We note that $I_{k}$ is the union of $k$ intervals of size $2^n$,
thus $\abs{I_{k}} \leq k \cdot 2^n \leq 4^n$.
If the construction terminates before the first iteration is complete,
then we have already shown that the result of the lemma holds.
When the construction terminates, say during iteration $k_{\max} < K = 2^n$,
then for all $k \leq k_{\max}$, considering the two cases, 
in all rounds $j$ between $\hat{\imath}_{k-1} + 2^n$ and $\hat{\imath}_k$ 
and in all rounds $j$ after round $\hat{\imath}_{k_{\max}} + 2^n$,
player~$2$ can inject probability 
$\frac{\beta}{n} \cdot \frac{\epsilon_{k-1}}{K} \geq \epsilon_K$ in $q^{U}_j$ 
for $U=U_{k}$, thus outside $T$, which concludes the proof 
with $N_w = 4^n$ and $\epsilon_w = \epsilon_K > \epsilon_0 \cdot (\frac{\beta}{n\cdot K})^{2^n} = 
\frac{1}{2n}\cdot \left(\frac{\eta^{(n+1) \cdot 2^n}}{n\cdot 4^n}\right)^{2^n}$.
\end{proof}

In the proof of Lemma~\ref{lem:epsilon-weakly}, the set $I_w$ 
of cardinality $N_w$ corresponds to the positions that player~$2$
does not cover from an initial distribution $d_0$ (where he would
pass in the substitution game). 
Note the order of the 
quantifiers in the statement of Lemma~\ref{lem:epsilon-weakly}:
player~$2$ may use different strategies to cover different
positions. We use the structure of argument of the substitution game
to show that a single strategy of player~$2$ can cover all but finitely
many positions, and we obtain the generalization of Lemma~\ref{lem:SCC-somewhere}
to stochastic games.

\begin{lemma}\label{lem:SCC-somewhere-gen-equiv}
Let $\G$ be a stochastic game. The following equivalence holds:
there exists a self-recurrent set $U \subseteq T$, 
if and only if, 
there exists a state from which player~$1$ is almost-sure winning 
for weakly synchronizing in $T$.
\end{lemma}

\begin{proof}
One direction of the lemma is given by Lemma~\ref{lem:SCC-somewhere-gen}.
For the converse direction, the proof uses Lemma~\ref{lem:epsilon-weakly}
and is similar to the proof of Lemma~\ref{lem:SCC-somewhere}. 

Given an arbitrary strategy $\straa$ for player~$1$, by Lemma~\ref{lem:epsilon-weakly} 
there exists an index $i_0$ (namely, $i_0 = \max I_w$ the largest 
index in the finite set $I_w$ constructed in the proof of Lemma~\ref{lem:epsilon-weakly})
such that for all $i\geq i_0$, there exists a player-$2$ strategy $\strab^i$
such that $d_i = \G^{\straa,\strab^i}_{i}(T) \leq 1 - \epsilon_w$. The strategies
$\strab^i$ correspond to the reserve team of tokens (in the substitution game). 

Now for each $d_i$ (as an initial distribution), considering the strategy $\straa$
of player~$1$ played from $d_i$, we also get from 
Lemma~\ref{lem:epsilon-weakly} a finite set $I_{d_i}$ of at most $N_w$ indices,
which are the times associated with $d_i$ where player~$2$ will pass in the 
substitution game. As we know that at most $N_w + 1$ tokens from the reserve team
may be needed in the substitution game, there is a time $i^*_0$ after which 
no token will ever be transferred from the reserve team and let $K$ be the number of
token in the main team of tokens at time $i^*_0$. We assign probability mass $\frac{1}{K}$
to each token. At every step $i \geq i^*_0$,
there is a token that does not pass, and therefore there exists a strategy $\strab$
transforming that token into a sub-distribution that carries probability at least
$\frac{\epsilon_w}{K}$ in $Q \setminus T$ at time $i$. The resulting sub-distribution
(at time $i$) has again from Lemma~\ref{lem:epsilon-weakly} a finite set of indices
of size $N_w$ 
where it can pass, as in the substitution game. Hence the process can continue
forever, and at every step $i \geq i^*_0$, there is a probability mass at least
$\frac{\epsilon_w}{K}$ in $Q \setminus T$, showing that $\straa$ is not 
almost-sure winning for weakly synchronizing in $T$.
\end{proof}

\begin{algorithm}[t]
\DontPrintSemicolon
\caption{$Solve(\H,p,T)$}
\label{alg:solve-as-ws-stoch}
{
  \SetKwInOut{Input}{Input}\SetKwInOut{Output}{Output}
 
  \Input{$\G = \tuple{Q, \Act, \delta}$ is a stochastic game, $T \subseteq Q$ is a target set.}
  \Output{The set $\{\Supp(d) \mid d \in \winas{weakly}(\G,T)\}$.}%
\begin{flushleft}
 \Begin
 {
     \nl $r \gets 1$ \label{alg:solve-as-ws-stoch-initr}\;
     \nl $\H \gets G \times [r]$ \label{alg:solve-as-ws-stoch-inith} \;
     \nl $S \gets Q \times [r]$ \label{alg:solve-as-ws-stoch-inits} \;

     \nl $K \gets S$ \label{alg:solve-as-ws-stoch-inits0} \;

     \nl \label{alg:solve-as-ws-stoch-loop}\Repeat{$K = \emptyset$} 
     {

     \nl \If{there is a self-recurrent set $U \subseteq T$ in 
             $\H\upharpoonright [K]$  \label{alg:solve-as-ws-stoch-main-test}}
     {\nl Let $\tuple{R,r,k}$ be the periodic scheme of $U$  \;
             \nl $\H \gets \G \times [r]$ \label{alg:solve-as-ws-stoch-newH}\;
             \nl $K,S \gets \expand(r,K,S)$ \label{alg:solve-as-ws-stoch-expand}\;
             \nl Let $W$ be the almost-sure winning \\region for the (state-based)
                 reachability \\ objective $\Diamond(R \times \{k \!\mod r\})$ in $\H\upharpoonright [K]$ \label{alg:solve-as-ws-stoch-W0}\;
             \nl $X \gets \PosAttr_1(W, \H\upharpoonright [K])$ \label{alg:solve-as-ws-stoch-posAttr1} \;
             \nl $K \gets K \setminus X$ \label{alg:solve-as-ws-stoch-posAttr2} \; 
     }
     \label{alg:solve-as-ws-stoch-else} \Else 
     {
             \nl $L \gets \PosAttr_2(K, \H\upharpoonright [S])$ \label{alg:solve-as-ws-stoch-L1}\;
             \nl $S \gets S \setminus  L$ \label{alg:solve-as-ws-stoch-SL}\;
             \nl $K \gets S$ \label{alg:solve-as-ws-stoch-L2}\;
     }
     }
     \nl \KwRet{$\{s \subseteq Q \mid \exists i: s \times \{i\} \subseteq S\}$} \label{alg:solve-as-ws-stoch-return}\;
 } 
\end{flushleft}
}
\end{algorithm}

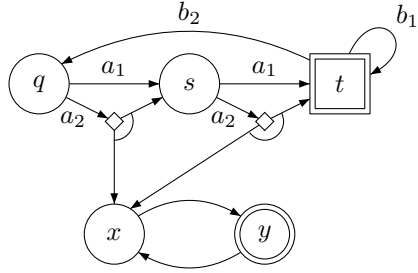
\begin{figure}
\begin{center}
  \hrule
  \begin{picture}(52,42)(0,-3)

\node[Nmarks=n](q)(5,25){$q$}
\node[Nmarks=n](r)(25,25){$s$}
\node[Nmarks=r, Nmr=0](s)(45,25){$t$}

\rpnode[Nmarks=n](qx)(15,20)(4,1.2){}
\drawarc[linegray=0](15,20,2.5,270,26.565)
\rpnode[Nmarks=n](rx)(35,20)(4,1.2){}
\drawarc[linegray=0](35,20,2.5,216.87,26.565)

\node[Nmarks=n](x)(15,5){$x$}
\node[Nmarks=r](y)(35,5){$y$}

\drawedge[ELside=l,ELpos=50, ELdist=1](q,r){$a_1$}
\drawedge[ELside=l,ELpos=50, ELdist=1](r,s){$a_1$}
\drawedge[ELside=r,ELpos=50, ELdist=1, curvedepth=-6, syo=1,eyo=1](s,q){$b_2$}

\drawedge[ELside=r,ELpos=55, ELdist=.5](q,qx){$a_2$}
\drawedge[ELside=r,ELpos=55, ELdist=.5](r,rx){$a_2$}
\drawedge[ELside=r,ELpos=50, ELdist=1](qx,x){}
\drawedge[ELside=r,ELpos=50, ELdist=1](qx,r){}
\drawedge[ELside=r,ELpos=50, ELdist=1](rx,s){}
\drawedge[ELside=r,ELpos=50, ELdist=1, eyo=2](rx,x){}

\drawedge[ELside=r,ELpos=50, ELdist=1, curvedepth=4.5](x,y){}
\drawedge[ELside=r,ELpos=50, ELdist=1, curvedepth=4.5](y,x){}


\drawloop[ELside=l,loopCW=y, loopangle=45, loopwidth=5, loopheight=4](s){$b_1$}

\end{picture}
  \hrule
  \caption{A stochastic game $\G$. \label{fig:game32stoch}}
  \Description{A stochastic game with 5 states.}
\end{center}
\end{figure}

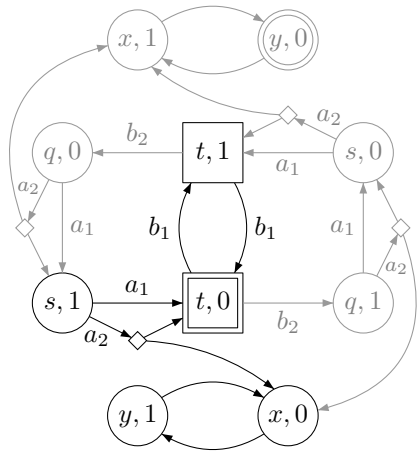
\begin{figure}
\begin{center}
  \hrule
  \begin{picture}(50,62)(0,-1)

\node[Nmarks=n,linegray=.6](q0)(5,40){\gray50{$q,0$}}
\node[Nmarks=n,linegray=.6](q1)(45,20){\gray50{$q,1$}}
\node[Nmarks=n,linegray=.6](r0)(45,40){\gray50{$s,0$}}
\node[Nmarks=n](r1)(5,20){$s,1$}
\node[Nmarks=r, Nmr=0](s0)(25,20){$t,0$}
\node[Nmarks=n, Nmr=0](s1)(25,40){$t,1$}

\rpnode[Nmarks=n,linegray=.6](qx0)(0,30)(4,1.2){}
\rpnode[Nmarks=n,linegray=.6](rx0)(35,45)(4,1.2){}

\rpnode[Nmarks=n,linegray=.6](qx1)(50,30)(4,1.2){}
\rpnode[Nmarks=n](rx1)(15,15)(4,1.2){}

\node[Nmarks=n](x0)(35,5){$x,0$}
\node[Nmarks=r,linegray=.6](y0)(35,55){\gray50{$y,0$}}
\node[Nmarks=n,linegray=.6](x1)(15,55){\gray50{$x,1$}}
\node[Nmarks=n](y1)(15,5){$y,1$}

\drawedge[ELside=l,ELpos=50, ELdist=1,linegray=.6](q0,r1){\gray50{$a_1$}}
\drawedge[ELside=l,ELpos=50, ELdist=1,linegray=.6](r0,s1){\gray50{$a_1$}}
\drawedge[ELside=r,ELpos=50, ELdist=1,linegray=.6](s0,q1){\gray50{$b_2$}}

\drawedge[ELside=l,ELpos=50, ELdist=1,linegray=.6](q1,r0){\gray50{$a_1$}}
\drawedge[ELside=l,ELpos=50, ELdist=1](r1,s0){$a_1$}
\drawedge[ELside=r,ELpos=50, ELdist=1,linegray=.6](s1,q0){\gray50{$b_2$}}

\drawedge[ELside=l,ELpos=50, ELdist=1, curvedepth=4.5](s0,s1){$b_1$}
\drawedge[ELside=l,ELpos=50, ELdist=1, curvedepth=4.5](s1,s0){$b_1$}

\drawedge[ELside=r,ELpos=55, ELdist=.2,linegray=.6](q0,qx0){\gray50{{\footnotesize $a_2$}}}
\drawedge[ELside=r,ELpos=55, ELdist=.5,linegray=.6](r0,rx0){\gray50{$a_2$}}
\drawedge[ELside=r,ELpos=55, ELdist=.2,linegray=.6](q1,qx1){\gray50{{\footnotesize $a_2$}}}
\drawedge[ELside=r,ELpos=55, ELdist=.5](r1,rx1){$a_2$}
\drawedge[ELside=r,ELpos=50, ELdist=1, curvedepth=9,linegray=.6](qx0,x1){}
\drawedge[ELside=r,ELpos=50, ELdist=1,linegray=.6](qx0,r1){}
\drawedge[ELside=r,ELpos=50, ELdist=1,linegray=.6](rx0,s1){}
\drawedge[ELside=r,ELpos=50, ELdist=1, curvedepth=2, eyo=-2,linegray=.6](rx0,x1){}
\drawedge[ELside=r,ELpos=50, ELdist=1, curvedepth=9,linegray=.6](qx1,x0){}
\drawedge[ELside=r,ELpos=50, ELdist=1,linegray=.6](qx1,r0){}
\drawedge[ELside=r,ELpos=50, ELdist=1](rx1,s0){}
\drawedge[ELside=r,ELpos=50, ELdist=1, curvedepth=2, eyo=2](rx1,x0){}

\drawedge[ELside=r,ELpos=50, ELdist=1, curvedepth=4.5](x0,y1){}
\drawedge[ELside=r,ELpos=50, ELdist=1, curvedepth=4.5,linegray=.6](y0,x1){}
\drawedge[ELside=r,ELpos=50, ELdist=1, curvedepth=4.5,linegray=.6](x1,y0){}
\drawedge[ELside=r,ELpos=50, ELdist=1, curvedepth=4.5](y1,x0){}



\end{picture}
  \hrule
  \caption{The game $\H = \G \times [2]$ for the game $\G$ of \figurename~\ref{fig:game32stoch}
with the shaded region $X$ computed by Algorithm~\ref{alg:solve-as-ws-stoch} 
(line~\ref{alg:solve-as-ws-stoch-posAttr1}).  \label{fig:game32stoch-counter}}
  \Description{A stochastic game with 10 states, among which 5 states are shaded.}
\end{center}
\end{figure}

We present Algorithm~\ref{alg:solve-as-ws-stoch} to compute the almost-sure 
winning set for weakly synchronizing objectives.
We use the game of \figurename~\ref{fig:game32stoch} for illustration.
The game contains a self-recurrent set $U = \{y\}$ (with period $2$),
thus player~$1$ is winning from $x$ and from $y$.
Player~$1$ is also winning from the other states: the mass of probability
that eventually stays in $t$ is winning, and the remaining mass of probability
can be injected in $U$ by player~$1$ at the correct times to be synchronized
modulo the period~$2$, thanks to the consecutive transitions on $a_2$
from $q$ and from $s$.

Given the game $\G$ as input, the algorithm considers subgames of 
$\H = \G \times [r]$ where $r=1$ initially 
(lines~\ref{alg:solve-as-ws-stoch-initr}-\ref{alg:solve-as-ws-stoch-inith}),
with state space $S$, initially $S = Q \times [1]$ (line~\ref{alg:solve-as-ws-stoch-inits}).
The working variable $K$ (line~\ref{alg:solve-as-ws-stoch-inits0}) 
is used to compute losing states for player~$1$.

The algorithm proceeds iteratively to construct $K$, by removing states from $S$.
In the loop of line~\ref{alg:solve-as-ws-stoch-loop}, 
as long as there is a self-recurrent set $U$
in the subgame $\H\upharpoonright[K]$, we expand the game $H$ to
track the number of rounds modulo the period of $U$ (which must be a multiple of $r$). 
Given a set $S \subseteq Q \times [p]$, and a period $r$ that is a multiple of $p$,
the \emph{$r$-expansion} of $S$ is the set 
$\{\tuple{q,i} \mid 0 \leq i \leq r-1 \land \tuple{q,i \mod p} \in S\}$.
The $\expand$ function computes the $r$-expansion of $S$ 
and $K$ at line~\ref{alg:solve-as-ws-stoch-expand} (where 
$p = \abs{\{i \mid \tuple{q,i} \in S\}}$ can be derived from the set $S$).
In the expanded game $\H$ (line~\ref{alg:solve-as-ws-stoch-newH}),
we compute the almost-sure winning region~$W$ for the (state-based) reachability 
objective $\Diamond(R \times \{k\})$,
which we call a \emph{core winning region},
where $k = \k(U) \!\!\mod r$. From the states in~$W$ player~$1$ 
is almost-sure weakly synchronizing in $T \times \{0\}$ (see the proof 
of Lemma~\ref{lem:SCC-somewhere-gen}). \figurename~\ref{fig:game32stoch-counter}
shows the $2$-expansion of the game of \figurename~\ref{fig:game32stoch}.
The value $k$ is such that from $R \times \{k\}$ player~$1$ can 
inject all the probability mass into $T \times \{0\}$ (in fact, in~$U \times \{0\}$).

The next iteration starts after removing from the state space $K$ 
the positive attractor for player~$1$ to $W$ (lines~\ref{alg:solve-as-ws-stoch-posAttr1}-\ref{alg:solve-as-ws-stoch-posAttr2}), 
thus ensuring $\H\upharpoonright[K]$ is again a subgame (the dark part of \figurename~\ref{fig:game32stoch-counter}).
Whenever there is no set $U$ in $\H\upharpoonright[K]$
satisfying the conditions of Lemma~\ref{lem:SCC-somewhere-gen-equiv},
the whole state space $K$ is losing for player~$1$ and we remove
its positive attractor for player~$2$ 
(lines~\ref{alg:solve-as-ws-stoch-L1}-\ref{alg:solve-as-ws-stoch-L2}).
This part of the algorithm is illustrated in the game of \figurename~\ref{fig:game11stoch}, 
where the self-recurrent set $U=\{y\}$
(with period $1$) induces a core winning region $W = \{x,y\}$, 
and in the subgame obtained by removing the positive attractor for
player~$1$ to $W$, the set $U = \{q\}$ is self-recurrent. 
The remaining subgame with state space $\{s\}$ has no self-recurrent
set, thus we remove $s$ and its positive attractor $\{q,s\}$ for player~$2$,
showing that $q$ and $s$ are losing for player~$1$.

The loop (line~\ref{alg:solve-as-ws-stoch-loop})
terminates when $K = \emptyset$, which can happen
if either the state space $S$ can be partitioned by positive attractors 
to core winning regions, and then the whole state space $S$ is winning
for player~$1$, or if all states in $S$ are losing ($L = S$), and then
the winning region for player~$1$ is empty.
The algorithm then returns the slices of the winning region, which correspond
to the support of the winning distributions.
\figurename~\ref{fig:iterations} illustrates the first iterations of the algorithm.

\begin{figure}[!t]
\begin{center}
\hrule
    \newcommand{\token}{{\LARGE $\cdot$}}
\newcommand{\minitoken}{{\scriptsize $\cdot$}}
\newcommand{\btoken}{\green{\token}}
\newcommand{\rtoken}{\red{\token}}
\newcommand{\bminitoken}{\green{\minitoken}}
\newcommand{\rminitoken}{\red{\minitoken}}

\begin{picture}(82,67)(0,0)

\gasset{Nh=2,Nw=1.6, Nmr=0, Nframe=y}


\drawpolygon[fillgray=.9, Nframe=n](50,53)(50,65)(26,65)(26,53)(29.2,52.5)(34.8,53)(38,53.7)(41.2,54)(46,52.5)
\drawpolygon[fillgray=.9, Nframe=n](82,53)(82,65)(58,65)(58,53)(61.2,52.5)(66.8,53)(70,53.7)(73.2,54)(78,52.5)

\drawline[AHnb=0,linegray=.6](12,35)(12,65)
\drawline[AHnb=0,linegray=.6](16,35)(16,65)

\drawline[AHnb=0,linegray=.6](38,35)(38,65)

\drawline[AHnb=0,linegray=.6](54,9)(54,25)

\drawcurve[AHnb=0,linegray=.6](8,43)(12.4,43.5)(15.6,42)(20,43)
\drawcurve[AHnb=0,linegray=.6](8,40)(11.6,39)(14,38.5)(20,38)

\drawcurve[AHnb=0,linegray=.6](26,43)(30.4,43.5)(33.6,42)(38,43) 
\drawcurve[AHnb=0,linegray=.6](26,40)(29.6,39)(32,38.5)(38,38) 
\drawcurve[AHnb=0,linegray=.6](38,43)(42.4,43.5)(45.6,42)(50,43) 
\drawcurve[AHnb=0,linegray=.6](38,40)(41.6,39)(44,38.5)(50,38) 

\drawcurve[AHnb=0,linegray=.6](26,53)(29.2,52.5)(34.8,53)(41.2,54)(46,52.5)(50,53)
\drawcurve[AHnb=0,linegray=.6](26,49)(30,48)(38,47)(50,46)

\drawcurve[AHnb=0,linegray=.6](58,53)(61.2,52.5)(66.8,53)(73.2,54)(78,52.5)(82,53) 
\drawcurve[AHnb=0](58,51)(62,49)(66,48)(70,48.4)(74,48)(78,49.5)(82,51)

\drawcurve[AHnb=0](0,25)(4,23)(8,22)(12,22.4)(16,22)(20,23.5)(24,25) 
\drawline[AHnb=0](0,25)(0,9)(24,9)(24,25)

\drawcurve[AHnb=0](30,25)(34,23)(38,22)(42,22.4)(46,22)(50,23.5)(54,25) 
\drawcurve[AHnb=0](54,25)(58,23)(62,22)(66,22.4)(70,22)(74,23.5)(78,25) 
\drawline[AHnb=0](30,25)(30,9)(78,9)(78,25) 

\drawcurve[AHnb=0,linegray=.6](30,16)(36.4,15.5)(47.6,17)(60.4,16)(70,15.5)(78,16)
\drawcurve[AHnb=0,linegray=.6](30,13)(38,12)(54,11.5)(78,11)

\drawline[AHnb=0](30,25)(30,9)(78,9)(78,25) 

\node[Nmarks=n, Nh=30, Nw=4](game)(2,50){}   
\node[Nmarks=n, Nframe=n](label)(6,50){{\small $\Rightarrow$}}
\node[Nmarks=n, Nh=30, Nw=12](game)(14,50){} 
\node[Nmarks=n, Nframe=n](label)(22.8,50){{\small $\Rightarrow$}}
\node[Nmarks=n, Nh=30, Nw=24](game)(38,50){} 
\node[Nmarks=n, Nframe=n](label)(54,50){{\small $\Rightarrow$}}
\node[Nmarks=n, Nh=30, Nw=24](game)(70,50){} 

\node[Nmarks=n, Nframe=n](label)(26.8,17){{\small $\Rightarrow$}}
\node[Nmarks=n, Nframe=n](label)(81.2,17){{\small $\dots$}}

\node[Nmarks=n, Nframe=n](q1)(2,33){}
\node[Nmarks=n, Nframe=n](q2)(12,33){}
\node[Nmarks=n, Nframe=n](q3)(16,33){}
\node[Nmarks=n, Nframe=n](q4)(38,33){}

\drawedge[ELside=r, curvedepth=-2](q1,q2){{\scriptsize $\expand(3,\cdot)$}}
\drawedge[ELside=r, curvedepth=-2.5](q3,q4){{\scriptsize $\expand(6,\cdot)$}}

\node[Nmarks=n, Nframe=n](q5)(12,7){}
\node[Nmarks=n, Nframe=n](q6)(54,7){}

\drawedge[ELside=r, curvedepth=-3](q5,q6){{\scriptsize $\expand(12,\cdot)$}}

\node[Nmarks=n, Nframe=n](label)(10.16,37){{\scriptsize $W$}}
\node[Nmarks=n, Nframe=n](label)(18,40){{\scriptsize $X$}}

\node[Nmarks=n, Nframe=n](label)(28.4,45.5){{\scriptsize $W$}}
\node[Nmarks=n, Nframe=n](label)(43.6,49.5){{\scriptsize $X$}}

\node[Nmarks=n, Nframe=n](label)(66,60){{\scriptsize $K_1$}}
\node[Nmarks=n, Nframe=n](label)(73.2,51){{\scriptsize $L_1$}}

\node[Nmarks=n, Nframe=n](label)(32,11){{\scriptsize $W$}}
\node[Nmarks=n, Nframe=n](label)(46,14){{\scriptsize $X$}}





\end{picture}
\hrule
\end{center} 
 \caption{View of the iterations of Algorithm~\ref{alg:solve-as-ws-stoch}. \label{fig:iterations}}
 \Description{Evolution of the blocks of state space constructed by Algorithm~\ref{alg:solve-as-ws-stoch}.}
\end{figure}
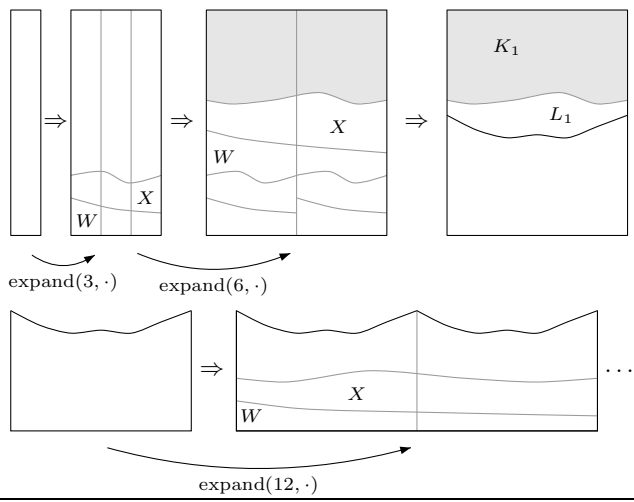

\smallskip\noindent{\em Termination.}
The termination of Algorithm~\ref{alg:solve-as-ws-stoch} 
is established by showing that $K \subseteq S$ is an invariant
of the repeat-loop and that at every iteration, either $(i)$ the size of $S$ is 
strictly decreasing, or $(ii)$ the size of $S$ is unchanged, and the size
of $K$ is strictly decreasing, where the \emph{size} of $S$ (and $K$)
is defined by 
$\max_i \abs{\{q \in Q \mid \tuple{q,i} \in S\}}$ (and similarly for $K$).
Note that the $\expand$ function may increase the cardinality of $S$ and $K$,
but not their size (line~\ref{alg:solve-as-ws-stoch-expand}).
To show $(i)$ and $(ii)$, if the condition of the test (line~\ref{alg:solve-as-ws-stoch-main-test})
holds, then $W$ is nonempty and contains at least one state for each value 
of the tracking counter, thus the size of $K$ decreases (and $S$ remains unchanged); 
otherwise, the size of $S$ decreases (line~\ref{alg:solve-as-ws-stoch-SL}) 
since $K$ is nonempty at the beginning of each iteration (the loop
terminates if $K = \emptyset$). 
It follows that the number of executions of the main 
loop (line~\ref{alg:solve-as-ws-stoch-loop}) is at most $n^2$.

\begin{figure}[!b]
\hrule
\begin{center}
    \begin{picture}(40,45)(0,-5)

\node[Nmarks=r, Nmr=0](q)(5,25){$q$}
\node[Nmarks=n](r)(35,25){$s$}
\node[Nmarks=n](x)(5,5){$x$}
\node[Nmarks=r](y)(35,5){$y$}

\rpnode[Nmarks=n](qx)(20,25)(4,1.2){}
\drawarc[linegray=0](20,25,2.5,233.13,360)

\rpnode[Nmarks=n](yy)(20,0)(4,1.2){}
\drawarc[linegray=0](20,0,2.5,177,250)

\drawedge[ELside=l,ELpos=50, ELdist=1](q,qx){$b_2$}
\drawedge[ELside=l,ELpos=50, ELdist=1](qx,r){}
\drawedge[ELside=l,ELpos=50, ELdist=1](qx,x){}

\drawedge[ELside=l,ELpos=50, ELdist=1, curvedepth=5](x,y){$a_2$}
\drawedge[ELside=r,ELpos=50, ELdist=1, curvedepth=1.4](y,yy){}
\drawedge[ELside=r,ELpos=50, ELdist=1, curvedepth=1.4](yy,x){}
\drawbpedge[ELside=r,ELpos=50, ELdist=1](yy,240,10,y,250,6){}


\drawloop[ELside=l,loopCW=y, loopangle=105, loopwidth=5, loopheight=5](x){$a_1$}
\drawloop[ELside=l,loopCW=y, loopangle=90, loopwidth=5, loopheight=5](q){$b_1$}
\drawloop[ELside=l,loopCW=y, loopangle=90, loopwidth=5, loopheight=5](r){}

\end{picture}
\end{center} 
\hrule
 \caption{A stochastic game. \label{fig:game11stoch}}
 \Description{A stochastic game with 4 states.}
\end{figure}
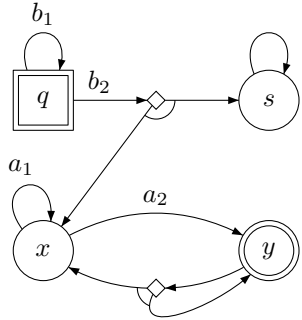

\smallskip\noindent{\em Correctness.}
Let $\G$ be a game with state space $Q$, let $T \subseteq Q$.

In the main loop of Algorithm~\ref{alg:solve-as-ws-stoch} 
(line~\ref{alg:solve-as-ws-stoch-loop}), let $L_1, L_2,\dots, L_{i_{\ell}}$ 
be the sequence of values of the variable $L$ successively computed in the 
else-clause (line~\ref{alg:solve-as-ws-stoch-L1}) and
let $X_1, X_2, \dots, X_{i_{x}}$ be the values of the variable $X$ successively 
computed in the then-clause \emph{after} the last execution of the else-clause 
(line~\ref{alg:solve-as-ws-stoch-posAttr1}).
Each of these two sequences may be empty, but not both.
The sets $L_i$ and $X_i$ are possibly computed with different values 
of $r$. In the sequel, we consider their $p$-expansion for the value $p$ of $r$
when the algorithm terminates, which leads to the property that 
in the game $\H = \G \times [p]$, the sets $\expand(p,L_i)$ ($i=1,\dots,i_{\ell}$) 
and $\expand(p,X_j)$ ($j=1,\dots,i_{x}$) form a partition of the state 
space $Q \times [r]$ (\figurename~\ref{fig:winning-losing-decomposition}).
We often omit the $\expand$ function and write $L_i$ instead of $\expand(p,L_i)$
(and analogously for $X_j$).

We refine this partition with a decomposition of the losing states $L_i$,
and of the winning states $X_i$ (\figurename~\ref{fig:winning-losing-refined-decomposition}).
There exist sets $K_i$, $F_i$ for $i=1,\dots,i_{i_{\ell}}$, and $F_{i_{\ell}+1}$
such that $F_1 = Q \times [r]$ and for all $1 \leq i \leq i_{\ell}$:
\begin{itemize}
\item $K_i$ is a trap for player~$1$ in $\H \upharpoonright [F_i]$, 

\item $L_i = \PosAttr_2(K_i, \H \upharpoonright [F_i])$,

\item $F_{i+1} = F_i \setminus L_i$.
\end{itemize}

The set $K_i$ is the value of variable $K$ (up to $p$-expansion) when $L_i$ 
is computed (at line~\ref{alg:solve-as-ws-stoch-L1}), and the value $F_i$
is the value of variable $S$ computed in the next line (line~\ref{alg:solve-as-ws-stoch-SL}). 
On the other hand, 
there exist nonempty sets $U_i,R_i,S_i,W_i$ for $i=1,\dots, i_{x}$, and $U_{i_{x}+1}$ such that 
$S_1 = F_{i_{\ell}+1}$ and for all $1 \leq i \leq i_{x}$:
\begin{itemize}
\item $U_i$ is a self-recurrent set in $\H \upharpoonright[S_i]$,
with periodic scheme $\tuple{R_i,r_i,k_i}$ where $r_i = p$
(since $\H$ is the $p$-expansion of $\G$, thus the least possible period 
in a subgame of $\H$ is $p$).

\item $X_i = \PosAttr_1(W_i, \H\upharpoonright [S_i])$ where $W_i$ is the 
almost-sure winning region for reachability to $R_i \times \{k_i\}$ in 
$\H \upharpoonright[S_i]$,

\item $S_{i+1} = S_i \setminus X_i$.

\end{itemize}

\begin{figure}
\centering
\begin{minipage}{.45\textwidth}
  \centering
  \begin{picture}(40,60)(0,0)

\gasset{Nh=2,Nw=2, Nmr=0, Nframe=y}

\node[Nmarks=n, Nh=60, Nw=40](game)(20,30){}   

\drawcurve[AHnb=0,linegray=.6](0,52)(15,53)(35,50)(40,52)
\node[Nmarks=n, Nframe=n](label)(32,53){{\small $L_1$}}

\drawcurve[AHnb=0](0,40)(7,39)(20,41)(35,38)(40,40)
\node[Nmarks=n, Nframe=n](label)(32,42){{\small $L_2$}}

\drawcurve[AHnb=0,linegray=.6](0,25)(10,25.5)(20,25.3)(30,24.7)(40,25)
\node[Nmarks=n, Nframe=n](label)(20,33){{\small $X_3$}}

\drawcurve[AHnb=0,linegray=.6](0,15)(10,15.5)(20,14.3)(30,13.7)(36,14)
\drawline[AHnb=0,linegray=.6](36,14)(40,14)
\node[Nmarks=n, Nframe=n](label)(20,20){{\small $X_2$}}

\node[Nmarks=n, Nframe=n](label)(20,6){{\small $X_1$}}

\node[Nmarks=n, Nh=60, Nw=40](game)(20,30){}  





\end{picture}
  \captionof{figure}{The sets $L_i$ of losing states, and $X_j$ of winning states in $\G \times [r]$
computed by Algorithm~\ref{alg:solve-as-ws-stoch} (see \emph{Correctness}). \label{fig:winning-losing-decomposition}}
\end{minipage}%
\hfill
\begin{minipage}{.45\textwidth}
  \centering
  \begin{picture}(40,60)(0,0)

\gasset{Nh=2,Nw=2, Nmr=0, Nframe=y}

\node[Nmarks=n, Nh=60, Nw=40](game)(20,30){}   

\drawline[AHnb=0,linegray=.4](4,0)(4,-1)
\drawline[AHnb=0,linegray=.4](8,0)(8,-1)
\drawline[AHnb=0,linegray=.4](12,0)(12,-1)
\drawline[AHnb=0,linegray=.4](16,0)(16,-1)
\drawline[AHnb=0,linegray=.4](20,0)(20,-1)
\drawline[AHnb=0,linegray=.4](24,0)(24,-1)
\drawline[AHnb=0,linegray=.4](28,0)(28,-1)
\drawline[AHnb=0,linegray=.4](32,0)(32,-1)
\drawline[AHnb=0,linegray=.4](36,0)(36,-1)

\node[Nmarks=n, Nframe=n](label)(38,-2){{\scriptsize $0$}}
\node[Nmarks=n, Nframe=n](label)(34,-2){{\scriptsize $1$}}
\node[Nmarks=n, Nframe=n](label)(18.2,-2.4){{\scriptsize $\dots$}}
\node[Nmarks=n, Nframe=n](label)(0,-2){{\scriptsize $r\!-\!1$}}

\drawline[AHnb=0,linegray=.6](0,56)(40,56)
\drawcurve[AHnb=0,linegray=.6](0,52)(15,53)(35,50)(40,52)
\node[Nmarks=n, Nframe=n](label)(4,58){{\small $K_1$}}
\node[Nmarks=n, Nframe=n](label)(32,53){{\small $L_1$}}

\drawline[AHnb=0,linegray=.6](0,46)(40,46)
\drawcurve[AHnb=0](0,40)(7,39)(20,41)(35,38)(40,40)
\node[Nmarks=n, Nframe=n](label)(4,49){{\small $K_2$}}
\node[Nmarks=n, Nframe=n](label)(32,42){{\small $L_2$}}

\drawcurve[AHnb=0,linegray=.6](0,39)(12,36)(40,33)
\node[Nmarks=n, Nframe=y, Nh=6, Nw=4, linegray=.6](U3)(38,29){{\small $U_3$}}
\node[Nmarks=n, Nframe=y, Nh=4, Nw=4, linegray=.6](R3)(26,29){{\small $R_3$}}
\drawcurve[AHnb=0,linegray=.6](0,25)(10,25.5)(20,25.3)(30,24.7)(40,25)
\node[Nmarks=n, Nframe=n](label)(20,38){{\small $X_3$}}
\node[Nmarks=n, Nframe=n](label)(6,32){{\small $W_3$}}

\drawcurve[AHnb=0,linegray=.6](0,24)(12,22)(40,19)
\node[Nmarks=n, Nframe=y, Nh=4, Nw=4, linegray=.6](U2)(38,16){{\small $U_2$}}
\node[Nmarks=n, Nframe=y, Nh=4, Nw=4, linegray=.6](R2)(30,17){{\small $R_2$}}
\node[Nmarks=n, Nframe=y, Nh=1, Nw=1, linegray=.6](q)(13,18){}
\rpnode[Nmarks=n, Nframe=y, linegray=.6](p)(18,16)(4,.5){}
\drawedge[ELpos=50, ELside=l, linegray=.6, curvedepth=0,AHdist=0.7,AHangle=20,AHLength=.75,AHlength=0.7](q,p){}
\node[Nmarks=n, Nframe=n, Nh=1, Nw=1, linegray=.6](r1)(18,12){}
\node[Nmarks=n, Nframe=n, Nh=1, Nw=1, linegray=.6](r2)(25,14){}
\node[Nmarks=n, Nframe=n, Nh=1, Nw=1, linegray=.6](r3)(29,27){}
\drawedge[ELpos=50, ELside=l, linegray=.6, curvedepth=0,AHdist=0.7,AHangle=20,AHLength=.75,AHlength=0.7](p,r1){}
\drawedge[ELpos=50, ELside=l, linegray=.6, sxo=.4, curvedepth=1,AHdist=0.7,AHangle=20,AHLength=.75,AHlength=0.7](p,r2){}
\drawedge[ELpos=50, ELside=l, linegray=.6, sxo=-.4, curvedepth=-2,AHdist=0.7,AHangle=20,AHLength=.75,AHlength=0.7](p,r3){}
\drawline[AHnb=0, linewidth=0.07](14.7,16.6)(15.7,17.6)
\drawline[AHnb=0, linewidth=0.07](14.7,17.6)(15.7,16.6)

\rpnode[Nmarks=n, Nframe=y, linegray=.6](p)(32,10)(4,.5){}
\node[Nmarks=n, Nframe=n, Nh=1, Nw=1, linegray=.6](r1)(32,4){}
\node[Nmarks=n, Nframe=n, Nh=1, Nw=1, linegray=.6](r2)(36,9){}
\node[Nmarks=n, Nframe=n, Nh=1, Nw=1, linegray=.6](r3)(34.5,18){}
\drawedge[ELpos=50, ELside=l, linegray=.6, curvedepth=0,AHdist=0.7,AHangle=20,AHLength=.75,AHlength=0.7](p,r1){}
\drawedge[ELpos=50, ELside=l, linegray=.6, curvedepth=0,AHdist=0.7,AHangle=20,AHLength=.75,AHlength=0.7](p,r2){}
\drawedge[ELpos=50, ELside=l, linegray=.6, curvedepth=0,AHdist=0.7,AHangle=20,AHLength=.75,AHlength=0.7](p,r3){}

\drawcurve[AHnb=0,linegray=.6](0,15)(10,15.5)(20,14.3)(30,13.7)(36,14)
\node[Nmarks=n, Nframe=n](label)(20,23){{\small $X_2$}}
\node[Nmarks=n, Nframe=n](label)(6,19){{\small $W_2$}}

\drawcurve[AHnb=0,linegray=.6](0,14)(12,10)(20,8.5)(40,8)
\node[Nmarks=n, Nframe=y, Nh=5, Nw=4, linegray=.6](U1)(38,2.5){{\small $U_1$}}
\node[Nmarks=n, Nframe=y, Nh=6, Nw=4, linegray=.6](R1)(18,4){{\small $R_1$}}
\node[Nmarks=n, Nframe=n](label)(20,11){{\small $X_1$}}
\node[Nmarks=n, Nframe=n](label)(6,5){{\small $W_1$}}

\node[Nmarks=n, Nh=60, Nw=40](game)(20,30){}  





\end{picture}
  \captionof{figure}{Refined view of the sets $L_i$ of losing states, and $X_j$ of winning 
states in $\G \times [r]$ (from \figurename~\ref{fig:winning-losing-decomposition}). 
\label{fig:winning-losing-refined-decomposition}}
\end{minipage}
\Description{View of some blocks of state space forming a partition.}
\end{figure}

Note that $S_{i_{x}+1} = \emptyset$ since the sets $\{L_i\}_{1 \leq i \leq i_{\ell}}$ and 
$\{X_j\}_{1 \leq j \leq i_{x}}$ form a partition of the state 
space $Q \times [p]$.
The values $U_i,R_i,W_i$ correspond to the variables $U,R,W$ (up 
to $p$-expansion) at the iteration where $X_i$ is computed 
(lines~\ref{alg:solve-as-ws-stoch-main-test}-\ref{alg:solve-as-ws-stoch-posAttr2}).
The set $S_i$ is the value of variable $K$ at the beginning of that iteration.
Note that in the subgame $\H\upharpoonright [S_i]$ player~$2$ cannot play the
actions for which positive probability 
leaves $S_i$ (\figurename~\ref{fig:winning-losing-refined-decomposition}).

We now show that in the game $\H = \G \times [p]$, every state in $\bigcup_i L_i$ 
is losing, and every state in $\bigcup_j X_j$ is winning, for almost-sure weakly 
synchronizing in $T \times \{0\}$. 

For states in $\bigcup_i L_i$, the claim
follows from Lemma~\ref{lem:SCC-somewhere-gen-equiv} and the fact that
$K_i$ is a trap for player~$1$ in $\H \upharpoonright [F_i]$. 
In $L_i \setminus K_i$, player~$2$ uses the attractor strategy, 
and in $K_i$, for every strategy of player~$1$ there is a strategy of player~$2$
to avoid almost-sure weakly synchronizing in $T \times \{0\}$.
By the attractor strategy, the probability mass that reaches $K_i$ is bounded
(at least $\eta^n$) regardless of the strategy of player~$1$, and    
thus there is also a lower bound on the probability mass outside $T \times \{0\}$
from some point on. 

For states in $X:= S_1 = \bigcup_j X_j$, we present 
an almost-sure winning strategy $\straa_{\as}$ for player~$1$
that successively plays according to the strategies $\straa_{1}$, $\straa_{2}, \dots$
where $\straa_{N}$ is defined as follows, for all $N \geq 1$.
The strategy $\straa_{N}$ plays in two phases: 
in the first phase, whenever a state in $R = \bigcup_j R_j \times \{k_j \mod p\}$
is reached, it plays to reach again $R$ after $p$ rounds (which is possible
since $p = r_j$); if no state in $R$ was reached, 
it plays according to the (memoryless) attractor strategy
in $X_j \setminus W_j$, and according to the (memoryless) almost-sure winning 
strategy to reach $R_j \times \{k_j \mod p\}$ in $W_j$.  
This first phase is played for $N$ rounds.
Let $c$ be the value of the tracking counter at the end of the first 
phase. The second phase is played for $N'$ rounds such that $N' > 2^n$ 
and $c+N' = 0 \mod p$, thus the tracking counter will be $0$ at the end
of the second phase. 
In the second phase, the strategy
$\straa_{N}$ plays like in the first phase, except
if a state in $R_j \times \{k_j \mod p\}$ is reached at step $N+N' - k_j$ (note 
that $N' > 2^n \geq k_j$),
where the strategy then plays according to the sure-winning strategy 
for eventually synchronizing in $U_j \times \{0\}$.

For $j=1,\dots,i_{x}$, consider the event 
\begin{linenomath*}
$$A_j = \{q_0 \, a_0b_0 \,  q_1 \dots \in (QAA)^{\omega} \mid \exists I \geq 0 \cdot \forall i \geq I: q_i \in W_j \}$$
\end{linenomath*}
where from some point on the play remains in the set $W_j$. 
We show that from every state in $X$, the strategy $\straa_{\as}$ is almost-sure 
winning for weakly synchronizing in $T \times \{0\}$. The argument has two parts:

\begin{itemize}
\item \emph{(Correctness under event $A_j$).}
For an arbitrary $\epsilon > 0$
let $N_{\epsilon}$ given by Lemma~\ref{lem:almost-sure-reach-game-pl1}.
Under the event $A_j$, 
since the strategy $\straa_{N_{\epsilon}}$
plays in $W_j$ according to an almost-sure winning strategy for the reachability objective
$\Diamond (R_j \times \{k_j \mod p\})$ for $N_{\epsilon}$ steps,
it follows that for all states $\tuple{q,t} \in X$
and strategies $\strab$ of player~$2$ in $\H$,
\begin{linenomath*}
$$\Prb_{\tuple{q,t}}^{\straa_{N_{\epsilon}},\strab}(\Diamond^{\leq N_{\epsilon}} (R_j \times \{k_j\!\!\mod p\}) \mid A_j) \geq 1-\epsilon,$$
\end{linenomath*}
and since $U_j \times \{0\} \subseteq W_j$ we can repeat the argument
after the second phase of the strategy $\straa_{N_{\epsilon}}$,
and get
$\Prb_{\tuple{q,t}}^{\straa_{\as},\strab}(\Diamond (R_j \times \{k_j \mod p\}) \mid A_j) = 1$
and $\Prb_{\tuple{q,t}}^{\straa_{\as},\strab}(\Box \Diamond R_j \times \{k_j \mod p\} \mid A_j) = 1$,
which implies that $\straa_{\as}$ is almost-sure 
winning for weakly synchronizing in $T \times \{0\}$ under event $A_j$.

\item \emph{(The event $\bigcup_{j} A_j$ has probability~$1$).}
Now we show that the strategy $\straa_{\as}$ is almost-sure winning for 
the event $\bigcup_{1 \leq j \leq i_{x}} A_j$, that is with probability~$1$ the play
will remain forever in some $W_j$. 
Intuitively, this is because if a play visits infinitely often the positive 
attractor of $W_1$ (namely, $X_1$), then the set $W_1$ is reached with probability $1$
and never left since it is a trap for player~$2$ in the subgame $\H[S_1]$ (where $S_1$ is also
a trap for player~$2$); on the other hand, if a play eventually
remains outside $X_1$, then from some point on the play remains always in the subgame 
$\H[S_2]$ (recall that $S_2 = S_1 \setminus X_1$) and then visiting $X_2$
infinitely often implies reaching and staying forever in $W_2$ with probability $1$.
Repeating this argument $i_{x}$ times shows that in all cases, the play has to remain 
forever in some $R_j$ with probability~$1$. Formally, fix an arbitrary state $\tuple{q,t} \in S_1$,
and a strategy $\strab$ of player~$2$ in $\H$, and we show that 
$\Prb_{\tuple{q,t}}^{\straa_{as},\strab}(\bigcup_{1 \leq j \leq i_{x}} A_j) = 1$. 
Let $B_1 = \{ q_0 \, a_0b_0 \,  q_1 \cdots \in (QAA)^\omega \mid \exists^\infty i \geq 0: q_i \in X_1 \}$
and for $j=2,\dots, i_{x}$, let
\begin{linenomath*}
$$B_j = \{ q_0 \, a_0b_0 \, q_1 \cdots \in (QAA)^\omega \mid  
\exists^\infty i \geq 0: q_i \in X_i \} \setminus \bigcup_{l < i} B_l$$
\end{linenomath*}
\noindent be the event that $X_i$ is visited infinitely often, and the states
  in $X_1 \cup \dots \cup X_{i-1}$ are visited only finitely often. 

  Under event $B_1$, the positive attractor $X_1$ of $W_1$ is visited infinitely often,
  and therefore the set $W_1$ is reached with probability~$1$ (by an argument similar
  to the proof of Lemma~\ref{lem:almost-sure-reach-game-pl1}, 
  under the positive-attractor strategy, there is a bounded
  probability $\eta_1 > 0$ to reach $W_1$ within a fixed number of steps, which entails that
  the probability to never reach $W_1$ is $\lim_{k \to \infty} (1-\eta_1)^k = 0$).
  Moreover, once the play is in $W_1$, it remains there forever (by definition 
  of the strategy $\sigma_{\as}$,
  and because $W_1$ is a trap for player~$2$ in $\H$).
  Thus, we have $\Prb_{\tuple{q,t}}^{\straa_{\as},\strab}(A_1 \mid B_1) = 1$ 
  (if $\Prb_{\tuple{q,t}}^{\straa_{\as},\strab}(B_1) \neq 0$).
  By a similar argument for $i=2,\dots,k$, under event $B_i$ the play eventually
  remains in the subgame $\H[S_i]$ since $S_i = S_1 \setminus \bigcup_{l < i} X_l$,
  and it follows that $\Prb_{\tuple{q,t}}^{\straa_{\as},\strab}(A_i \mid B_i) = 1$ 
  (if $\Prb_{\tuple{q,t}}^{\straa_{\as},\strab}(B_i) \neq 0$).
  Finally, since $\{X_i\}_{1 \leq i \leq k}$ is a partition of $S_1$ we have
  $\Prb_{\tuple{q,t}}^{\straa_{\as},\strab}(\bigcup_{j} B_j) = 1$, 
  and thus $\Prb_{\tuple{q,t}}^{\straa_{\as},\strab}(\bigcup_{1 \leq j \leq i_{x}} A_j) = 
  \Prb_{\tuple{q,t}}^{\straa_{\as},\strab}(\bigcup_{1 \leq j \leq i_{x}} A_j \mid 
  \bigcup_{1 \leq j \leq i_{x}} B_j) = 1$.
\end{itemize}

In the game $\G$, if $1_{q_1}$ and $1_{q_2}$ are almost-sure winning
for weakly synchronizing in $T$, it does not necessarily imply that the 
distribution with support $\{q_1,q_2\}$ is almost-sure winning
for weakly synchronizing in $T$. However, we show that in the game $\H$,
if the counter value is the same in two almost-sure winning states
$\tuple{q_1,t_1}$ and $\tuple{q_2,t_2}$ (i.e., $t_1 = t_2 = t$),
then the distribution with support $\{\tuple{q_1,t},\tuple{q_2,t}\}$ 
is also almost-sure winning. This is because under the strategy $\straa_{as}$,
the end of the second phase of the strategies $\straa_{N}$
occurs at the same time from all states with the same counter value,  
Therefore, from all states in $X$ with a given counter value, 
and for all strategies $\strab$ of player~$2$, the probability mass is 
at least $1-\epsilon$ in $T \times \{0\}$ at the last round played
by each $\straa_{N}$ for all $N$ sufficiently large.

Given a distribution $d_0 \in \dist(Q)$ over the states of $\G$
if $\Supp(d_0)$ is a slice in $X$,
that is there exists a counter value $t$ such that 
$\{\tuple{q,t} \mid q \in \Supp(d_0)\} \subseteq X$, then 
$d_0$ is an almost-sure winning distribution in $\G$.
On the other hand, if for every $0 \leq t < p$, 
there is a state $q \in \Supp(d_0)$ such that $\tuple{q,t} \in L$,
then from $d_0$, player~$2$ can spoil all strategies of player~$1$,
using a superposition of spoiling strategies for each $\tuple{q,t} \in L$.

In conclusion, the sets $s \subseteq Q$ such that $s \times \{t\} \subseteq X = S_1$
for some $t$ are the supports of the almost-sure winning distributions
(line~\ref{alg:solve-as-ws-stoch-return}),
which establishes the correctness of Algorithm~\ref{alg:solve-as-ws-stoch}.

\begin{lemma}\label{lem:ssgvalalg}
  Given a stochastic game and a set $T$ of target states,
  Algorithm~\ref{alg:solve-as-ws-stoch} computes the supports of the distributions
  from which player~$1$ is almost-sure winning for weakly synchronizing
  in $T$. This algorithm can be implemented in PSPACE.
\end{lemma}

\begin{proof}
  The correctness of the algorithm follows from the arguments given above. 
  We show that this algorithm can be implemented in PSPACE by a similar
  proof as for deterministic games (Lemma~\ref{lem:alg-det-PSPACE}).
  We already showed that the number of executions of the main loop 
  (line~\ref{alg:solve-as-ws-stoch-loop}) is at most $n^2$ (see \emph{Termination}),
  and since the period of the set $U$ (line~\ref{alg:solve-as-ws-stoch-main-test})
  is at most $2^n$ times the period of the game $\H$ (slices are subsets of $Q$),
  the final value of $r$ (upon termination of the algorithm) is at most 
  $(2^n)^{n^2} = 2^{n^3}$.

  A PSPACE implementation of Algorithm~\ref{alg:solve-as-ws-stoch} can be obtained
  by using a PSPACE procedure to determine the transitions of $\H\upharpoonright [K]$
  and $\H\upharpoonright [S]$, as in the proof of Lemma~\ref{lem:alg-det-PSPACE}.
  The computation of the positive attractors, and almost-sure winning regions
  can then be done in PSPACE as well.
\end{proof}

We obtain the following theorem, where the PSPACE upper bound is given by 
Lemma~\ref{lem:ssgvalalg}, and the lower bound and memory requirement hold 
in the special case of MDPs~\cite[Theorem~6]{DMS19}.

\begin{theorem}\label{theo:as-weakly-winning}
The membership problem for almost-sure weakly synchronizing in stochastic games 
is PSPACE-complete, and pure counting strategies are sufficient for player~$1$. 
Infinite memory is necessary in general.
\end{theorem}

\subsection{Other synchronizing objectives}\label{sec:other}

The almost-sure winning region for the other synchronizing
objectives can be computed relatively easily.

\begin{lemma}\label{lem:always-sure-as}
For always synchronizing, the sure and almost-sure winning modes coincide:
$\winsure{always}(T) = \winas{always}(T)$.
\end{lemma}

\begin{proof}
The inclusion $\winsure{always}(T) \subseteq \winas{always}(T)$ follows 
from the definitions (Section~\ref{sec:def}). For the converse inclusion, consider an initial
distribution $d_0$ from which player~$1$ has an almost-sure winning
strategy for always synchronizing in $T$. 
We claim that player~$1$ has a strategy $\straa_{{\rm safe}}$ to ensure,
for all plays $\rho = q_0 \, a_0b_0 \, q_1 \ldots q_k$ compatible with $\straa$,
that $q_i \in T$ for all $0 \leq i \leq k$ (i.e., player~$1$ is sure-winning
for the safety objective $\Box T$~\cite{ConcOmRegGames}).
By contradiction, if that is not the case, then player~$2$ has a strategy
to ensure reaching a state in $Q \setminus T$ within at most $n$ steps
with positive probability (at least $\eta_0 \cdot \eta^n$) against all
strategies of player~$1$, in contradiction with player~$1$ being 
almost-sure winning for always synchronizing in $T$. 
Hence such a strategy $\straa_{{\rm safe}}$ exists and we conclude the proof 
by observing that $\straa_{{\rm safe}}$ is sure winning for always synchronizing in $T$.
\end{proof}

The following lemma generalizes to games a result that holds 
for MDPs~\cite[Section~5.1.2]{Shi14}.
\begin{lemma}\cite{V18}\label{lem:event-sync}
In stochastic games, we have 
$\winas{event}(T) = \winsure{event}(T) \cup \winas{weakly}(T)$.
\end{lemma}

\begin{proof}
The inclusions $\winsure{event}(T) \subseteq \winas{event}(T)$ 
and $\winas{weakly}(T) \subseteq \winas{event}(T)$ follow 
from the definitions (Section~\ref{sec:def}).
For the converse, consider an initial distribution $d_0$ 
from which player~$1$ has an almost-sure winning
strategy for eventually synchronizing in $T$. 
Towards contradiction, assume that $d_0 \not\in \winsure{event}(T) \cup \winas{weakly}(T)$.
Then for all player-$1$ strategies $\straa$, 
there exist player-$2$ strategies $\strab_{\rm e}$ and $\strab_{\rm w}$
such that $\G^{\straa,\strab_{\rm e}}_{d_0}$ is not sure eventually synchronizing in $T$,
and $\G^{\straa,\strab_{\rm w}}_{d_0}$ is not almost-sure weakly synchronizing in $T$.
Let $\strab$ be the strategy playing the superposition of $\frac{1}{2}\strab_{\rm e}$ 
and $\frac{1}{2}\strab_{\rm w}$. Consider $\strab_{\rm w}$ and following the definitions, 
there exists $\epsilon_{w} > 0$ and $i^* \geq 0$ such that $\G^{\straa,\strab}_i(T) < 1-\epsilon_w$
for all $i \geq i^*$. Moreover, considering $\strab_{\rm e}$, 
we have $\G^{\straa,\strab}_i(T) < 1- \frac{\eta_0}{2} \cdot \eta^{i^*}$
for all $i \leq i^*$. For $\epsilon = \min\{\epsilon_w, \frac{\eta_0}{2} \cdot \eta^{i^*}\} > 0$,
we get $\G^{\straa,\strab}_i(T) < 1- \epsilon$ for all $i \geq 0$, in contradiction
to player~$1$ being almost-sure winning for eventually synchronizing in~$T$.
\end{proof}

The reduction presented in the proof of Theorem~\ref{theo:alg-det-PSPACE-complete}
also shows PSPACE-hardness for almost-sure eventually synchronizing 
in deterministic games.

\begin{lemma}\label{lem:strong-sync}
Let $\G$ be a stochastic game.
Given a target set~$T$, an initial distribution $d_0$ almost-sure winning 
for strongly synchronizing in $T$
if and only if $d_0$ is almost-sure winning 
for the coB\"uchi objective~$\Diamond \Box T$.
\end{lemma}

\begin{proof}
First, if player~$1$ is almost-sure winning 
for the coB\"uchi objective~$\Diamond \Box T$ in $\G$, 
then there exists a memoryless winning strategy $\straa_{\as}$
for player~$1$, which is such that all states in the (reachable) end-components 
of the MDP obtained from the game $\G$ after fixing the strategy $\straa_{\as}$
are in $T$~\cite{CY95}.
It is then easy to show that the strategy $\straa_{\as}$
is almost-sure winning for strongly synchronizing in~$T$. 

For the converse direction, if player~$1$ is not almost-sure winning
for the coB\"uchi objective~$\Diamond \Box T$ in $\G$, 
then there exists a strategy $\strab$ for player~$2$, 
which we can assume to be memoryless~\cite{CJH04},
such that for all player-$1$ strategies $\straa$ we have
$\Prb_{d_0}^{\straa,\strab}(\Diamond \Box T) < 1$. 
Hence in the MDP $\G_{\strab}$ obtained from $\G$ by fixing the strategy $\strab$,
player~$1$ is not almost-sure winning for the coB\"uchi 
objective~$\Diamond \Box T$, which is equivalent 
to say that in $\G_{\strab}$ no player-$1$ strategy is
almost-sure winning for strongly synchronizing in $T$~\cite[Lemma~27]{DMS19},
and concludes the proof.
\end{proof}

We note that in deterministic games, the sure and almost-sure winning modes
coincide for state-based objectives, thus it follows from Lemma~\ref{lem:strong-sync}
that for strongly synchronizing the sure and almost-sure winning modes
coincide in deterministic games.

We summarize the results of Section~\ref{sec:as} for almost-sure synchronizing.

\begin{theorem}\label{theo:as-winning}
The membership problem for almost-sure always and strongly synchronizing can be solved
in polynomial time, and pure memoryless strategies are sufficient for player~$1$.

The membership problem for almost-sure eventually and weakly synchronizing is PSPACE-complete, 
and pure counting strategies are sufficient for player~$1$. Infinite memory is necessary
in general.
\end{theorem}

\section{Conclusion}

Stochastic games with synchronizing objectives combine stochasticity with
the presence of an adversary and a flavour of imperfect information,
which together tend to bring undecidability in a continuous 
setting~\cite{PAZBook,MHC03}. The form of imperfect information in these games
differs from the traditional setting where the strategy of player~$1$
is uniform (the same action is played in all states)~\cite{BGB12,CFO20}.
Here, player~$1$ can see the local state of the game, but needs to enforce 
a global objective defined on state distributions, which are not visible 
to player~$1$. 
Beyond decidability, it is perhaps surprising that the membership 
problem for games is no harder than for MDPs (PSPACE-complete), although
the proof techniques are significantly more involved, mainly due
to the presence of an adversary, and the lack of determinacy. 

The main question raised by this model is whether it is possible to extend it
with a form of communication, while remaining decidable. This would bring
us closer to a wide range of applications in synthetic biology~\cite{NDS16,VLN19} and chemical 
reaction networks~\cite{CKL18}.
In another direction, considering other classes of objectives, such as combination
of multiple synchronizing objectives, or more quantitative conditions (where the
probability threshold is not $1$) are completely open problems.

\paragraph{{\bf Acknowledgment}}
We are grateful to Mahsa Shirmohammadi and Marie van den Bogaard for 
preliminary discussions about this problem 
and for inspiring the example of \figurename~\ref{fig:gamelose}.

\bibliographystyle{alpha} 
\bibliography{biblio} 
\end{document}